\documentclass[10pt,draftclsnofoot,onecolumn,letterpaper]{IEEEtran}
\usepackage{amsmath,amssymb,mathrsfs,amsthm}
\usepackage{graphicx}
\usepackage{psfrag}
\usepackage{cite}
\usepackage{multirow}
\usepackage{color}
\usepackage{subcaption}
\usepackage{suffix}
\usepackage{array}
\usepackage{xspace}
\newcolumntype{x}[1]{>{\centering\arraybackslash\hspace{0pt}}m{#1}}

\usepackage{calc}
\usepackage{graphicx}
        \graphicspath{{figs/}{matlab/}{../figs/}} 
        
\usepackage{xspace}
\usepackage{bbm}
\input{mathlig}

\newcommand{\muspace}{\mspace{1mu}}

\DeclareRobustCommand{\scond}{\mathchoice{\muspace\vert\muspace}{\vert}{\vert}{\vert}}
\mathlig{|}{\scond}
\newcommand{\cond}{\mathchoice{\,\vert\,}{\mspace{2mu}\vert\mspace{2mu}}{\vert}{\vert}}

\DeclareRobustCommand{\discint}{\mathchoice{\mspace{-1.5mu}:\mspace{-1.5mu}}{\mspace{-1.5mu}:\mspace{-1.5mu}}{:}{:}}
\mathlig{::}{\discint}
\newcommand{\suchthat}{\mathchoice{\colon}{\colon}{:\mspace{1mu}}{:}}

\newcommand{\Cc}{\mathcal{C}}

\newcommand{\Ec}{\mathcal{E}}

\newcommand{\Uc}{\mspace{1.5mu}\mathcal{U}}

\newcommand{\Wc}{\mathcal{W}}
\newcommand{\Xc}{\mathcal{X}}

\newcommand{\Rr}{\mathscr{R}}

\newcommand{\rv}{{\bf r}}

\newcommand{\Iv}{{\bf I}}

\newcommand{\pen}{{P_e^{(n)}}}

\newcommand{\aep}{{\mathcal{T}_{\epsilon}^{(n)}}}

\newcommand{\Mh}{{\hat{M}}}

\newcommand{\mh}{{\hat{m}}}

\newcommand{\Nt}{{\tilde{N}}}

\newcommand{\Rt}{{\tilde{R}}}

\newcommand{\Yt}{{\tilde{Y}}}
\newcommand{\Zt}{{\tilde{Z}}}
\newcommand{\dt}{{\tilde{d}}}

\def\a{\alpha}
\def\b{\beta}
\def\g{\gamma}
\def\d{\delta}
\def\e{\epsilon}

\def\l{\lambda}

\DeclareMathOperator\E{\textsf{E}}
\let\P\relax
\DeclareMathOperator\P{\textsf{P}}

\DeclareMathOperator\C{\textsf{C}}

\def\error{\mathrm{e}}

\newcommand{\Bern}{\mathrm{Bern}}

\newcommand{\N}{\mathrm{N}}
\newcommand{\U}{\mathrm{Unif}}

\def\textiid{i.i.d.\@\xspace}
\newcommand\iid{\ifmmode\text{ i.i.d. } \else \textiid \fi}

\def\mathllap{\mathpalette\mathllapinternal}
\def\mathllapinternal#1#2{%
  \llap{$\mathsurround=0pt#1{#2}$}}

\def\mathrlap{\mathpalette\mathrlapinternal}
\def\mathrlapinternal#1#2{%
  \rlap{$\mathsurround=0pt#1{#2}$}}

\def\clap#1{\hbox to 0pt{\hss#1\hss}}
\def\mathclap{\mathpalette\mathclapinternal}
\def\mathclapinternal#1#2{%
  \clap{$\mathsurround=0pt#1{#2}$}}

\let\oldstackrel\stackrel
\renewcommand{\stackrel}[2]{\oldstackrel{\mathclap{#1}}{#2}}

\renewcommand{\hbar}{h\mathllap{\overline{\vphantom{h}\hphantom{\rule{4.6pt}{0pt}}}\mspace{0.77mu}}}

\catcode`~=11 
\newcommand{\urltilde}{\kern -.06em\lower -.06em\hbox{~}\kern .02em}
\catcode`~=13 

\theoremstyle{plain}
\newtheorem{theorem}{Theorem}
\newtheorem{lemma}{Lemma}
\newtheorem{proposition}{Proposition}

\theoremstyle{definition}

\theoremstyle{remark}
\newtheorem{remark}{Remark}

\let\j\relax
\newcommand{\j}{j}

\newcommand{\RIAN}[1]{\Rr_{{#1},\mathrm{IAN}}}
\WithSuffix\newcommand\RIAN*{\Rr_{\mathrm{IAN}}}
\newcommand{\RSCD}[1]{\Rr_{{#1},\mathrm{SCD}}}
\WithSuffix\newcommand\RSCD*{\Rr_{\mathrm{SCD}}}

\title{Sliding-Window Superposition Coding:\\Two-User Interference Channels}

\author{Lele Wang, Young-Han Kim, Chiao-Yi Chen, Hosung Park, and Eren \c{S}a\c{s}o\u{g}lu

\thanks{The material in this paper was presented in part in the IEEE International Symposium on Information Theory (ISIT) 2014, Honolulu, HI, and in part in the IEEE Globecom Workshops (GC Wkshps) 2014, Austin, TX.} 

\thanks{L. Wang is jointly with the Department of Electrical Engineering, Stanford University, Stanford, CA 94305 USA and the Department of Electrical Engineering - Systems, Tel Aviv University, Tel Aviv, Israel (email: wanglele@stanford.edu).}

\thanks{Y.-H. Kim is with the Department of Electrical and Computer Engineering, University of California, San Diego, La Jolla, CA 92093 USA (e-mail: yhk@ucsd.edu).}

\thanks{C.-Y. Chen is with Broadcom Limited, 190 Mathilda Place, Sunnyvale, CA 94086 USA (email: uscychen@gmail.com).}

\thanks{H. Park is with the School of Electronics and Computer Engineering, Chonnam National University, Gwangju 61186, Korea (e-mail: hpark1@jnu.ac.kr).}

\thanks{E. \c{S}a\c{s}o\u{g}lu is  with  Intel  Corporation, Santa Clara, CA 95054 USA (e-mail: eren.sasoglu@gmail.com).}
}

\begin{document}
\maketitle

\begin{abstract}
 A low-complexity coding scheme is developed to achieve the rate region of maximum likelihood decoding for interference channels. As in the classical rate-splitting multiple access scheme by Grant, Rimoldi, Urbanke, and Whiting, the proposed coding scheme uses superposition of multiple codewords with successive cancellation decoding, which can be implemented using standard point-to-point encoders and decoders. Unlike rate-splitting multiple access, which is not rate-optimal for multiple receivers, the proposed coding scheme transmits codewords over multiple blocks in a staggered manner and recovers them successively over sliding decoding windows, achieving the single-stream optimal rate region as well as the more general Han--Kobayashi inner bound for the two-user interference channel. The feasibility of this scheme in practice is verified by implementing it using commercial channel codes over the two-user Gaussian interference channel.
\end{abstract}

\section{Introduction}

For high data rates and massive connectivity, next-generation cellular networks are expected to deploy many small base stations. While such dense deployment provides the benefit of bringing radio closer to end users, it also increases the amount of interference from neighboring cells. Consequently, efficient and effective management of interference is expected to become one of the main challenges for high-spectral-efficiency, low-power, broad-coverage wireless communications.

Over the past few decades, several techniques at different protocol layers \cite{Boudreau--Panicker--Guo--Chang--Wang--Vrzic2009, Seol--Cheun2009, Cadambe--Jafar2008} have been proposed to mitigate adverse effects of interference in wireless networks. One important conceptual technique at the physical layer is \emph{simultaneous decoding} \cite[Section~6.2]{El-Gamal--Kim2011},~\cite{Bidokhti--Prabhakaran2014}. In this decoding method, each receiver attempts to recover both the intended and a subset of the interfering codewords at the same time. When the interference is strong \cite{Costa--El-Gamal1987, Sato1978b} and weak~\cite{Shang--Kramer--Chen2009,Annapureddy--Veeravalli2009,Motahari--Khandani2011,Liu--Nair--Xia2014}, simultaneous decoding of random code ensembles achieves the capacity of the two-user interference channel. In fact, for any given random code ensemble, simultaneous decoding achieves the same rates achievable by the optimal maximum likelihood decoding \cite{Motahari--Khandani2011, Baccelli--El-Gamal--Tse2011, Bandemer--El-Gamal--Kim2015}. The celebrated Han--Kobayashi coding scheme \cite{Han--Kobayashi1981} also relies on simultaneous decoding as a crucial component. As a main drawback, however, each receiver in simultaneous decoding (or maximum likelihood decoding) has to employ some form of multiuser sequence detection, which usually has high computational complexity. This issue has been tackled recently by a few approaches based on emerging spatially coupled and polar codes~\cite{Yedla--Nguyen--Pfister--Narayanan2011, Wang--Sasoglu2014}, but these solutions involve very long block lengths.

For this reason, most practical communication systems use conventional point-to-point low-complexity decoding. The simplest method is \emph{treating interference as noise}, in which only statistical properties (such as the distribution and power), rather than the actual codebook information, of the interfering signals, are used. In \emph{successive cancellation decoding}, similar low-complexity point-to-point decoding is performed in steps, first recovering interfering codewords and then incorporating them as part of the channel output for decoding of desired codewords. Successive cancellation decoding is particularly well suited when the messages are split into multiple parts by rate splitting, encoded into separate codewords, and transmitted via superposition coding. In particular, when there is only one receiver (i.e., for a multiple access channel), this rate-splitting coding scheme with successive cancellation decoding was proposed by Rimoldi and Urbanke~\cite{Rimoldi--Urbanke1996} for the Gaussian case and Grant, Rimoldi, Urbanke, and Whiting~\cite{Grant--Rimoldi--Urbanke--Whiting2001} for the discrete case, and achieves the optimal rate region of the polymatroidal shape (the pentagon for two senders). When there are two or more receivers---as in the two-user interference channel or the compound multiple access channel---the rate-splitting multiple access scheme fails to achieve the optimal rate region as demonstrated earlier in~\cite{Zhao--Tan--Avestimehr--Diggavi--Pottie2012} for Gaussian codes and in Section~\ref{sec:rs-subopt} of this paper (and~\cite{Wang--Sasoglu--Kim2014}) for general codes.

A natural question is whether low-complexity point-to-point coding techniques, which could achieve capacity for multiple access and single-antenna Gaussian broadcast channels, are fundamentally deficient for the interference channel, and high-complexity simultaneous decoding would be critical to achieve the capacity in general. In this paper, we develop a new coding scheme, called \emph{sliding-window superposition coding}, that overcomes the limitations of low-complexity decoding through a new diagonal superposition structure. The main ingredients of the scheme are block Markov coding, sliding-window decoding (both commonly used for multihop relaying and feedback communication), superposition coding, and successive cancellation decoding (crucial for low-complexity implementation using standard point-to-point codes). Each message is encoded into a single long codeword that are transmitted diagonally over multiple blocks and multiple signal layers, which helps avoid the performance bottleneck for the aforementioned rate-splitting multiple access scheme. Receivers recover the desired and interfering codewords over a decoding window spanning multiple blocks. Successive cancellation decoding is performed within each decoding window as well as across a sequence of decoding windows for streams of messages. When the number and distribution of signal layers are properly chosen, the sliding-window superposition coding scheme can achieve every rate pair in the rate region of maximum likelihood decoding for the two-user interference channel with single streams, providing a constructive answer to our earlier question. We develop a more complete theory behind the number and distribution of signal layers and the choice of decoding orders, which leads to an extension of this coding scheme that achieves the entire Han--Kobayashi inner bound.

For practical communication systems, the conceptual sliding-window superposition coding scheme can be readily adapted to a coded modulation scheme using binary codes and common signal constellations. We compare this \emph{sliding-window coded modulation scheme} with two well-known coded modulation schemes, multi-level coding~\cite{Imai--Hirakawa1977,Wachsmann--Fischer--Huber1999} and bit-interleaved coded modulation~\cite{Zehavi1992,Caire--Taricco--Biglieri1998}. We implement the sliding-window coded modulation scheme for the two-user Gaussian channel using the 4G LTE turbo code and demonstrate its performance improvement over treating interference as noise. Following earlier conference versions~\cite{Wang--Sasoglu--Kim2014,Park--Kim--Wang2014} of this paper, several practical implementations of sliding-window superposition coding have been investigated~\cite{Kim--Ahn--Kim--Park--Wang--Chen--Park2015, Kim--Ahn--Kim--Park--Chen--Kim2016} and proposed to the 5G standards~\cite{5Gdiscussion1,5Gdiscussion3,5Gdiscussion4,5Gdiscussion5,5Gdiscussion6,5Gdiscussion7}.

The rest of the paper is organized as follows. We first define the problem and the relevant rate regions in Section~\ref{sec:setup}. Then, we explain the rate-splitting scheme and demonstrate its fundamental deficiency for the interference channel in Section~\ref{sec:rs}. We introduce the new sliding-window superposition coding in Section~\ref{sec:swsc-sd}, first by developing a simple scheme that achieves the corner points of simultaneous decoding region, and then extending it to achieve every point in the region. We devote Section~\ref{sec:swcm} to sliding-window coded modulation and its application in a practical communication setting. In Section~\ref{sec:layering-order}, we present a more complete theory of the sliding-window superposition coding scheme with a discussion on the number of superposition layers and alternative decoding orders.  With further extensions and augmentations, we develop a scheme that achieves the Han--Kobayashi inner bound~\cite{Han--Kobayashi1981} for the two-user interference channel with point-to-point encoders and decoders in Section~\ref{sec:hk}. We offer a couple of concluding remarks in Section~\ref{sec:rmk}.

Throughout the paper, we closely follow the notation in~\cite{El-Gamal--Kim2011}. In particular, for $X \sim p(x)$ and $\e \in (0,1)$, we define the set of $\e$-typical $n$-sequences $x^n$ (or the typical set in short)~\cite{Orlitsky--Roche2001} as
\[
\aep(X) = \bigl\{ x^n : | \#\{i \suchthat x_i = x\}/n - p(x) | \le \e p(x) \text{ for all } x \in \Xc \bigr\}.
\]
We use $X_k^n$ to denote the vector $(X_{k1},X_{k2},\ldots,X_{kn})$. For $n = 1,2,\ldots, [n] = \{1,2,\ldots,n\}$ and for $a \ge 0, [2^a]=\{1,2,\ldots,2^{\left\lceil a \right\rceil}\}$, where $\left\lceil a \right\rceil$ is the smallest integer greater than or equal to $a$. The probability of an event $\mathcal{A}$ is denoted by $\P(\mathcal{A})$.


\section{Two-User Interference Channels}
\label{sec:setup}
Consider the communication system model depicted in Fig.~\ref{fig:2-user-ic}, whereby
senders~1 and~2 wish to communicate independent messages $M_1$ and $M_2$ to their respective
receivers over a shared channel $p(y_1,y_2|x,w)$. Here $X$ and $W$ are channel inputs from senders~1 and~2, respectively,
and $Y_1$ and $Y_2$ are channel outputs at receivers~1 and~2, respectively. In network information theory, this model is commonly referred to
as the \emph{two-user interference channel}. 

The Gaussian interference channel in Fig.~\ref{fig:gaussian-sim}
is an important special case with
channel outputs
\begin{equation} \label{eq:channel}
\begin{split}
    Y_1 = g_{11}X + g_{12}W + N_1, \\
    Y_2 = g_{21}X + g_{22}W + N_2,
\end{split}
\end{equation}
where $g_{jk}$ denotes the channel gain coefficient from sender $k$ to receiver $j$,
and $N_1$ and $N_2$ are independent $\N(0,1)$ noise components.
Under the average power constraint $P$ on each input $X$ and $W$,
we denote the received signal-to-noise ratios (SNRs) as $S_1 = g_{11}^2 P$ and $S_2 = g_{22}^2 P$, and the received interference-to-noise ratios (INRs) as $I_1 = g_{12}^2 P$ and $I_2 = g_{21}^2 P$.

\begin{figure}[t]
\centering
\psfrag{m1}[b]{$M_1$}
\psfrag{m2}[b]{$M_2$}
\psfrag{x1}[b]{$X^n$}
\psfrag{x2}[b]{$W^n$}
\psfrag{y1}[b]{$Y_1^n$}
\psfrag{y2}[b]{$Y_2^n$}
\psfrag{mh1}[b]{$\Mh_1$}
\psfrag{mh2}[b]{$\Mh_2$}
\psfrag{e1}[c]{Encoder $1$}
\psfrag{e2}[c]{Encoder $2$}
\psfrag{d1}[c]{Decoder $1$}
\psfrag{d2}[c]{Decoder $2$}
\psfrag{p}[c]{$p(y_1,y_2|x,w)$}
\includegraphics[scale=0.6]{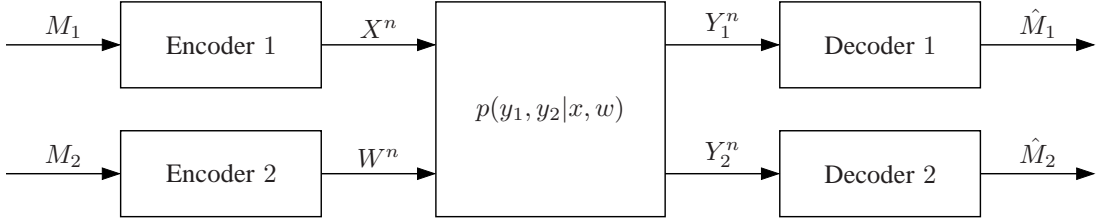}
\caption{The interference channel with two sender--receiver pairs.}
\label{fig:2-user-ic}
\end{figure}

\begin{figure}[t]
\centering
\psfrag{z1}[bl]{\hspace{-.6em}$N_1$}
\psfrag{z2}[tl]{\hspace{-.6em}$N_2$}
\psfrag{x1}[r]{$X$}
\psfrag{x2}[r]{$W$}
\psfrag{y1}[l]{$Y_1$}
\psfrag{y2}[l]{$Y_2$}
\psfrag{a}[t]{$g_{12}$}
\psfrag{b}[b]{$g_{21}$}
\psfrag{c}[t]{$g_{22}$}
\psfrag{d}[b]{$g_{11}$}
\includegraphics[width= .3\linewidth]{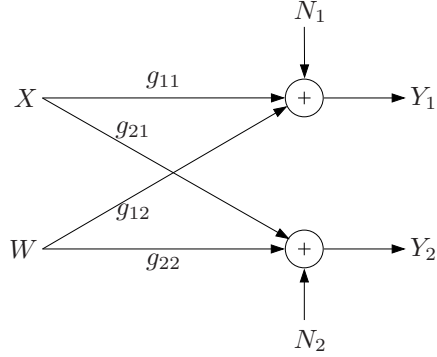}
\medskip
\caption{The two-user Gaussian interference channel.}
\label{fig:gaussian-sim}
\end{figure}

A $(2^{nR_1},2^{nR_2},n)$ code $\Cc_n$ for the (two-user) interference channel consists of
\begin{itemize}
\item two message sets $[2^{nR_1}] := \{1, \ldots, 2^{\lceil nR_1\rceil}\}$ and $[2^{nR_2}]$,

\item two encoders, where encoder~1 assigns a codeword $x^n(m_1)$ to each message 
$m_1 \in [2^{nR_1}]$ and encoder~2 assigns a codeword $w^n(m_2)$ to each message $m_2 \in [2^{nR_2}]$, and

\item two decoders, where decoder~1 assigns an estimate $\mh_1$ or an error message $\error$ to each received sequence $y_1^n$ and decoder~2 assigns an estimate $\mh_2$ or an error message $\error$ to each received sequence $y^n_2$.
\end{itemize}
The performance of a given code $\Cc_n$ for the interference channel is measured by its average probability of error
\[
\pen(\Cc_n) = \P\bigl\{ (\Mh_1, \Mh_2) \ne (M_1,M_2) \bigr\},
\]
where the message pair $(M_1,M_2)$ is uniformly distributed over $[2^{nR_1}]\times [2^{nR_2}]$. 
A rate pair $(R_1,R_2)$ is said to be \emph{achievable} if there exists a sequence of $(2^{nR_1}, 2^{nR_2}, n)$ codes $(\Cc_n)_{n=1}^\infty$ such that $\lim_{n
\to \infty} \pen(\Cc_n) = 0$. A set of rate pairs, typically referred to as a \emph{rate region}, is said to be achievable if every rate pair in the interior of the set is achievable. The \emph{capacity region} is the closure of the set of achievable rate pairs $(R_1, R_2)$, which is the largest achievable rate region and captures the optimal tradeoff between the two rates of reliable communication over the interference channel. The capacity region for the two-user interference channel is not known in general.

Let $p = p(x)p(w)$ be a given product pmf on $\Xc \times \Wc$. Suppose that the codewords $x^n(m_1), m_1 \in [2^{nR_1}]$, and $w^n(m_2), m_2 \in [2^{nR_2}]$, that constitute the codebook are generated randomly and independently according to $\prod_{i=1}^n p_{X}(x_{i})$ and $\prod_{i=1}^n p_{W}(w_{i})$, respectively. We refer to 
the codebooks generated in this manner collectively as the $(2^{nR_1},2^{nR_2},n;p)$ \emph{random code ensemble} (or the \emph{$p$-distributed random code ensemble} in short). 

Fixing the encoders as such, we now consider a few alternative decoding schemes. Here and henceforth, we assume $p = p(x)p(w)$ is fixed and write rate regions without $p$ whenever it is clear from the context.

\begin{itemize}
\item \emph{Treating interference as noise (IAN).} Receiver~1 recovers $M_1$ by treating the interfering codeword $W^n(M_2)$ as noise generated according to a given (memoryless)
distribution $p(w)$. In other words, receiver~1 performs 
point-to-point decoding (either a specific decoding technique or a conceptual scheme)
for the channel 
\begin{align*}
p(y_1^n|x^n) &= \sum_{w^n} p(w^n) p(y_1^n | x^n, w^n) \\
&= \prod_{i=1}^n \sum_{w_i} p_W(w_i)p_{Y_1|X,W}(y_{1i}|x_i,w_i) = \prod_{i=1}^n p_{Y_1|X}(y_{1i}|x_i).
\end{align*}
For example, if joint typicality decoding~\cite[Section 7.7]{Cover--Thomas2006} is used, 
the decoder finds $\mh_1$ such that $(x^n(\mh_1),y_1^n) \in \aep(X,Y_1)$. Similarly, receiver~2 can 
recover $M_2$ by treating $X^n$ as noise. 
For the $p$-distributed random code ensemble, 
treating noise as interference achieves 
\[
\RIAN* = \RIAN1 \cap \RIAN2
\]
where $\RIAN1$ and $\RIAN2$ denote the sets of all rate pairs $(R_1, R_2)$
such that $R_1 \le I(X;Y_1)$ and $R_2 \le I(W;Y_2)$, respectively; see Fig.~\ref{fig:ian}.

\item \emph{Successive cancellation decoding (SCD).} Receiver~1 recovers $M_2$ by treating $X^n$ as noise and then recovers $M_1$ based on $W^n(M_2)$ (and $Y_1^n$). For example, in joint typicality decoding, the decoder finds a unique $\mh_2$ such that $(w^n(\mh_2),y_1^n) \in \aep(W,Y_1)$ and then a unique $\mh_1$ such that $(x^n(\mh_1), w^n(\mh_2), y_1^n) \in \aep(X,W,Y_1)$. Receiver~2 operates in a similar manner. For the $p$-distributed random code ensemble, successive cancellation decoding achieves 
\[
\RSCD* = \RSCD1 \cap \RSCD2,
\]
where $\RSCD1$ consists of $(R_1, R_2)$ such that
\begin{align*}
 R_2 \le I(W;Y_1),\quad
 R_1 \le I(X;Y_1|W),
\end{align*}
and similarly $\RSCD2$ consists of $(R_1,R_2)$ such that
\begin{align*}
 R_1 \le I(X;Y_2),\quad
 R_2 \le I(W;Y_2|X).
\end{align*}
See Fig.~\ref{fig:sc} for an illustration of $\RSCD*$.

\item \emph{Mix and match.} Each receiver can choose 
between treating interference as noise and successive cancellation decoding. 
This mix-and-match achieves 
\begin{equation}
\label{eqn:ian-sc}
(\RIAN1 \cup \RSCD1) \cap (\RIAN2 \cup \RSCD2).
\end{equation}
The achievable rate region for mixing and matching is illustrated in Fig.~\ref{fig:comp}.

\item \emph{Simultaneous (nonunique) decoding (SND).} Receiver~1 recovers both the desired message $M_1$ and the interfering message $M_2$ simultaneously. It then keeps $M_1$ as the message estimate and ignores the error in estimating $M_2$. Receiver~2 operates in a similar manner. For example, in joint typicality decoding, receiver~1 finds a unique $\mh_1$ such that $(x^n(\mh_1),w^n(m_2),y_1^n) \in \aep(X,W,Y_1)$ for some $m_2 \in [2^{nR_2}]$, and receiver~2 finds a unique $\mh_2$ such that $(x^n(m_1),w^n(\mh_2),y_2^n) \in \aep(X,W,Y_2)$ for some $m_1 \in [2^{nR_1}]$. For the $p$-distributed random code ensemble, simultaneous decoding achieves
\[
\Rr_{\mathrm{SND}} = \Rr_{1,\mathrm{SND}} \cap \Rr_{2,\mathrm{SND}},
\]
where $\Rr_{1,\mathrm{SND}}$ consists of $(R_1,R_2)$ such that
\begin{equation}
\label{eq:mld-region1a}
~~R_1 \le I(X;Y_1)
\end{equation}
or
\begin{equation}
\label{eq:mld-region1b}
\begin{split}
 R_2 &\le I(W;Y_1|X),\\
 R_1 + R_2 &\le I(X,W;Y_1),
\end{split}
\end{equation}
and $\Rr_{2,\mathrm{SND}}$ is characterized by index substitution $1 \leftrightarrow 2$ and variable substitution $X \leftrightarrow W$ in \eqref{eq:mld-region1a} and~\eqref{eq:mld-region1b}, i.e., 
\[
~~R_2 \le I(W;Y_2)
\]
or 
\begin{align*}
 R_1 &\le I(X;Y_2|W),\\
 R_1 + R_2 &\le I(X,W;Y_2).
\end{align*}
Note that $\Rr_{\mathrm{SND}}$ can be written as
\begin{align}
 \Rr_{\mathrm{SND}} &= \left(\Rr_{1,\mathrm{IAN}} \cup \Rr_{1,\mathrm{SD}}\right) \cap \left(\Rr_{2,\mathrm{IAN}} \cup \Rr_{2,\mathrm{SD}}\right)\nonumber\\
 &= \left(\Rr_{1,\mathrm{IAN}} \cap \Rr_{2,\mathrm{IAN}}\right) \cup \left(\Rr_{1,\mathrm{SD}} \cap \Rr_{2,\mathrm{IAN}}\right) \cup \left(\Rr_{1,\mathrm{IAN}} \cap \Rr_{2,\mathrm{SD}}\right) \cup \left(\Rr_{1,\mathrm{SD}} \cap \Rr_{2,\mathrm{SD}}\right),\label{eqn:component}
\end{align}
where $\Rr_{1,\mathrm{SD}}$ is defined as the set of rate pairs $(R_1,R_2)$ such that
\begin{equation}
\label{eqn:sd}
\begin{split}
 R_1 &\le I(X;Y_1|W),\\
 R_2 &\le I(W;Y_1|X),\\
 R_1 + R_2 &\le I(X,W;Y_1), 
\end{split}
\end{equation}
and $\Rr_{2,\mathrm{SD}}$ is defined similarly by making the index substitution $1 \leftrightarrow 2$ and variable substitution $X \leftrightarrow W$ in $\Rr_{1,\mathrm{SD}}$.  
\end{itemize}

\begin{figure}[t!]
\centering
\begin{subfigure}[b]{0.48\textwidth}
\centering
\small
\psfrag{r1}[b]{$R_2$}
\psfrag{r2}[l]{$R_1$}
\psfrag{i1}[cc]{$I(X;Y_1)$}
\psfrag{i2}[r]{$I(W;Y_2)$}
\psfrag{rr1}[l]{$\RIAN1$}
\psfrag{rr2}[l]{$\RIAN2$}
\includegraphics[scale=0.6]{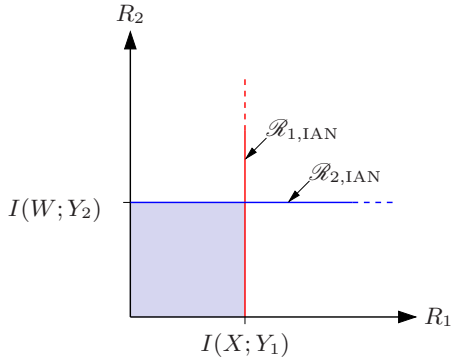}\\[.5em]
\caption{$\RIAN*$ is the intersection of 
the red-lined region $\RIAN1$ and the blue-lined region $\RIAN2$.}
\label{fig:ian}
\end{subfigure}%
~~
\begin{subfigure}[b]{0.48\textwidth}
 \centering
 \small
\psfrag{r1}[b]{$R_2$}
\psfrag{r2}[l]{$R_1$}
\psfrag{i1}[r]{$I(W;Y_2|X)$}
\psfrag{i2}[r]{$I(W;Y_1)$}
\psfrag{i3}[cc]{$I(X;Y_2)$}
\psfrag{i4}[cc]{$I(X;Y_1|W)$}
\psfrag{rr1}[l]{$\RSCD1$}
\psfrag{rr2}[l]{$\RSCD2$}
\includegraphics[scale=0.6]{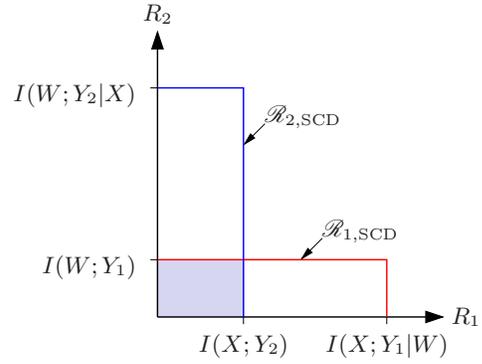}\\[.5em]
\caption{$\RSCD*$ is the intersection of
the red-lined region $\RSCD1$ and the blue-lined region $\RSCD2$.}
\label{fig:sc}
\end{subfigure}%
\\[1em]
\begin{subfigure}[b]{0.48\textwidth}
 \centering
 \small
\psfrag{r1}[b]{$R_2$}
\psfrag{r2}[l]{$R_1$}
\psfrag{i1}[r]{$I(W;Y_2|X)$}
\psfrag{i2}[r]{$I(W;Y_2)$}
\psfrag{i3}[r]{$I(W;Y_1)$}
\psfrag{i4}[cc]{$I(X;Y_2)$}
\psfrag{i5}[cc]{$I(X;Y_1)$}
\psfrag{i6}[cc]{$I(X;Y_1|W)$}
\psfrag{rr1}[l]{$\RIAN1 \cup \RSCD1$}
\psfrag{rr2}[l]{$\RIAN2 \cup \RSCD2$}
\includegraphics[scale = .6]{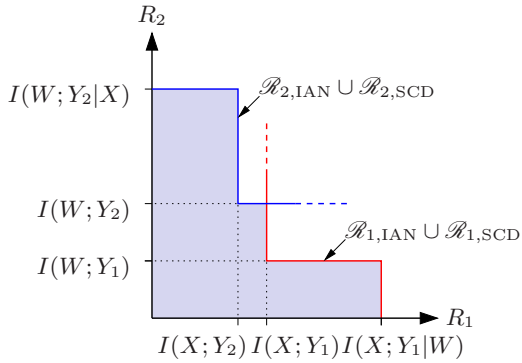}\\[.5em]
\caption{The mix-and-match region is the intersection of the red-lined region $\RIAN1 \cup \RSCD1$ and the blue-lined region $\RIAN2 \cup \RSCD2$.}
\label{fig:comp}
\end{subfigure}%
~~
\begin{subfigure}[b]{.48\textwidth}
 \centering
 \small
 \psfrag{r1}[b]{$R_2$}
\psfrag{r2}[l]{$R_1$}
\psfrag{i1}[r]{$I(W;Y_2|X)$}
\psfrag{i2}[r]{$I(W;Y_2)$}
\psfrag{i3}[r]{$I(W;Y_1)$}
\psfrag{i4}[cc]{$I(X;Y_2)$}
\psfrag{i5}[cc]{$I(X;Y_1)$}
\psfrag{i6}[cc]{$I(X;Y_1|W)$}
\psfrag{i7}[r]{$I(W;Y_1|X)$}
\psfrag{i8}[cc]{$I(X;Y_2|W)$}
\psfrag{rr1}[l]{$\Rr_{1,\mathrm{SND}}$}
\psfrag{rr2}[l]{$\Rr_{2,\mathrm{SND}}$}
\includegraphics[scale=0.6]{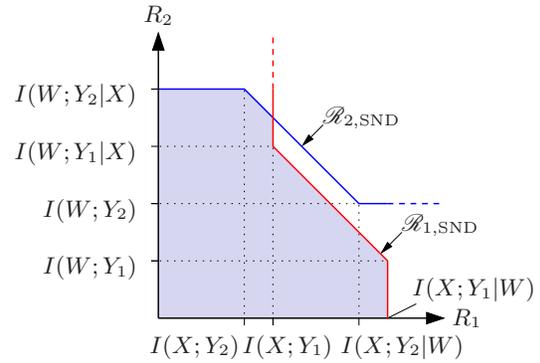}\\[.5em]
\caption{$\Rr_{\mathrm{SND}}$ is the intersection of the red-lined region $\Rr_{1,\mathrm{SND}}$
and the blue-lined region $\Rr_{2,\mathrm{SND}}$. $\Rr_{\mathrm{SND}}$ is identical to the MLD region $\Rr^*$.}
\label{fig:mld}
\end{subfigure}%
\caption{Illustration of the MLD, IAN, SCD regions and their comparison.}
\end{figure}

As illustrated in Fig.~\ref{fig:mld}, $\Rr_{\mathrm{SND}}$ is in general strictly larger than the mix-and-match region in~\eqref{eqn:ian-sc}.

It turns out no decoding rule can improve upon $\Rr_{\mathrm{SND}}$. More precisely, given any codebook $\{(x^n(m_1),w^n(m_2))\}$, the probability of decoding error is minimized by the maximum likelihood decoding (MLD) rule
\begin{equation}
\label{eqn:mld}
\begin{split}
 \mh_1 = \arg\max_{m_1} \sum_{m_2} \prod_{i=1}^n p_{Y_1|X,W} (y_{1i}|x_{i}(m_1),w_{i}(m_2)),\\
 \mh_2 = \arg\max_{m_2} \sum_{m_1} \prod_{i=1}^n p_{Y_2|X,W} (y_{2i}|x_{i}(m_1),w_{i}(m_2)).
 \end{split}
\end{equation}
The \emph{optimal rate region} (or the MLD region) $\Rr^*(p)$  for the $p$-distributed random code ensembles 
is the closure of the set of rate pairs $(R_1,R_2)$ such that 
the sequence of $(2^{nR_1},2^{nR_2},n;p)$ random code ensembles satisfies
\[
 \lim_{n \to \infty} \E[\pen(\Cc_n)] = 0,
\]
where the expectation is with respect to the randomness in codebook generation. It is established in~\cite{Bandemer--El-Gamal--Kim2015} that SND is optimal for the $p$-distributed random code ensembles, i.e.,
\[
 \Rr^* = \Rr_{\mathrm{SND}}.
\]

As shown in Fig.~\ref{fig:mld}, $\Rr^* = \Rr_{\mathrm{SND}}$ is in general strictly larger than the mix-and-match region in~\eqref{eqn:ian-sc}, the gain of which may be attributed to high-complexity multiple sequence detection. The goal of this paper is to develop a coding scheme that achieves $\Rr^*$ using \emph{low-complexity} encoders and decoders.

\section{Rate Splitting for the Interference Channel}
\label{sec:rs}

In order to improve upon the mix-and-match scheme in the previous section at comparable complexity, one can incorporate the rate-splitting technique by Rimoldi and Urbanke~\cite{Rimoldi--Urbanke1996} and Grant, Rimoldi, Urbanke, and Whiting~\cite{Grant--Rimoldi--Urbanke--Whiting2001}. 

\subsection{Rate-Splitting Multiple Access}
\label{sec:rs-one}

Consider the multiple access channel $p(y_1|x,w)$ with two inputs $X$ and $W$ and the common output $Y_1$. It is well-known that simultaneous decoding of the random code ensemble generated according to $p = p(x)p(w)$ achieves $\Rr_{1,\mathrm{SD}}(p)$ in~\eqref{eqn:sd}. In the following, we show how to achieve this region via rate splitting with point-to-point decoders.

Suppose that the message $M_1 \in [2^{nR_1}]$ is split into two parts $(M_{11},M_{12}) \in [2^{nR_{11}}]\times [2^{nR_{12}}]$ while the message $M_2 \in [2^{nR_2}]$ is not split. The messages $m_{11}$ and $m_{12}$ are encoded into codewords $x_1^n$ and $x_2^n$, respectively, which are then symbol-by-symbol mapped to the transmitted sequence $x^n$, that is, $x_i(m_{11},m_{12}) = x(x_{1i}(m_{11}),x_{2i}(m_{12}))$, $i \in [n]$, for some function $x(x_1,x_2)$. The message $m_2$ is mapped to $w^n$. For decoding, the receiver recovers $\mh_{11}$, $\mh_2$, and $\mh_{12}$, successively, which is denoted as the decoding order 
\[
d_1 \suchthat  \mh_{11} \to \mh_2 \to \mh_{12}.
\]
This rate-splitting scheme~\cite{Rimoldi--Urbanke1996} with so-called homogeneous superposition coding~\cite{Wang--Sasoglu--Bandemer--Kim2013} and successive cancellation decoding in Fig.~\ref{fig:rs-mac} can be easily implemented by low-complexity point-to-point encoders and decoders.
\begin{figure}[b]
\centering
\psfrag{m1}[b]{$M_{11}$}
\psfrag{m2}[b]{$M_{12}$}
\psfrag{m3}[b]{$M_2$}
\psfrag{n}[b]{}
\psfrag{u}[b]{$X_1^n$}
\psfrag{v}[b]{$X_2^n$}
\psfrag{x1}[b]{$X^n$}
\psfrag{x2}[b]{$W^n$}
\psfrag{y1}[b]{$Y_1^n$}
\psfrag{y2}[b]{$Y_2^n$}
\psfrag{mh1}[b]{$\Mh_{11} \to \Mh_2 \to \Mh_{12}$}
\psfrag{mh2}[b]{$\Mh_{11} \to \Mh_{12} \to \Mh_2$}
\psfrag{e1}[c]{Encoder $1$}
\psfrag{e2}[c]{Encoder $2$}
\psfrag{d1}[c]{Decoder $1$}
\psfrag{d2}[c]{Decoder $2$}
\psfrag{p}[c]{$p(y_1|x,w)$}
\psfrag{p2}[c]{$p(y_2|x,w)$}
\includegraphics[scale=0.6]{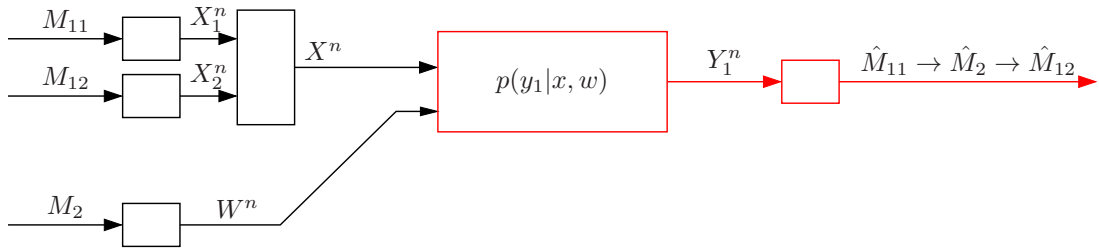}%
\caption{Rate-splitting with successive cancellation for receiver~1.}
\label{fig:rs-mac}
\end{figure}
Following the standard analysis for random code ensembles generated by $p'(x_1)p'(x_2)p'(w)$, decoding is successful if
\begin{align*}
R_{11} &< I(X_1;Y_1),\\
R_{2} &< I(W;Y_1|X_1),\\
R_{12} &< I(X_2;Y_1|X_1,W) = I(X;Y_1|X_1,W).
\end{align*}
By setting $R_1 = R_{11} + R_{12}$, it follows that the scheme achieves the rate region $\Rr_{\mathrm{RS}}(p)$ consisting of $(R_1,R_2)$ such that
\begin{equation}
\label{eqn:rs-mac}
\begin{split}
 R_1 &\le I(X_1;Y_1) + I(X;Y_1|X_1,W),\\
 R_2 &\le I(W;Y_1|X_1).
 \end{split}
\end{equation}
By varying $p'(x_1)p'(x_2)$ and $x(x_1,x_2)$, while maintaining $p'(x) = \sum_{x_1,x_2\suchthat x(x_1,x_2) = x} p'(x_1)p'(x_2) = p(x)$ and $p'(w) = p(w)$, which we compactly denote by $p' \simeq p$, the rectangular region~\eqref{eqn:rs-mac} traces the boundary of rate region $\Rr_{1,\mathrm{SD}}(p)$. More precisely, we have the following identity; see Appendix~\ref{appx:layer-splitting} for the proof.

\begin{lemma}[Layer-splitting lemma~\cite{Grant--Rimoldi--Urbanke--Whiting2001}]
 \label{lem:layer-splitting}
\[
 \Rr_{1,\mathrm{SD}}(p) = \bigcup_{p' \simeq p} \Rr_{\mathrm{RS}}(p').
\]
\end{lemma}

\begin{remark}
\label{rmk:sd-rs}
 Simultaneous decoding of $\Mh_{11},\Mh_{12}$, and $\Mh_{2}$ cannot achieve rates beyond $\Rr_{1,\mathrm{SD}}(p)$ and therefore it does not improve upon (the union of) successive cancellation decoding for the multiple access channel.
\end{remark}

\subsection{Rate Splitting for the Interference Channel}
\label{sec:rs-subopt}
The main idea of rate splitting for the multiple access channel is to represent the messages by multiple parts and encode each into one of the superposition layers. Combined with successive cancellation decoding, this superposition coding scheme transforms the multiple access channel into a sequence of point-to-point channels, over which low-complexity encoders and decoders can be used. For the interference channel with multiple receivers, however, this rate-splitting scheme can no longer achieve the rate region of simultaneous decoding (cf. Remark~\ref{rmk:sd-rs}). The root cause of this deficiency is not rate splitting per se, but
suboptimal successive cancellation decoding. Indeed, proper rate splitting can achieve rates better than no splitting when simultaneous decoding is used (cf.~Han--Kobayashi coding).

To understand the limitations of successive cancellation decoding, we consider the rate-splitting scheme with the same encoder structure as before and two decoding orders
\begin{align*}
d_1 \suchthat & \mh_{11} \to \mh_2 \to \mh_{12}, \\
d_2 \suchthat & \mh_{11} \to \mh_{12} \to \mh_2,
\end{align*}
as depicted in Fig.~\ref{fig:rs}.
\begin{figure}[htbp]
\centering
\psfrag{m1}[b]{$M_{11}$}
\psfrag{m2}[b]{$M_{12}$}
\psfrag{m3}[b]{$M_2$}
\psfrag{n}[b]{}
\psfrag{u}[b]{$X_1^n$}
\psfrag{v}[b]{$X_2^n$}
\psfrag{x1}[b]{$X^n$}
\psfrag{x2}[b]{$W^n$}
\psfrag{y1}[b]{$Y_1^n$}
\psfrag{y2}[b]{$Y_2^n$}
\psfrag{mh1}[b]{$\Mh_{11} \to \Mh_2 \to \Mh_{12}$}
\psfrag{mh2}[b]{$\Mh_{11} \to \Mh_{12} \to \Mh_2$}
\psfrag{e1}[c]{Encoder $1$}
\psfrag{e2}[c]{Encoder $2$}
\psfrag{d1}[c]{Decoder $1$}
\psfrag{d2}[c]{Decoder $2$}
\psfrag{p}[c]{$p(y_1|x,w)$}
\psfrag{p2}[c]{$p(y_2|x,w)$}
\includegraphics[scale=0.6]{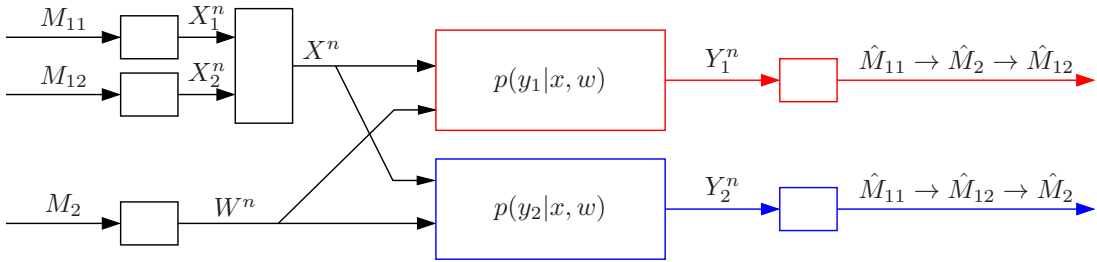}%
\caption{Rate-splitting with successive cancellation in the two-user interference channel.}
\label{fig:rs}
\end{figure}
Following the standard analysis, decoding is successful at receiver~1 if
\begin{subequations}\label{eqn:ex-grp}
\begin{align}
R_{11} &< I(X_1;Y_1),\\
R_{2} &< I(W;Y_1|X_1),\\
R_{12} &< I(X;Y_1|X_1,W).\\
\intertext{and at receiver~2 if }
R_{11} &< I(X_1;Y_2),\label{eqn:ex-r11}\\
R_{12} &< I(X;Y_2|X_1),\label{eqn:ex-r12}\\
R_2 &< I(W;Y_2|X).
\end{align}
\end{subequations}
By Fourier--Motzkin elimination, this scheme achieves the rate region consisting of $(R_1,R_2)$ such that
\begin{subequations}\label{eqn:ex-region}
\begin{align}
R_1 &\le \min\{I(X_1;Y_1),\, I(X_1;Y_2)\} + \min\{ I(X;Y_1|X_1,W),\, I(X;Y_2|X_1)\},
\label{eqn:ex-r1}\\
R_2 &\le \min\{I(W;Y_1|X_1),\, I(W;Y_2|X)\}. \label{eqn:ex-r2}
\end{align}
\end{subequations}

\begin{remark}[Min of the sum vs. sum of the min]
\label{rmk:minsum}
We note a common misconception in the literature, reported also in~\cite{Fawzi--Savov2012} (see the references therein), that the bounds on $R_{11}$ and $R_{12}$ in~\eqref{eqn:ex-grp} would simplify to
\begin{equation}
\label{eqn:min-sum}
 R_1 \le \min\{I(X_1;Y_1)+I(X;Y_1|X_1,W),\, I(X;Y_2)\},
\end{equation}
which could be sufficient to achieve the MLD region $\Rr^*(p)$ in Section~\ref{sec:setup}. This conclusion is incorrect, since the bound in~\eqref{eqn:ex-r1} is strictly smaller than~\eqref{eqn:min-sum} in general. In fact, the rate region in~\eqref{eqn:ex-region}, even after taking the union over all $p'\simeq p$ is strictly smaller than $\Rr^*(p)$. In order to ensure reliable communication over two different underlying multiple access channels $p(y_i|x,w), i = 1,2$, the message parts in the rate-splitting scheme have to be loaded at the rate of the worse channel on each superposition layer, which in general incurs a total rate loss.
\end{remark}

%

It turns out that this deficiency is fundamental and cannot be overcome by introducing more superposition layers and different decoding orders (which include treating interference as noise $d_1\suchthat \mh_{11} \to \mh_{12}$ and $d_2\suchthat \mh_2$). To be more precise, we define the general $(p',s,t,d_1,d_2)$ rate-splitting scheme. The message $M_1$ is split into $s$ independent parts $M_{11}, M_{12}, \ldots, M_{1s}$ with rates $R_{11}, R_{12}, \ldots, R_{1s}$, respectively, and the message $M_2$ is split into $t$ independent parts $M_{21}, M_{22}, \ldots, M_{2t}$ at rates $R_{21}, R_{22}, \ldots, R_{2t}$, respectively. These messages are encoded by the random code ensemble generated according to $p' = \bigl(\prod_{j=1}^s p'(x_j)\bigr) \bigl(\prod_{j=1}^t p'(w_j)\bigr)$ and the corresponding codewords are superimposed by symbol-by-symbol mappings $x(x_1,\ldots,x_s)$ and $w(w_1,\ldots,w_t)$. The receivers use successive cancellation decoding with decoding orders $d_1$ and $d_2$, where $d_1$ is an ordering of elements in $\{\mh_{11},\ldots,\mh_{1s}, \mh_{21}, \ldots, \mh_{2k}\}$, $k \le t$, and $d_2$ is an ordering of elements in $\{\mh_{11},\ldots,\mh_{1l}, \mh_{21},$ $\ldots, \mh_{2t}\}$, $l \le s$. The achievable rate region of this rate-splitting scheme is denoted by $\Rr_{\mathrm{RS}}(p',s,t,d_1,d_2)$. We establish the following statement in Appendix~\ref{appx:insufficiency}.

\begin{theorem}
\label{thm:insufficiency}
There exists an interference channel $p(y_1,y_2|x,w)$ and some input pmf $p = p(x)p(w)$ such that
\[
\bigcup_{p' \simeq p}\Rr_{\mathrm{RS}}(p',s,t,d_1,d_2) \subsetneq \Rr^*(p)
\]
for any finite $s$ and $t$, and decoding orders $d_1$ and $d_2$.
\end{theorem}

\begin{remark}
\label{rmk:bottleneck}
It can be easily checked that the first three regions in the decomposition of $\Rr^*$ in~\eqref{eqn:component} are achievable by properly chosen $(p',2,1,d_1,d_2)$ rate-splitting schemes. The fourth region $\Rr_{1,\mathrm{SD}} \cap \Rr_{2,\mathrm{SD}}$ is the bottleneck in achieving the entire $\Rr^*$ with rate splitting and successive cancellation. 
\end{remark}

\section{Sliding-Window Superposition Coding}
\label{sec:swsc-sd}

In this section, we develop a new coding scheme, termed sliding-window superposition coding (SWSC), that overcomes the limitation of rate splitting by encoding the message to multiple superposition layers across consecutive blocks.

\subsection{Corner Points}
\label{sec:2-1split}
We first show how to achieve the rate region in~\eqref{eqn:ex-r2} and~\eqref{eqn:min-sum}, which will be shown to be sufficient to achieve the corner points of $\Rr_{1,\mathrm{SD}} \cap \Rr_{2,\mathrm{SD}}$. 

In SWSC, we consider a stream of messages, $(m_1(1),m_2(1)), (m_1(2),m_2(2)), \ldots,$ to be communicated over multiple blocks. As before, $m_2(j)$ is encoded into a codeword $w^n$ to be transmitted in block $j$. The message $m_1(j)$, which was split and transmitted in two layers $X_1$ and $X_2$ in the previous rate-splitting scheme, is now encoded into two sequences $x_2^n$ and $x_1^n$ to be transmitted in two consecutive blocks $j$ and $j+1$, respectively; see Table~\ref{tab:sw-rs}. The transmitted sequence $x^n$ in block $j$ is the symbol-by-symbol superposition of $x_1^n(m_1(j))$ and $x_2^n(m_1(j-1))$, which has the same superposition coding structure as in the rate-splitting scheme, but without actual splitting of message rates. Note that similar diagonal transmission of message streams has been already used in block Markov coding for relaying and feedback communication~\cite{Cover--El-Gamal1979,Cover--Leung1981}. For $b$ blocks of communication, the scheme is initialized with $m_1(0) = 1$ and terminated with $m_1(b) = 1$, incurring a slight rate loss. 

\begin{table}[b]
\centering
\begin{tabular}{cc@{}c@{}c@{}c@{}ccc@{}c@{}cc}
block $\j$ & $1$ && $2$ && $3$ & $\cdots$  & $b-1$ && $b$ & \hphantom{AAA}\\
\hline\\[-.8em]
$X_1$ & $1$ && $m_{1}(1)$ & & $m_{1}(2)$  & $\ldots$ & $\ldots$ & & $m_{1}(b-1)$\\[-.2em]
&& $\diagup$ && $\diagup$ && && $\diagup$ & \\[-.2em]
$X_2$ & $m_{1}(1)$ && $m_{1}(2)$ && $\ldots$  & $\ldots$ & $m_{1}(b-1)$ && $1$\\[.5em]
$W$ & $m_{2}(1)$ && $m_{2}(2)$ && $\ldots$ & $\ldots$ & $\ldots$ && $m_{2}(b)$\\[.5em]
\hline\\[-.8em]
&&& $\mh_{1}(1)$ && $\mh_{1}(2)$ & $\ldots$ & $\ldots$ && $\mh_{1}(b-1)$\\
$Y_1$ && $\nearrow$ & $\downarrow$ & $\nearrow$ & $\downarrow$ &&&& $\downarrow$\\
& $\mh_{2}(1)$ && $\mh_{2}(2)$ && $\mh_{2}(3)$ & $\ldots$ & $\ldots$ && $\mh_{2}(b)$\\[.5em]
\hline\\[-.8em]
&&& $\mh_{1}(1)$ && $\mh_{1}(2)$  & $\ldots$ & $\ldots$ && $\mh_{1}(b-1)$\\
$Y_2$ &&& $\downarrow$ & $\nearrow$ & $\downarrow$ &&&& $\downarrow$\\
&&& $\mh_{2}(1)$ && $\mh_{2}(2)$  & $\ldots$ & $\ldots$ &&  $\mh_{2}(b-1) \mathrlap{\to\mh_{2}(b)}$\\[.5em]
\hline\\
\end{tabular}

\caption{Sliding-window superposition coding scheme.}
\label{tab:sw-rs}
\end{table}

For decoding at receiver~1, $\mh_1(j-1)$ and $\mh_2(j)$ are recovered successively from the channel outputs $y_1^n(j-1)$ and $y_1^n(j)$, as shown in Fig.~\ref{fig:2-1diagram}. In the language of typicality decoding, at the end of block $j$, it finds the unique message $\mh_{1}(\j-1)$
such that
\[
 (x_1^n(\mh_{1}(\j-2)), x_2^n(\mh_{1}(\j-1)), w^n(\mh_{2}(\j-1)), y_1^n(\j-1)) \in \aep(X_1,X_2,W,Y_1)
\]
and 
\[
 (x_1^n(\mh_{1}(\j-1)), y_1^n(\j)) \in \aep(X_1,Y_1)
\]
simultaneously, where $\mh_{1}(\j-2)$ and $\mh_{2}(\j-1)$ are already known from the previous block. Then it finds the unique $\mh_{2}(\j)$  such that
\[
(x_1^n(\mh_{1}(\j-1)),w^n(\mh_{2}(\j)), y_1^n(\j)) \in \aep(X_1,W,Y_1).
\]
We represent this successive cancellation decoding operation compactly as
\begin{equation}
\label{eqn:shortchain1}
 d_1\suchthat \mh_1(j-1) \to \mh_2(j),
\end{equation}
which is performed at the end of block $j$. To recover the next pair of messages $\mh_1(j)$ and $\mh_2(j+1)$, receiver~1 slides the decoding window to $y_1^n(j)$ and $y_1^n(j+1)$ at the end of block $j+1$. This sliding-window decoding scheme is originally due to Carleial~\cite{Carleial1982} and used in the network decode--forward relaying scheme~\cite{Xie--Kumar2005,Kramer--Gastpar--Gupta2005}.
The overall schedule of message decoding is shown in Table~\ref{tab:sw-rs}. As can be easily checked by inspection, decoding is successful if
\begin{equation}
 \label{eqn:21swsc-1}
 \begin{split}
 R_1 &< I(X_1;Y_1)+I(X;Y_1|X_1,W),\\
 R_2 &< I(W;Y_1|X_1).
 \end{split}
\end{equation}
\begin{figure*}[b]
 \centering
\begin{subfigure}[b]{0.33\textwidth}
\centering
\begin{tabular}{ccc}
block    & $j-1$ & $j$\\
 \hline\\[-1em]
 $X_1$ & $*$ & $m_1(j-1)$ \\[.5em]
 $X_2$ & $m_1(j-1)$ & $m_1(j)$ \\[.5em]
 $W$   & $*$ & $m_2(j)$ \\[.5em]
\hline
\end{tabular}
\caption{The initial state at the end of block $j$.}
\end{subfigure}%
~
\begin{subfigure}[b]{0.3\textwidth}
\centering
\begin{tabular}{@{\qquad}c@{\qquad\qquad}c@{\qquad}}
   $j-1$ & $j$\\
 \hline\\[-1em]
 $*$ & $*$ \\[.5em]
 $*$  & $m_1(j)$ \\[.5em]
 $*$  & $m_2(j)$ \\[.5em]
\hline
\end{tabular}
\caption{Step 1: recover $\mh_1(j-1)$.}
\end{subfigure}%
~
\begin{subfigure}[b]{0.3\textwidth}
\centering
\begin{tabular}{@{\qquad}c@{\qquad\qquad}c@{\qquad}}
 $j-1$ & $j$\\
 \hline\\[-1em]
 $*$ & $*$ \\[.5em]
 \clap{$*$}  & $m_1(j)$ \\[.5em]
 $*$ & $*$ \\[.5em]
 \hline
\end{tabular}
\caption{Step 2: recover $\mh_2(j)$.}
\end{subfigure}%
\caption{Illustration of the decoding process at receiver~1.}
\label{fig:2-1diagram}
\end{figure*}
A formal proof of this error analysis along with a complete description of the corresponding random coding scheme is delegated to Appendix~\ref{appx:2-1swsc}. Receiver~2 similarly uses successive cancellation decoding at the end of each block $j$ as
\[
 d_2\suchthat m_1(j-1) \to m_2(j-1),
\]
which is successful if
\begin{align*}
 R_1 &< I(X_1;Y_2) + I(X_2;Y_2|X_1) = I(X;Y_2),\\
 R_2 &< I(W;Y_2|X).
\end{align*} 
When the nominal message rate pair of each block is $(R_1,R_2)$, the scheme achieves $(\frac{b-1}{b}R_1,R_2)$ on average, which can be made arbitrarily close to $(R_1,R_2)$ by letting $b \to \infty$. We summarize the performance of this SWSC scheme as follows.

\begin{proposition}
 \label{pro:sw-rs}
 Let $p'(x_1)p'(x_2)p'(w)$ and $x(x_1,x_2)$ be fixed. Then the SWSC scheme in Table~\ref{tab:sw-rs} achieves the rate region $\Rr_{\mathrm{SWSC}}(p',2,1,d_1,d_2)$ that consists of the set of rate pairs $(R_1,R_2)$ such that
 \begin{align*}
 R_1 &\le \min\{I(X_1;Y_1)+I(X;Y_1|X_1,W),\,I(X;Y_2)\},\\
 R_2 &\le \min\{I(W;Y_1|X_1),\,I(W;Y_2|X) \}.
 \end{align*}
\end{proposition}

We now note that each corner point of $\Rr_{1,\mathrm{SD}} \cap \Rr_{2,\mathrm{SD}}$ is contained in one of the four regions
\begin{equation}
\label{eqn:4corner-cases}
\begin{split}
 \Rr_{1,\mathrm{SD}} &\cap \Rr_{2,\mathrm{SCD}1\to 2},\\
 \Rr_{1,\mathrm{SD}} &\cap \Rr_{2,\mathrm{SCD}2\to 1}, \\
 \Rr_{1,\mathrm{SCD}1\to 2} &\cap \Rr_{2,\mathrm{SD}}, \\
 \Rr_{1,\mathrm{SCD}2\to 1} &\cap \Rr_{2,\mathrm{SD}},
\end{split}
\end{equation}
where $\Rr_{j,\mathrm{SCD}1\to 2}, j = 1,2$, is the set of rate pairs $(R_1,R_2)$ such that $R_1 \le I(X;Y_j), R_2 \le I(W;Y_j|X)$ and $\Rr_{j,\mathrm{SCD}2\to 1}$ is the set of rate pairs $(R_1,R_2)$ such that $R_1 \le I(X;Y_j|W), R_2 \le I(W;Y_j)$.
Since any boundary point in $\Rr_{1,\mathrm{SD}}$ can be expressed as~\eqref{eqn:rs-mac} by Lemma~\ref{lem:layer-splitting}, $\Rr_{1,\mathrm{SD}}(p) \cap \Rr_{2,\mathrm{SCD}1\to 2}(p)$ is contained in $\Rr_{\mathrm{SWSC}}(p',2,1)$ for some $p'\simeq p$ and is achieved by the SWSC scheme. The other three regions in~\eqref{eqn:4corner-cases} can be achieved similarly by using different decoding orders, and thus SWSC achieves every corner point of $\Rr_{1,\mathrm{SD}} \cap \Rr_{2,\mathrm{SD}}$.

\begin{remark}
In the SWSC scheme above, for \emph{finite} $b$, there is a rate loss $(1/b) R_1$ for message $M_1$, since no message is scheduled via $X_1^n$ in block 1 and via $X_2^n$ in block $b$. The decoding delay of one block ($\mh_1(j)$ recovered in block $j+1$) is independent of $b$, while the overall probability of error is, by the union-of-events bound, linear in $b$ due to error propagation.
\end{remark} 

\begin{remark} 
\label{rmk:rateloss}
In order to reduce the rate loss, we can instead send message $M_1$ at the treating-interference-as-noise rate $\min\{I(X_1;Y_1),I(X_1;Y_2)\}$ for $X_1^n$ in block 1 and at rate $\min\{I(X;Y_1|X_1,W),I(X;Y_2|X_1)\}$ for $X_2^n$ in block $b$. This increases the overall $R_1$ by 
\begin{align*}
\frac{1}{b}\left[\min\{I(X_1;Y_1),I(X_1;Y_2)\} + \min\{I(X;Y_1|X_1,W),I(X;Y_2|X_1)\}\right],
\end{align*}
which is the same as $1/b$ times the achievable $R_1$ by rate-splitting in~\eqref{eqn:ex-r1}. 
\end{remark}

\subsection{General Rate Points}
\label{sec:3-1split}
The SWSC scheme developed in the previous section cannot achieve the entire region of $\Rr_{1,\mathrm{SD}} \cap \Rr_{2,\mathrm{SD}}$ in general. As illustrated in Fig.~\ref{fig:mac-region-combined}, the scheme can achieve any point on the dominant face of $\Rr_{1,\mathrm{SD}}$ or $\Rr_{2,\mathrm{SD}}$ at the respective receiver. (This is clearly an improvement over the rate-splitting multiple access scheme as noted in Remark~\ref{rmk:minsum}.) In general, however, these two points are not aligned, which may result in a rate region strictly smaller than $\Rr_{1,\mathrm{SD}} \cap \Rr_{2,\mathrm{SD}}$. To overcome this deficiency, we introduce an additional layer to $X$ while keeping $W$ unsplit. The receivers now have the flexibility of merging three layers $X_1,X_2,X_3$ into two groups, for example, $(X_1),(X_2,X_3)$ at receiver~1 and $(X_1,X_2),(X_3)$ at receiver~2, which can align the two points on the dominant faces of $\Rr_{1,\mathrm{SD}}$ and $\Rr_{2,\mathrm{SD}}$ as illustrated in Fig.~\ref{fig:mac-region-combined2}. 
\begin{figure}[htbp]
\centering
\begin{subfigure}[b]{0.48\textwidth}
\centering
\small
\psfrag{r1}[b]{$R_2$}
\psfrag{r2}[l]{$R_1$}
\includegraphics[scale=0.6]{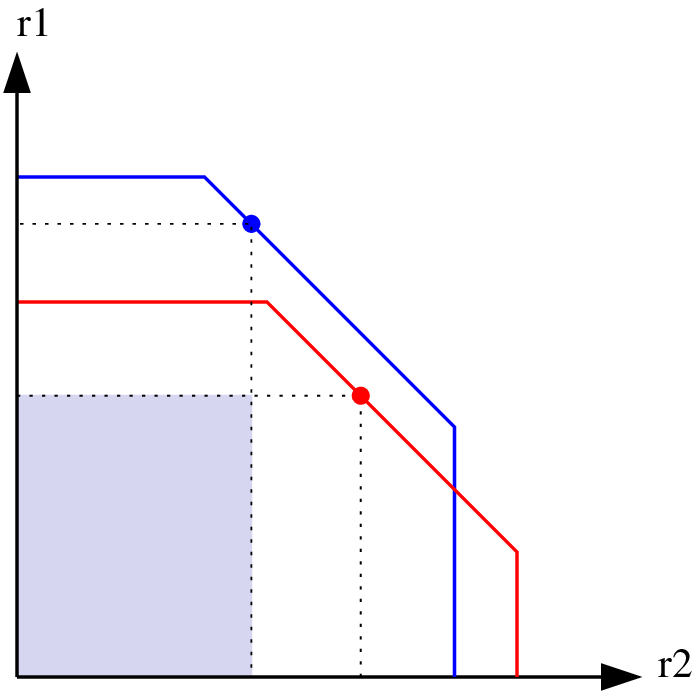}\\[.5em]
\caption{Sum of the min.}
\label{fig:mac-region-combined}
\end{subfigure}%
~~
\begin{subfigure}[b]{0.48\textwidth}
 \centering
 \small
\psfrag{r1}[b]{$R_2$}
\psfrag{r2}[l]{$R_1$}
\includegraphics[scale=0.6]{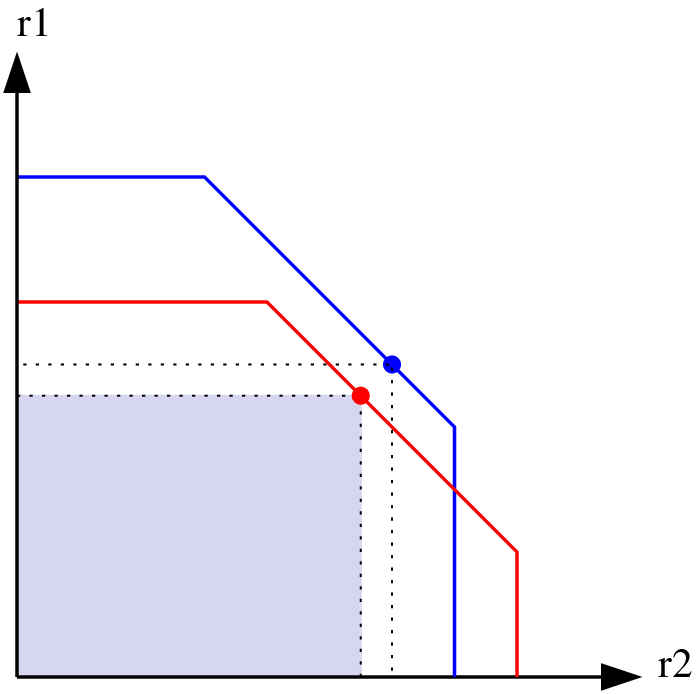}\\[.5em]
\caption{Min of the sum.}
\label{fig:mac-region-combined2}
\end{subfigure}%
\caption{Rate loss of rate splitting in the interference channel.}
\label{fig:loss}
\end{figure}

To be more precise, we first present a coding scheme that achieves the rate region consisting of rate pairs $(R_1,R_2)$ such that
 \begin{equation}
 \label{eqn:31swsc-region1}
 \begin{split}
 R_1 &\le \min\{I(X_1;Y_1) + I(X_2,X_3;Y_1|X_1,W),\,I(X_1,X_2;Y_2) + I(X_3;Y_2|X_1,X_2,W)\},\\
 R_2 &\le \min\{I(W;Y_1|X_1),\,I(W;Y_2|X_1,X_2) \}.
 \end{split}
 \end{equation}
In this SWSC scheme, the message $m_1(j)$ is encoded into three sequences $x_3^n$, $x_2^n$, and $x_1^n$ to be transmitted in three consecutive blocks $j$, $j+1$, and $j+2$, respectively. The message $m_2(j)$ is encoded into a codeword $w^n$ to be transmitted in block $j$. The encoding structure is illustrated in Table~\ref{tab:31swsc1}. The transmitted sequence $x^n$ in block $j$ is the symbol-by-symbol superposition of $x_1^n(m_1(j))$, $x_2^n(m_1(j-1))$, and $x_3^n(m_1(j-2))$.

\begin{table}[b]
\centering
\begin{tabular}{cc@{}c@{}c@{}c@{}c@{}c@{}cc@{}c@{}c@{}c@{}cc}
block $\j$  & $1$ && $2$ && $3$ && $4$ & $\cdots$  && $b-1$ && $b$ & \hphantom{AAA}\\
 \hline\\[-.8em]
  $X_1$ & $1$ && $1$ && $m_{1}(1)$  && $m_{1}(2)$ & $\ldots$ && $\ldots$ &&$m_{1}(b-2)$ \\[-.2em]
  &&  && $\diagup$ && $\diagup$ && &&& $\diagup$ & \\[-.2em]
  $X_2$ & $1$ && $m_{1}(1)$ && $m_{1}(2)$  && $\ldots$ & $\ldots$ && $m_{1}(b-2)$ && $1$ \\[-.2em]
 && $\diagup$ && $\diagup$ && &&& $\diagup$ & && \\[-.2em]
 $X_3$ & $m_{1}(1)$ && $m_{1}(2)$ && $\ldots$  && $\ldots$ & $m_{1}(b-2)$ && $1$ && $1$ \\[.5em]
 $W$ & $m_{2}(1)$ && $m_{2}(2)$ && $\ldots$ && $\ldots$ & $\ldots$ && $\ldots$ && $m_{2}(b)$ \\[.5em]
\hline\\[-.8em]
 &&&&& $\mh_{1}(1)$ && $\mh_{1}(2)$ & $\ldots$ && $\ldots$ && $\mh_{1}(b-2)$ \\
 $Y_1$ &&&& $\nearrow$ & $\downarrow$ & $\nearrow$ & $\downarrow$ &&&&& $\downarrow$\\
 & $\mh_{2}(1)$ & $\to$ &$\mh_{2}(2)$ && $\mh_{2}(3)$ && $\mh_{2}(4)$ & $\ldots$ && $\ldots$ && $\mh_{2}(b)$ \\[.5em]
\hline\\[-.8em]
 &&&&& $\mh_{1}(1)$ && $\mh_{1}(2)$ & $\ldots$ && $\ldots$ && $\mh_{1}(b-2)$ \\
 $Y_2$ &&&& $\nearrow$ & $\downarrow$ & $\nearrow$ & $\downarrow$ &&&&& $\downarrow$\\
 &&& $\mh_{2}(1)$ && $\mh_{2}(2)$ && $\mh_{2}(3)$ & $\ldots$ &&  $\ldots$ && $\mh_{2}(b-1) \mathrlap{\to\mh_{2}(b)}$\\[.5em]
\hline\\
\end{tabular}
\caption{SWSC scheme with decoding orders in~\eqref{eqn:31split-d1}.}
\label{tab:31swsc1}
\end{table}

\smallskip
For decoding, the message $\mh_{1}(\j)$ is recovered via sliding-window decoding over three blocks. The decoding orders at two receivers are
\begin{subequations}
 \label{eqn:31split-d1}
 \begin{align}
 d_1\suchthat & \mh_1(j-2) \to \mh_2(j),\label{eqn:d1-r1}\\
 d_2\suchthat & \mh_1(j-2) \to \mh_2(j-1).\label{eqn:d1-r2}
 \end{align}
\end{subequations}
The decoding process is illustrated in Table~\ref{tab:31swsc1}. 
Following the standard analysis, the decoding is successful at receiver~1 if
\begin{align*}
\begin{split}
 R_1 &< I(X_1;Y_1) + I(X_2;Y_1|X_1,W) + I(X_3;Y_1|X_1,X_2,W),\\
 R_2 &< I(W;Y_1|X_1), 
 \end{split}
\end{align*}
and at receiver~2 if
\begin{align*}
\begin{split}
 R_1 &< I(X_1;Y_2) + I(X_2;Y_2|X_1) + I(X_3;Y_2|X_1,X_2,W), \\
 R_2 &< I(W;Y_2|X_1,X_2),
 \end{split}
\end{align*}
which establishes the achievability of the rate region in~\eqref{eqn:31swsc-region1}. We denote this rate region by $\Rr_{\mathrm{SWSC}}(p',3,1,d_1,d_2)$.

\smallskip

By swapping the decoding orders between receivers~1~and~2, i.e., 
 \begin{subequations}
 \label{eqn:31split-d2}
 \begin{align}
d_1'\suchthat & \mh_1(j-2) \to \mh_2(j-1), \label{eqn:d2-r1}\\
d_2'\suchthat & \mh_1(j-2) \to \mh_2(j),\label{eqn:d2-r2}
  \end{align}
 \end{subequations}
the SWSC scheme achieves the rate region $\Rr_{\mathrm{SWSC}}(p',3,1,d_1',d_2')$ characterized by
\begin{align*}
 R_1 &\le \min\{I(X_1,X_2;Y_1) + I(X_3;Y_1|X_1,X_2,W), \,I(X_1;Y_2) + I(X_2,X_3;Y_2|X_1,W)\},\\
 R_2 &\le \min\{I(W;Y_2|X_1),\,I(W;Y_1|X_1,X_2) \}.
\end{align*}


This SWSC scheme turns out to be sufficient to achieve any rate point in the simultaneous decoding region; see Appendix~\ref{appx:31swsc} for the proof.

\begin{proposition}
 \label{pro:31swsc}
 \[
\Rr_{1,\mathrm{SD}}(p) \cap \Rr_{2,\mathrm{SD}}(p) = \bigcup_{p'\simeq p} \bigcup_{(d_1,d_2)=\eqref{eqn:31split-d1} \text{ or}~\eqref{eqn:31split-d2}} \Rr_{\mathrm{SWSC}}(p',3,1,d_1,d_2).
\]
\end{proposition}

\subsection{SWSC Achieves the MLD Region $\Rr^*$}
\label{sec:swsc-mld}
We now show that the other three component regions of $\Rr^*$ in~\eqref{eqn:component}, namely, $\Rr_{1,\mathrm{IAN}} \cap \Rr_{2,\mathrm{IAN}}$, $\Rr_{1,\mathrm{SD}} \cap \Rr_{2,\mathrm{IAN}}$, and $\Rr_{1,\mathrm{IAN}} \cap \Rr_{2,\mathrm{SD}}$, can be also achieved by the SWSC scheme in Table~\ref{tab:31swsc1} (with the same encoding scheme, but with different decoding orders).

\begin{itemize}
 \item $\Rr_{1,\mathrm{IAN}} \cap \Rr_{2,\mathrm{IAN}}$: 
 \begin{subequations}
 \label{eqn:31split-d3}
 \begin{align}
  d_1\suchthat & \mh_1(j-2),\label{eqn:d3-r1}\\
  d_2\suchthat & \mh_2(j).\label{eqn:d3-r2}
  \end{align}
 \end{subequations}
 The corresponding achievable rate region is the set of rate pairs $(R_1,R_2)$ such that
\begin{align*}
 R_1 &\le I(X_1;Y_1) + I(X_2;Y_1|X_1) + I(X_3;Y_1|X_1,X_2) = I(X;Y_1),\\
 R_2 &\le I(W;Y_2).
\end{align*}

\item $\Rr_{1,\mathrm{SD}} \cap \Rr_{2,\mathrm{IAN}}$: 
 \begin{subequations}
 \label{eqn:31split-d4}
 \begin{align}
  d_1\suchthat & \mh_1(j-2) \to \mh_2(j), \label{eqn:d4-r1}\\
  d_2\suchthat & \mh_2(j). \label{eqn:d4-r2}
  \end{align}
 \end{subequations}
The corresponding achievable rate region is the set of rate pairs $(R_1,R_2)$ such that
 \begin{align*}
 R_1 &\le I(X_1;Y_1) + I(X;Y_1|X_1,W),\\
 R_2 &\le \min\{ I(W;Y_1|X_1),\, I(W;Y_2)\},
\end{align*}
which, after taking the union over all $p'\simeq p$, is equivalent to $\Rr_{1,\mathrm{SD}}(p) \cap \Rr_{2,\mathrm{IAN}}(p)$ by Lemma~\ref{lem:layer-splitting}.

\item $\Rr_{1,\mathrm{IAN}} \cap \Rr_{2,\mathrm{SD}}$: 
 \begin{subequations}
 \label{eqn:31split-d5}
 \begin{align}
  d_1\suchthat & \mh_1(j-2), \label{eqn:d5-r1}\\
  d_2\suchthat & \mh_1(j-2) \to \mh_2(j). \label{eqn:d5-r2}
  \end{align}
 \end{subequations}
 The corresponding achievable rate region is the set of rate pairs $(R_1,R_2)$ such that
\begin{align*}
 R_1 &\le \min\{I(X;Y_1), \, I(X_1;Y_2) + I(X;Y_2|X_1,W)\},\\
 R_2 &\le I(W;Y_2|X_1),
\end{align*}
which, after taking the union over all $p'\simeq p$, is equivalent to $\Rr_{1,\mathrm{IAN}}(p) \cap \Rr_{2,\mathrm{SD}}(p)$ by Lemma~\ref{lem:layer-splitting}.
\end{itemize}

In summary, the SWSC scheme in Table~\ref{tab:31swsc1}, with $p' \simeq p$ and decoding orders~\eqref{eqn:31split-d1}--\eqref{eqn:31split-d5}, achieves the MLD region $\Rr^*$.

\begin{theorem}
 \label{thm:swsc-opt}
 \[
  \Rr^*(p) = \bigcup_{p'\simeq p} \bigcup_{(d_1,d_2)=\eqref{eqn:31split-d1}\text{--}\eqref{eqn:31split-d5}} \Rr_{\mathrm{SWSC}}(p',3,1,d_1,d_2).
 \]
\end{theorem}

\section{Sliding-Window Coded Modulation}
\label{sec:swcm}

Coded modu lation is the interface between channel coding and modulation, and specifies how (typically binary) codewords are mapped to sequences of constellation points. In this section, we show how the SWSC scheme can be specialized to a coded modulation scheme, termed \emph{sliding-window coded modulation (SWCM)}, and demonstrate through practical implementation that conventional point-to-point encoders and decoders can be utilized to achieve the performance expected from high-complexity coding schemes. We also compare SWCM with existing coded modulation schemes, such as multilevel coding (MLC)~\cite{Imai--Hirakawa1977,Wachsmann--Fischer--Huber1999} and bit-interleaved coded modulation (BICM)~\cite{Zehavi1992,Caire--Taricco--Biglieri1998}.

\subsection{An Illustration of SWCM for 4PAM}

Each coded modulation scheme is specified by two mappings: the symbol-level mapping and the block-level mapping. In SWCM, the symbol-level mapping is specified by the symbol-by-symbol mapping in superposition coding. For example, let $X_1,X_2 \in \{-1,+1\}$ be two BPSK symbols (throughout this section we assume the unit power constraint). Then a uniformly-spaced 4-PAM signal can be formed as
\begin{equation}
\label{eqn:4pam1}
 X = \frac{1}{\sqrt{5}}(X_1 + 2 X_2) \in \{-\frac{3}{\sqrt{5}}, -\frac{1}{\sqrt{5}}, \frac{1}{\sqrt{5}}, \frac{3}{\sqrt{5}}\}.
\end{equation}

The block-level mapping of SWCM is specified by the message scheduling of SWSC. For example, in the encoding scheme in Table~\ref{tab:sw-rs}, each message is encoded to a length-$2n$ binary codeword (potentially with interleaving), the first $n$ bits of which are carried by $X_2$ symbols in the current block, and the second $n$ bits of which are carried by $X_1$ symbols in the next block. 
Accordingly, each transmissed symbol $X$ is then generated by~\eqref{eqn:4pam1}, using a symbol $X_2$ from the current codeword and a symbol $X_1$ from the previous codeword.
See Fig.~\ref{fig:swcm-enc} for an illustration of the symbol-level and block-level mappings of the SWCM scheme that corresponds to Table~\ref{tab:sw-rs}. 

It is instructive to compare SWCM with two other popular coded modulation schemes, BICM and MLC. The key difference among the three lies in the block-level mapping; see Fig.~\ref{fig:comparison}. Assuming the same symbol-level mapping~\eqref{eqn:4pam1}, in BICM, the two length-$n$ parts $x_1^n$ and $x_2^n$ of a length-$2n$ codeword are transmitted in the same block. This contrasts the staggered transmission of $x_1^n$ and $x_2^n$ in SWCM. In MLC, instead of a single length-$2n$ codeword, two standalone length-$n$ codewords $x_1^n$ and $x_2^n$ are generated by splitting the message (say $M$) into two parts (say $M'$ and $M''$). When used for a point-to-point channel $p(y|x)$, SWCM achieves
\[
 I(X_1;Y)+ I(X_2;Y|X_1) = I(X;Y).
\]
MLC achieves the same rate if individual rates of the two component codes are properly matched, while BICM achieves
\[
 I(X_1;Y)+I(X_2;Y) < I(X;Y),
\]
the loss in which is due to self-interference between $X_1$ and $X_2$. The finite-block performance is better in SWCM and BICM than in MLC thanks to the longer codeword length of $2n$. More fundamentally, individual component codewords in MLC should be rate-controlled (which is difficult to be done optimally in practice) and reliably decoded (which results in rate loss under channel uncertainty or multiple receivers). The latter limitation is reflected in the deficiency of the rate-splitting scheme for the interference channel, as pointed out in Remark~\ref{rmk:minsum}. In summary, SWCM has the advantage of high rate over BICM and the advantage of long block length and robustness over MLC, but at the same time suffers from error propagation over blocks and rate loss due to initialization/termination. 
\begin{figure*}[b]
 \centering
\begin{subfigure}[b]{0.33\textwidth}
\small
\centering
\psfrag{U}[r]{$X_1$}
\psfrag{V}[r]{$X_2$}
\psfrag{a2}[c]{$M'$}
\psfrag{b2}[c]{$M''$}
\includegraphics[scale=0.4]{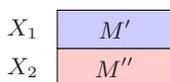}
\caption{MLC: Short, nonuniversal.\\ $R' < I(X_1; Y), R'' < I(X_2; Y | X_1)$.}
\end{subfigure}%
~
\begin{subfigure}[b]{0.33\textwidth}
\small
\centering
\psfrag{U}[r]{$X_1$}
\psfrag{V}[r]{$X_2$}
\psfrag{a2}[c]{$M$}
\psfrag{b1}[c]{$M$}
\includegraphics[scale=0.4]{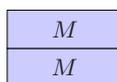}
\caption{BICM: Treat other layers as noise. \\$R < I(X_1; Y) + I(X_2; Y)$. }
\end{subfigure}%
~
\begin{subfigure}[b]{0.33\textwidth}
\small
\centering
\psfrag{U}[r]{$X_1$}
\psfrag{V}[r]{$X_2$}
\psfrag{a2}[c]{$M$}
\psfrag{b1}[c]{$M$}
\includegraphics[scale=0.4]{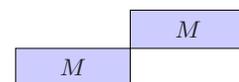}
\caption{SWCM: Error propagation, rate loss.\\ $R < I(X_1; Y) + I(X_2; Y | X_1)$.}
\end{subfigure}
\caption{Comparison of three coded modulation schemes.}
\label{fig:comparison}
\end{figure*}

\subsection{The Generalization to Other Constellations}

The SWSC framework provides great flexibility in the symbol-level mapping and the number of layers, which results in a variety of practical coded modulation schemes. For example, a Gray mapping from two BPSK symbols to the 4PAM constellation can be formed by a different symbol-level mapping
\begin{equation}
\label{eqn:4pam2}
 X = \frac{1}{\sqrt{5}}(X_1 + 2 X_1 \cdot X_2).
\end{equation}
There are four other symbol-level mappings for 4PAM.

Higher-order constellations have richer structures and allow for more diverse decompositions. For example, a uniformly-spaced 8PAM symbol can be decomposed as the superposition
\begin{equation}
 \label{eqn:8pam1}
 X = \frac{1}{\sqrt{21}}(X_1 + 2 X_2 + 4 X_3) 
\end{equation}
of three BPSK layers $X_1,X_2,X_3,$ or as the superposition
\begin{equation}
 \label{eqn:8pam2}
 X = \frac{1}{\sqrt{21}}(X_1 + 2\sqrt{5} X_2)
\end{equation}
of one BPSK layer $X_1$ and one 4PAM layer $X_2$. For the block-level mapping, each message is encoded into a length-$3n$ binary codeword. In case of~\eqref{eqn:8pam1}, the three parts of the codeword, each of length $n$, are transmitted over three consecutive blocks. In case of~\eqref{eqn:8pam2}, the first $2n$ bits of the codeword are carried by the 4PAM $X_2$ sequence (2 bits per symbol by the Gray or natural mapping) and the remaining $n$ bits are carried by the BPSK $X_1$ sequence over two consecutive blocks.

As another example, consider the 16QAM coded modulation, which can be decomposed as the superposition
\begin{equation}
 \label{eqn:16qam1}
 X = \frac{1}{\sqrt{5}}(X_1 + 2X_2)
\end{equation}
of two QPSK symbols $X_1,X_2 \in \{e^{i\frac{\pi}{4}},e^{i\frac{3\pi}{4}},e^{-i\frac{3\pi}{4}},e^{-i\frac{\pi}{4}}\}$, or as the superposition
\begin{equation}
 \label{eqn:16qam2}
 X = \frac{1}{\sqrt{2}}(X_1+ iX_2)
\end{equation}
of two 4PAM symbols $X_1,X_2 \in \{-\frac{3}{\sqrt{5}}, -\frac{1}{\sqrt{5}}, \frac{1}{\sqrt{5}}, \frac{3}{\sqrt{5}}\}$. For both cases, two halves of a length-$4n$ binary codeword are carried by $x_1^n$ and $x_2^n$ over two consecutive blocks. Alternatively, four BPSK layers can be used for staggered transmission over four consecutive blocks. 

For multiple-input multiple-output (MIMO) transmission, there is a natural correspondence between the antenna ports and the symbol-level mapping. Suppose that there are $t$ transmitting antennas. Then, each antenna port  $X^{(k)}$ can transmit the codeword carried by the SWCM layer $X_k$, that is,
\begin{equation}
 \label{eqn:dblast}
 \begin{split}
 X &= (X^{(1)},\ldots,X^{(t)})\\
 &= (X_1,\ldots,X_t).
 \end{split}
\end{equation}
The SWCM scheme with the symbol-level mapping in~\eqref{eqn:dblast} is in fact equivalent to the block-level diagonal Bell Labs layered space-time (D-BLAST) architecture~\cite{Foschini1996}. Note that horizontal BLAST~\cite{Li--Huang--Foschini--Valenzuela2000,Foschini--Chizhik--Gans--Papadias--Valenzuela2003} and vertical BLAST~\cite{Wolniansky--Foschini--Golden--valenzuela1998} correspond to MLC and BICM, respectively. In this sense, the encoder structure of sliding-window superposition coding may well be called \emph{diagonal superposition coding} in contrast to the conventional \emph{horizontal} superposition coding structure of MLC. 

SWCM, however, can provide much greater flexibility than D-BLAST since the symbol-level mapping can be controlled at the constellation level, not just at the antenna level. For example, consider a MIMO system with two transmitting antennas, both of which use the 4PAM constellation as in~\eqref{eqn:4pam1}
\begin{equation}
\begin{split}
X^{(1)} &= \frac{1}{\sqrt{5}}(A_{11}+2A_{12}),\\                                                                                                                                                                                                                                                                                                                                                                                                                                                                    
X^{(2)} &= \frac{1}{\sqrt{5}}(A_{21}+2A_{22}),
\end{split}
\end{equation}
where $A_{11},A_{12},A_{21},A_{22}$ are BPSK symbols. As in D-BLAST, we can use the symbol-level mapping in~\eqref{eqn:dblast}, or equivalently,
\[
 X_1 = (A_{11},A_{12}), \quad X_2 = (A_{21},A_{22}),
\]
and communicate the two halves of a length-$4n$ binary codeword by $x_2^n$ and $x_1^n$ over two consecutive blocks. As an alternative, we can map the least significant bits in the two antennas to layer~1 and the remaining bits to layer~2, i.e.,
\[
 X_1 = (A_{11},A_{21}), \quad X_2=(A_{12},A_{22}).
\]
As another alternative, we can use 4 layers with symbols $A_{11},A_{12},A_{21},A_{22}$, each carrying one fourth of the codewords over four consecutive blocks. There can be other possibilities. This richness can be utilized for adaptive transmission for wireless fading channels, as demonstrated in~\cite{Kim--Ahn--Kim--Park--Wang--Chen--Park2015}.

\subsection{Implementation With LTE Turbo Codes}
\label{sec:implementation}

We now demonstrate the feasibility of SWCM in practice by implementing the basic 4PAM coded modulation scheme in~\eqref{eqn:4pam1} for the Gaussian interference channel. More extensive studies for cellular networks are reported in~\cite{Kim--Ahn--Kim--Park--Chen--Kim2016}.

Consider the 2-user Gaussian interference channel in~\eqref{eq:channel}, where sender~1 uses 4PAM as in~\eqref{eqn:4pam1} and sender~2 uses BPSK. Sender~1 uses a binary code of length $2n$ and rate $R_1/2$ to communicate $m_1(j)$ through $x_2^n$ in block $j$ and $x_1^n$ in block $j+1$, while sender~2 uses a binary code of length $n$ and rate $R_2$ to communicate $m_2(j)$ through $w^n$ in block $j$; see Fig.~\ref{fig:swcm-enc}. 
\begin{figure}[b]
\centering
 \footnotesize
 \psfrag{sender1}[cc]{Sender $1$}
 \psfrag{sender2}[cc]{Sender $2$}
 \psfrag{m11}[cc]{$m_1(j-1)$}
 \psfrag{m12}[cc]{$m_1(j)$}
 \psfrag{c11}[cc]{$c_1(j-1)$}
 \psfrag{c12}[cc]{$c_1(j)$}
 \psfrag{ct11}[cc]{$\tilde{c}_1(j-1)$}
 \psfrag{ct12}[cc]{$\tilde{c}_1(j)$}
 \psfrag{v1}[cc]{$X_1(j-1)$}
 \psfrag{v2}[cc]{$X_1(j)$}
 \psfrag{u2}[cc]{$X_2(j)$}
 \psfrag{u3}[cc]{$X_2(j+1)$}
 \psfrag{x}[cc]{$X(j)$}
 \psfrag{m2}[cc]{$m_2(j)$}
 \psfrag{c2}[cc]{$c_2(j)$}
 \psfrag{ct2}[cc]{$\tilde{c}_2(j)$}
 \psfrag{w1}[cc]{$W(j)$}
 \psfrag{t1}[lb]{Turbo encoding with rate matching}
 \psfrag{t2}[lb]{Interleaving and scrambling}
 \psfrag{t3}[lb]{BPSK modulation}
 \psfrag{t4}[lb]{Superposition and block Markov coding}
 \psfrag{t5}[l]{$=\frac{1}{\sqrt{5}}(X_1(j) + 2 X_2(j))$}
 \includegraphics[scale = 0.47]{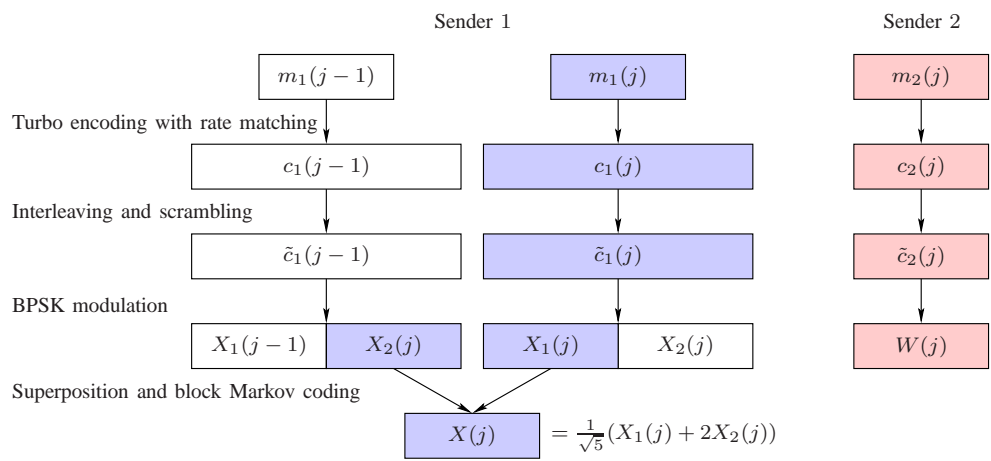}
 \caption{Encoding diagram for the LTE-turbo implementation of SWCM.}
  \label{fig:swcm-enc}
\end{figure}
We adopt the LTE standard turbo code~\cite{LTE}, which has the flexibility in the code rate and the block length. In particular, we start with the rate $1/3$ mother code and adjust the rates and lengths according to the rate matching algorithm in the standard. Note that for $R_1 < 2/3$, some code bits are repeated and for $R_1 > 2/3$, some code bits are punctured. We set the block length $n = 2048$ and the number of blocks $b = 20$. We use the LOG-MAP algorithm with up to 8 iterations in each stage of turbo decoding. We assume that a rate pair $(R_1,R_2)$ is achieved for a given channel if the resulting block-error rate (BLER) is below $0.1$ over 200 independent sets of simulations. Sliding-window decoding is performed at both receivers. Fig.~\ref{fig:swcm-dec} illustrates the decoding operation at receiver~1, under decoding order $d_1\suchthat \mh_1(j-1) \to \mh_2(j)$.

\begin{figure}[t]
\footnotesize
 \psfrag{y1}[cc]{$Y_1(j-1)$}
 \psfrag{y2}[cc]{$Y_1(j)$}
 \psfrag{x11}[cc]{$X_1(j-1)$}
 \psfrag{x12}[cc]{$X_1(j)$}
 \psfrag{x21}[cc]{$X_2(j-1)$}
 \psfrag{x22}[cc]{$X_2(j)$}
 \psfrag{w1}[cc]{$W(j-1)$}
 \psfrag{w2}[cc]{$W(j)$}
 \psfrag{l1}[cc]{$l_1(j-1)$}
 \psfrag{lt1}[cc]{$\tilde{l}_1(j-1)$}
 \psfrag{l2}[cc]{$l_2(j)$}
 \psfrag{lt2}[cc]{$\tilde{l}_2(j)$}
 \psfrag{m1}[cc]{$\mh_1(j-1)$}
 \psfrag{m2}[cc]{$\mh_2(j)$}
 \psfrag{t1}[lb]{known from}
 \psfrag{t2}[lc]{the previous block}
 \psfrag{t3}{LLR calculation}
 \psfrag{t4}{Descrambling and deinterleaving}
 \psfrag{t5}{Turbo decoding}
 \psfrag{t6}[lc]{treated as noise}
 \psfrag{t7}{~~~Encoding for}
 \psfrag{t8}[lc]{known from phase 1}
 \psfrag{t9}{successive cancellation}
 \psfrag{t10}{Encoding for successive}
 \psfrag{t11}{cancellation in the next block}
 \psfrag{t12}[lc]{to be recovered}
 \psfrag{ph1}[cc]{Phase 1}
 \psfrag{ph2}[cc]{Phase 2}
\hspace{1em}
 \includegraphics[scale = 0.48]{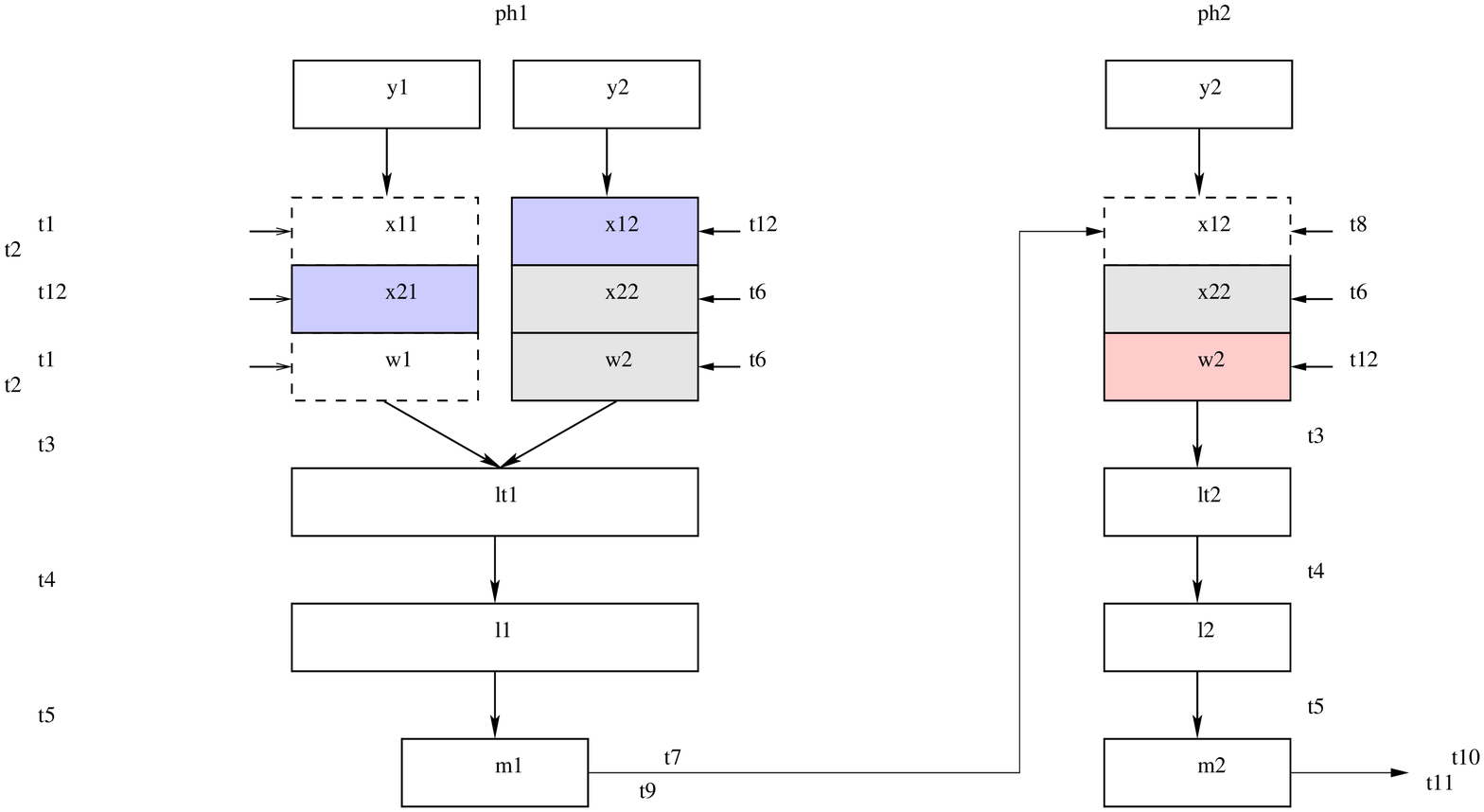}
 \caption{Decoding diagram for the decoding order $\mh_1(j-1)\to \mh_2(j)$.}
 \label{fig:swcm-dec}
\end{figure}

Fig.~\ref{fig:swcm-simulation} plots the symmetric rate $(R_1 = R_2)$ against the INR for the symmetric Gaussian interference channel $(S_1 = S_2 \text{ and } I_1 = I_2)$ when the SNR is held fixed at 8 dB. The solid lines represent theoretical achievable rates (mutual information) of MLD/SND, SWCM, and IAN.
In IAN decoding, the interference is treated as \emph{Gaussian} noise of the same power
and the constellation information of interference is not used. In SWCM decoding, the optimal decoding orders are
used at the given channel parameters. There is a gap between MLD/SND and SWCM (cf.\@ Theorem~\ref{thm:swsc-opt}), since the encoder is fixed using a symbol-level mapping~\eqref{eqn:4pam1} with only two layers $X_1,X_2 \sim\U\{-1,+1\}$.
The dashed lines represent the achievable rates of the actual implementation using the LTE turbo codes.
The 4PAM encoding at sender~1 uses BICM for IAN.
As the INR grows, the gain of SWCM over IAN increases from 53.44\% (at the INR of 6 dB) 
to 150.32\% (at 8 dB) and to 266.51\% (at 10 dB). 

\begin{figure}[htbp]
\centering
 \small
 \psfrag{aa}[l][l]{\footnotesize MLD/SND}
 \psfrag{bb}[l][l]{\footnotesize SWCM}
 \psfrag{cc}[l][l]{\footnotesize SWCM (turbo)}
 \psfrag{dd}[l][l]{\footnotesize IAN}
 \psfrag{ee}[l][l]{\footnotesize IAN (turbo)}
 \psfrag{S}[c][c]{Symmetric Rate (bits/s/Hz)}
 \psfrag{I}[c][c]{INR (dB)}
 \psfrag{6}[l]{$6$}
 \psfrag{6.5}{$6.5$}
\includegraphics[width=0.6\linewidth]{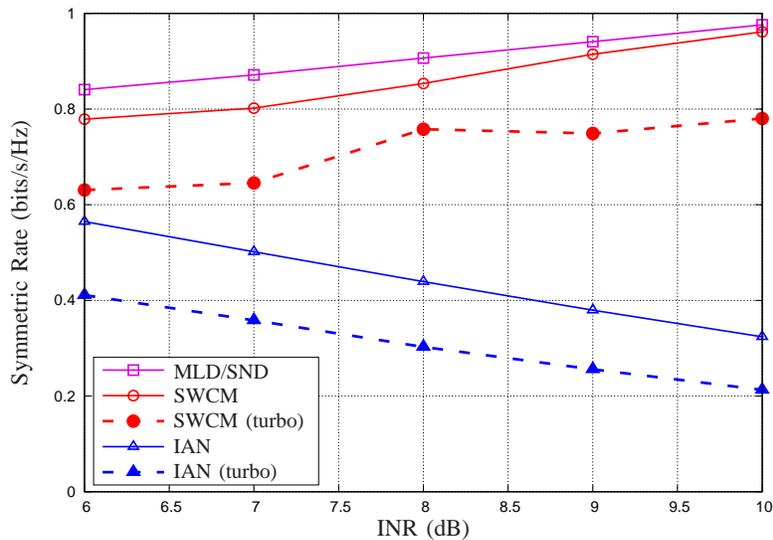}
\caption{Performance comparison in the symmetric Gaussian interference channel. The solid lines correspond to the theoretical performance. The dashed lines correspond to the simulation performance by the implementations using the LTE turbo codes.}
\label{fig:swcm-simulation}
\end{figure}

\section{Superposition Layers and Decoding Orders}
\label{sec:layering-order}

SWSC with a given encoder structure allows multiple decoding schemes, each with a different rate region. In this section, we provide a more systematic treatment of the relationship between superposition layers and decoding orders.


Suppose that we split $X$ into $K$ layers $(X_1,\ldots,X_K)$ and $W$ into $L$ layers $(W_1,\ldots,W_L)$. Consider a stream of messages, $(m_1(1),m_2(1)), (m_1(2),m_2(2)), \ldots,$ to be communicated over multiple blocks. The message $m_1(j)$ is encoded into $K$ sequences $x_K^n$, $x_{K-1}^n,\ldots,$ and $x_1^n$ to be transmitted in $K$ consecutive blocks $j$, $j+1,\ldots,$ and $j+K-1$, respectively. Similarly, the message $m_2(j)$ is encoded into $L$ sequences $w_L^n$, $w_{L-1}^n,\ldots,$ and $w_1^n$ to be transmitted in $L$ consecutive blocks $j$, $j+1,\ldots,$ and $j+L-1$, respectively. The transmitted sequence $x^n$ in block $j$ is the symbol-by-symbol superposition of $x_1^n(m_1(j))$, $x_2^n(m_1(j-1)),\ldots,$ and $x_K^n(m_1(j-K+1))$. The transmitted sequence $w^n$ in block $j$ is the symbol-by-symbol superposition of $w_1^n(m_2(j))$, $w_2^n(m_2(j-1)),\ldots,$ and $w_L^n(m_2(j-L+1))$. We refer to such a layer split and message schedule as the \emph{$K$-$L$ split}. Table~\ref{tab:3-2-one-dim} illustrates the encoding of the $3$-$2$ split.

\begin{table}[b]
\centering
\begin{tabular}{cc@{}c@{}c@{}c@{}c@{}c@{}cc@{}c@{}c@{}c@{}c}
block $\j$  & $1$ && $2$ && $3$ && $4$ & $\cdots$  && $b-1$ && $b$\\
 \hline\\[-.8em]
  $X_1$ & $1$ && $1$ && $m_{1}(1)$  && $m_{1}(2)$ & $\ldots$ && $\ldots$ &&$m_{1}(b-2)$ \\[-.2em]
  &&  && $\diagup$ && $\diagup$ && &&& $\diagup$ & \\[-.2em]
  $X_2$ & $1$ && $m_{1}(1)$ && $m_{1}(2)$  && $\ldots$ & $\ldots$ && $m_{1}(b-2)$ && $1$ \\[-.2em]
 && $\diagup$ && $\diagup$ && &&& $\diagup$ & && \\[-.2em]
 $X_3$ & $m_{1}(1)$ && $m_{1}(2)$ && $\ldots$  && $\ldots$ & $m_{1}(b-2)$ && $1$ && $1$ \\[.5em]
 $W_1$ & $1$ && $m_{2}(1)$ && $m_{2}(2)$  && $\ldots$ & $\ldots$ && $\ldots$ && $m_{2}(b-1)$\\[-.2em]
 && $\diagup$ && $\diagup$ && &&& && $\diagup$ &  \\[-.2em]
 $W_2$ & $m_{2}(1)$ && $m_{2}(2)$ && $\ldots$ && $\ldots$ & $\ldots$ && $m_2(b-1)$ && $1$\\[.5em]
 \hline\\
\end{tabular}
\caption{SWSC encoding with a 3-2 split.}
\label{tab:3-2-one-dim}
\end{table}

As we saw in the previous section, different decoding orders may result in different achievable rate regions. A feasible decoding order for a $K$-$L$ split is of the following form. At the end of block $j$, receiver $k = 1,2$ either recovers 
\[
\mh_1(j-K+1) \to \mh_2(j-K+1-t_1),
\]
for some $t_1 = \min\{K,L\}-1,\ldots, 1,0$, or 
\[
\mh_2(j-L+1) \to \mh_1(j-L+1-t_2), 
\]
for some $t_2 = 0,1,\ldots,\max\{K,L\}-1$. For the 3-2 split in Table~\ref{tab:3-2-one-dim}, there are five feasible decoding orders: 
\begin{align}
 1 &\suchthat \mh_1(j-2) \to \mh_2(j-3) \quad (t_1 = 1)\label{eqn:d-order1}\\
 2 &\suchthat \mh_1(j-2) \to \mh_2(j-2) \quad (t_1 = 0)\label{eqn:d-order2}\\
 3 &\suchthat \mh_1(j-2) \to \mh_2(j-1) \quad (t_2 = 0)\nonumber\\
 4 &\suchthat \mh_2(j-1) \to \mh_1(j-2) \quad (t_2 = 1)\nonumber\\
 5 &\suchthat \mh_2(j-1) \to \mh_1(j-3) \quad (t_2 = 2)\nonumber
\end{align}

In order to write the achievable rate region corresponding to each decoding order, we introduce the notion of \emph{layer order}. Let $\l \suchthat Z_1 \to Z_2 \to \cdots \to Z_{K+L}$ be an ordering of the variables $\{X_1,\ldots, X_K,W_1,\ldots,W_L\}$ such that the relative orders $X_1 \to X_2 \to \cdots \to X_K$ and $W_1 \to W_2 \to \cdots \to W_L$ are preserved. We say that a layer order is \emph{alternating} if it starts with either $X_1 \to \cdots \to X_{a_1}$, $a_1 = \max\{K,L\}-1, \ldots,1,0$, or $W_1 \to \cdots \to W_{a_2}$, $a_2 = 1,2,\ldots,\min\{K,L\}$, followed by one $X$ and one $W$ alternately until one of them is exhausted, and then by the remaining variables. As in the decoding orders, there are $K+L$ alternating layer orders. 
For the 3-2 split in Table~\ref{tab:3-2-one-dim}, the five alternating layer orders are listed as follows
\begin{align}
 1 &\suchthat X_1 \to X_2 \to X_3 \to W_1 \to W_2, \label{eqn:32-order1}\\
 2 &\suchthat X_1 \to X_2 \to W_1 \to X_3 \to W_2, \label{eqn:32-order2}\\
 3 &\suchthat X_1 \to W_1 \to X_2 \to W_2 \to X_3, \nonumber\\
 4 &\suchthat W_1 \to X_1 \to W_2 \to X_2 \to X_3, \nonumber\\
 5 &\suchthat W_1 \to W_2 \to X_1 \to X_2 \to X_3. \nonumber
\end{align}
A layer order indicates which variable (signal layer) is recovered first in successive cancellation decoding. For example, in decoding order $d = 1$ in~\eqref{eqn:d-order1}, $X_1,X_2,X_3$ carrying $m_1(j-2)$ are recovered before $W_1,W_2$ carrying $m_2(j-3)$. In other words, all the $X$ layers are recovered before the $W$ layers in successive cancellation decoding, which corresponds to the layer order $\l = 1$ in~\eqref{eqn:32-order1}. For another example, in decoding order $d = 2$ in~\eqref{eqn:d-order2}, at the end of block $j$, $3 \le j \le b$, we alternately recover $\mh_1(j-2)$ and $\mh_2(j-2)$.
The layers $X_1$ and $X_2$ are recovered before the layer $W_1$, while the layer $X_3$ is recovered after the layer $W_1$, which is followed by the layer $W_2$. This corresponds to the layer order $\l = 2$ in~\eqref{eqn:32-order2}. 

Given a layer order, the achievable rates $R_1$ and $R_2$ are given as sums of the corresponding mutual information terms. For example, for the layer order $\l = 1$ in~\eqref{eqn:32-order1}, the achievable rate region at receiver~$k =1,2$ is the set of rate pairs $(R_1,R_2)$ such that
\begin{equation}
\label{eqn:32-rate1}
\begin{split}
 R_1 &\le I(X_1;Y_k) + I(X_2;Y_k|X_1) + I(X_3;Y_k|X_1,X_2),\\ 
 R_2 &\le I(W_1;Y_k|X_1,X_2,X_3) + I(W_2;Y_k|X_1,X_2,X_3,W_1). 
 \end{split}
\end{equation}
Similarly, for the layer order $\l = 2$ in~\eqref{eqn:32-order2}, the achievable rate region at receiver~$k$ is characterized as
\begin{equation}
\label{eqn:32-rate2}
\begin{split}
 R_1 &\le I(X_1;Y_k) + I(X_2;Y_k|X_1) + I(X_3;Y_k|X_1,X_2,W_1),\\ 
 R_2 &\le I(W_1;Y_k|X_1,X_2) + I(W_2;Y_k|X_1,X_2,X_3,W_1). 
 \end{split}
\end{equation}
Given a layer order $\l\suchthat Z_1\to \cdots \to Z_{K+L}$, define
\begin{equation}\label{eqn:index-set}
\begin{split}
 \mathcal{I}_1 &= \{i\suchthat Z_i \in\{X_1,\ldots,X_K\}\},\\
 \mathcal{I}_2 &= \{i\suchthat Z_i \in\{W_1,\ldots,W_L\}\}.
\end{split}
\end{equation}
Then the achievable rate region 
at receiver~$k$ with corresponding decoding order $d = \l$ is the set of rate pairs $(R_1,R_2)$ such that
\begin{equation}\label{eqn:LO-rate}
\begin{split}
 R_1 &\le \sum_{i\in \mathcal{I}_1} I(Z_i;Y_k|Z^{l-1}),\\
 R_2 &\le \sum_{i\in \mathcal{I}_2} I(Z_i;Y_k|Z^{l-1}).
\end{split}
\end{equation}

\section{Han--Kobayashi Inner Bound}
\label{sec:hk}

The Han--Kobayashi coding scheme~\cite{Han--Kobayashi1981}, illustrated in Fig.~\ref{fig:hk}, is the most powerful among known single-letter coding techniques for the two-user interference channel. In this scheme, rate splitting is used for the messages $M_1 = (M_{10},M_{11})$ and $M_2 = (M_{20},M_{22})$. The messages $M_{10},M_{11},M_{20},M_{22}$ are carried by codewords $s^n,t^n,u^n,v^n$, which are then superimposed into $x^n$ and $w^n$ by symbol-by-symbol mappings $x(s,t)$ and $w(u,v)$. Receiver~1 recovers $\Mh_{10}, \Mh_{20}, \Mh_{11}$ and receiver~2 recovers $\Mh_{10}, \Mh_{20}, \Mh_{22}$  using simultaneous decoding. If we consider $S,T,U,V$ as the channel inputs, the original two-user interference channel can then be viewed as a four-sender two-receiver channel with conditional pmf
\[
 p(y_1,y_2|s,t,u,v) = p(y_1,y_2|x(s,t), w(u,v)).
\]
\begin{figure}[b]
\small
\centering
\def\svgscale{1.4}
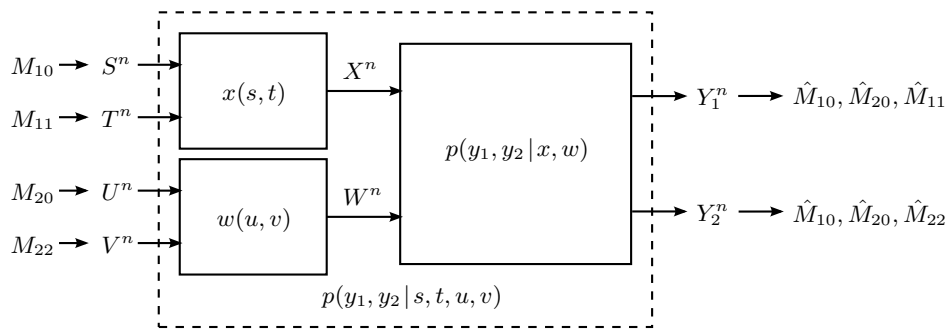
\caption{Han--Kobayashi coding scheme.}
\label{fig:hk}
\end{figure}

For a fixed input pmf $p(s)p(t)p(u)p(v)$ and functions $x(s,t), w(u,v)$, the Han--Kobayashi coding scheme achieves the 4-dimensional auxiliary rate region
\begin{equation}
\label{eqn:hk}
 \Rr_{1,\mathrm{MAC}} \cap \Rr_{2,\mathrm{MAC}}
\end{equation}
where 
\begin{align*}
 \Rr_{1,\mathrm{MAC}} &= \{(R_{10},R_{11},R_{20},R_{22})\suchthat (R_{10},R_{11},R_{20}) \in \Rr_{\mathrm{MAC}}(S,T,U;Y_1)\}\\
 \Rr_{2,\mathrm{MAC}} &= \{(R_{10},R_{11},R_{20},R_{22})\suchthat (R_{10},R_{20},R_{22}) \in \Rr_{\mathrm{MAC}}(S,U,V;Y_2)\},
\end{align*}
and $\Rr_{\mathrm{MAC}}(A,B,C;Y)$ is the standard rate region for a three-user MAC $p(y|a,b,c)$ by random code ensemble $p(a)p(b)p(c)$. Recall that $\Rr_{\mathrm{MAC}}(A,B,C;Y)$ consists of rate triples $(r_1,r_2,r_3)$ such that
\begin{align*}
 r_1 &\le I(A;Y|B,C),\\
 r_2 &\le I(B;Y|A,C),\\
 r_3 &\le I(C;Y|A,B),\\
 r_1+r_2 &\le I(A,B;Y|C),\\
 r_1+r_3 &\le I(A,C;Y|B),\\
 r_2+r_3 &\le I(B,C;Y|A),\\
 r_1+r_2+r_3 &\le I(A,B,C;Y).
\end{align*}
Finally, the Han--Kobayashi inner bound is the union over $p(s)p(t)p(u)p(v)$ and functions $x(s,t),w(u,v)$ of the rate region
\begin{equation}
\label{eqn:hk-inner}
 \text{Proj}_{4\to 2}\bigl(\Rr_{1,\mathrm{MAC}} \cap \Rr_{2,\mathrm{MAC}}\bigr),
\end{equation}
where $\text{Proj}_{4\to 2}$ denotes the projection of the 4-dimensional region of rate quadruples $(R_{10},R_{11},R_{20},R_{22})$ to the 2-dimensional region of rate pairs $(R_1,R_2) = (R_{10}+R_{11}, R_{20}+R_{22})$.

\smallskip

Now we present a coding scheme that achieves the Han--Kobayashi inner bound with single-user decoding by showing the achievability of the 4-dimensional auxiliary region in~\eqref{eqn:hk}. The two common messages $M_{10}$ and $M_{20}$ are transmitted using SWSC, with the 3-1 split in Section~\ref{sec:3-1split}. The two private messages $M_{11}$ and $M_{22}$ are transmitted using the single-block rate-splitting scheme in Section~\ref{sec:rs-one}. The signal $S$ is further split into three layers $S_1, S_2$, and $S_3$. For $j \in [b-2]$, the message $M_{10}(j)$ is carried by $s_3^n, s_2^n$, and $s_1^n$ over blocks $j,j+1$, and $j+2$ respectively. Since the signal $T$ is kept unsplit, the message $M_{20}(j)$ is carried by a single-block code $u^n$ in block $j$. The private messages are further split into two parts $M_{11} = (M_{11}',M_{11}'')$ and $M_{22} = (M_{22}',M_{22}'')$. The four messages $M_{11}',M_{11}'',M_{22}',M_{22}''$ are carried by $t_1^n, t_2^n, v_1^n,v_2^n$, respectively, in a single block. The encoding is illustrated in Table~\ref{tab:hk-31}.


At receiver~1, messages are recovered in the order $d_1$, which is one of the following six (trivial messages at the first and last blocks are skipped):
{\allowdisplaybreaks
\begin{align*}
 1\suchthat & \mh_{11}'(j-1) \to \mh_{10}(j-2) \to \mh_{20}(j-2) \to \mh_{11}''(j-2),\\
 2\suchthat & \mh_{11}'(j) \to \mh_{20}(j) \to \mh_{10}(j-2) \to \mh_{11}''(j),\\
 3\suchthat & \mh_{11}'(j) \to \mh_{10}(j-2) \to \mh_{11}''(j-2) \to \mh_{20}(j),\\
 4\suchthat & \mh_{11}'(j-2) \to \mh_{10}(j-2) \to \mh_{20}(j-2) \to \mh_{11}''(j-2),\\
 5\suchthat & \mh_{11}'(j) \to \mh_{20}(j) \to \mh_{10}(j-2) \to \mh_{11}''(j-1),\\
 6\suchthat & \mh_{11}'(j) \to \mh_{10}(j-2) \to \mh_{11}''(j-2) \to \mh_{20}(j-1).
\end{align*}}%

\begin{table}[t]
\centering
\begin{tabular}{cc@{}c@{}c@{}c@{}c@{}c@{}cc@{}c@{}c@{}c@{}c}
block $\j$  & $1$ && $2$ && $3$ && $4$ & $\cdots$  && $b-1$ && $b$ \\
 \hline\\[-.8em]
  $S_1$ & $1$ && $1$ && $m_{10}(1)$  && $m_{10}(2)$ & $\ldots$ && $\ldots$ &&$m_{10}(b-2)$ \\[-.2em]
  &&  && $\diagup$ && $\diagup$ && &&& $\diagup$ & \\[-.2em]
  $S_2$ & $1$ && $m_{10}(1)$ && $m_{10}(2)$  && $\ldots$ & $\ldots$ && $m_{10}(b-2)$ && $1$ \\[-.2em]
 && $\diagup$ && $\diagup$ && &&& $\diagup$ & && \\[-.2em]
 $S_3$ & $m_{10}(1)$ && $m_{10}(2)$ && $\ldots$  && $\ldots$ & $m_{10}(b-2)$ && $1$ && $1$ \\[.5em]
 $U$ & $m_{20}(1)$ && $m_{20}(2)$ && $\ldots$ && $\ldots$ & $\ldots$ && $\ldots$ && $m_{20}(b)$ \\[.5em]
 $T_1$ & $m_{11}'(1)$ && $m_{11}'(2)$ && $\ldots$ && $\ldots$ & $\ldots$ && $\ldots$ && $m_{11}'(b)$\\[.5em]
 $T_2$ & $m_{11}''(1)$ && $m_{11}''(2)$ && $\ldots$ && $\ldots$ & $\ldots$ && $\ldots$ && $m_{11}''(b)$\\[.5em]
 $V_1$ & $m_{22}'(1)$ && $m_{22}'(2)$ && $\ldots$ && $\ldots$ & $\ldots$ && $\ldots$ && $m_{22}'(b)$\\[.5em]
 $V_2$ & $m_{22}''(1)$ && $m_{22}''(2)$ && $\ldots$ && $\ldots$ & $\ldots$ && $\ldots$ && $m_{22}''(b)$\\[.5em]
\hline\\
\end{tabular}
\caption{A scheme that achieves the Han--Kobayashi inner bound with single-user decoding.}
\label{tab:hk-31}
\end{table}

Fig.~\ref{fig:hk-decoding} illustrates the decoding process for $d_1 = 1$, where $\ast$ indicates messages that were recovered previously. 
\begin{figure*}[htbp]
 \centering
\begin{subfigure}[b]{0.5\textwidth}
\centering
\begin{tabular}{cccc}
block   & $j-2$ & $j-1$ & $j$\\
 \hline\\[-1em]
 $S_1$ & $\ast$ & $\ast$ & $m_{10}(j-2)$ \\[.5em]
 $S_2$ & $\ast$ & $m_{10}(j-2)$  & $m_{10}(j-1)$ \\[.5em]
 $S_3$ & $m_{10}(j-2)$ & $m_{10}(j-1)$  & $m_{10}(j)$ \\[.5em]
 $U$  & $m_{20}(j-2)$ &  $m_{20}(j-1)$  & $m_{20}(j)$ \\[.5em]
 $T_1$ & $\ast$ & $m_{11}'(j-1)$ & $m_{11}'(j)$ \\[.5em]
 $T_2$ & $m_{11}''(j-2)$ & $m_{11}''(j-1)$ & $m_{11}''(j)$  \\[.5em]
\hline
\end{tabular}
\caption{The initial state at the end of block $j$.}
\end{subfigure}%
~
\begin{subfigure}[b]{0.33\textwidth}
\centering
\begin{tabular}{ccc}
  $j-2$ & $j-1$ & $j$\\
 \hline\\[-1em]
 $\ast$ & $\ast$ & $m_{10}(j-2)$ \\[.5em]
 $\ast$ & $m_{10}(j-2)$  & $m_{10}(j-1)$ \\[.5em]
 $m_{10}(j-2)$ & $m_{10}(j-1)$  & $m_{10}(j)$ \\[.5em]
 $m_{20}(j-2)$ &  $m_{20}(j-1)$  & $m_{20}(j)$ \\[.5em]
 $\ast$ & $\ast$ & $m_{11}'(j)$ \\[.5em]
 $m_{11}''(j-2)$ & $m_{11}''(j-1)$ & $m_{11}''(j)$  \\[.5em]
\hline
\end{tabular}
\caption{Step 1: recover $\mh_{11}'(j-1)$.}
\end{subfigure}\\[1em]

\begin{subfigure}[b]{0.5\textwidth}
\centering
\begin{tabular}{cccc}
block   & $j-2$ & $j-1$ & $j$\\
 \hline\\[-1em]
 $S_1$ & $\ast$ & $\ast$ & $\ast$ \\[.5em]
 $S_2$ & $\ast$ & $\ast$  & $m_{10}(j-1)$ \\[.5em]
 $S_3$ & $\ast$ & $m_{10}(j-1)$  & $m_{10}(j)$ \\[.5em]
 $U$  & $m_{20}(j-2)$ &  $m_{20}(j-1)$  & $m_{20}(j)$ \\[.5em]
 $T_1$ & $\ast$ & $\ast$ & $m_{11}'(j)$ \\[.5em]
 $T_2$ & $m_{11}''(j-2)$ & $m_{11}''(j-1)$ & $m_{11}''(j)$  \\[.5em]
\hline
\end{tabular}
\caption{Step 2: recover $\mh_{10}(j-2)$ over blocks $j-2,j-1$, and $j$.}
\end{subfigure}%
~
\begin{subfigure}[b]{0.33\textwidth}
 \centering
\begin{tabular}{ccc}
 $j-2$ & $j-1$ & $j$\\
 \hline\\[-1em]
 $\ast$ & $\ast$ & $\ast$ \\[.5em]
 $\ast$ & $\ast$  & $m_{10}(j-1)$ \\[.5em]
 $\ast$ & $m_{10}(j-1)$  & $m_{10}(j)$ \\[.5em]
 $\ast$ &  $m_{20}(j-1)$  & $m_{20}(j)$ \\[.5em]
 $\ast$ & $\ast$ & $m_{11}'(j)$ \\[.5em]
 $m_{11}''(j-2)$ & $m_{11}''(j-1)$ & $m_{11}''(j)$  \\[.5em]
\hline
\end{tabular}
\caption{Step 3: recover $\mh_{20}(j-2)$.}
\end{subfigure}\\[1em]

\begin{subfigure}[b]{0.5\textwidth}
 \centering
 \begin{tabular}{@{\quad}c@{\quad\qquad}c@{\qquad\quad}c@{\quad\quad}c@{\quad}}
block   & $j-2$ & $j-1$ & $j$\\
 \hline\\[-1em]
 $S_1$ & $\ast$ & $\ast$ & $\ast$ \\[.5em]
 $S_2$ & $\ast$ & $\ast$  & $m_{10}(j-1)$ \\[.5em]
 $S_3$ & $\ast$ & $m_{10}(j-1)$  & $m_{10}(j)$ \\[.5em]
 $U$  & $\ast$ &  $m_{20}(j-1)$  & $m_{20}(j)$ \\[.5em]
 $T_1$ & $\ast$ & $\ast$ & $m_{11}'(j)$ \\[.5em]
 $T_2$ & $\ast$ & $m_{11}''(j-1)$ & $m_{11}''(j)$  \\[.5em]
\hline
\end{tabular}
\caption{Step 4: recover $\mh_{11}''(j-2)$.}
\end{subfigure}
\caption{Illustration of the decoding process for $d_1 =1$.}
\label{fig:hk-decoding}
\end{figure*}
By the standard analysis, the achievable rate region for this decoding order is the set of rate quadruples $(R_{10},R_{11},R_{20},R_{22})$ such that
\begin{equation}
\label{eqn:hk-region1}
\begin{split}
 R_{10} &\le I(S_1;Y_1) + I(S_2;Y_1|S_1,T_1) + I(S_3;Y_1|S_1,T_1,S_2),\\
 R_{20} &\le I(U;Y_1|S_1,T_1,S_2,S_3),\\
 R_{11} &\le I(T_1;Y_1|S_1) + I(T_2;Y_1|S_1,T_1,S_2,S_3,U),
\end{split}
\end{equation}
which is exactly the rate region corresponding to the layer order $\l_1$ 
\[
  1\suchthat S_1 \to T_1 \to S_2 \to S_3 \to U \to T_2. 
\]
One can similarly verify that the layer orders $\l_1$ corresponding to decoding orders $d_1= 2, \ldots, 6$ are 
\begin{align*}
 2\suchthat & T_1 \to U \to S_1 \to T_2 \to S_2 \to S_3,\\
 3\suchthat & T_1 \to S_1 \to U \to S_2 \to S_3 \to T_2,\\
 4\suchthat & S_1 \to S_2 \to T_1 \to S_3 \to U \to T_2,\\
 5\suchthat & T_1 \to U \to S_1 \to S_2 \to T_2 \to S_3,\\
 6\suchthat & T_1 \to S_1 \to S_2 \to U \to S_3 \to T_2.
\end{align*}
At receiver~2, the messages are recovered in the order $d_2$, which is one of the following six: 
{\allowdisplaybreaks\begin{align*}
 7\suchthat & \mh_{22}'(j-1) \to \mh_{10}(j-2) \to \mh_{20}(j-2) \to \mh_{22}''(j-2),\\
 8\suchthat & \mh_{22}'(j) \to \mh_{20}(j) \to \mh_{10}(j-2) \to \mh_{22}''(j),\\
 9\suchthat & \mh_{22}'(j) \to \mh_{10}(j-2) \to \mh_{22}''(j-2) \to \mh_{20}(j),\\
 10\suchthat & \mh_{22}'(j-2) \to \mh_{10}(j-2) \to \mh_{20}(j-2) \to \mh_{22}''(j-2),\\
 11\suchthat & \mh_{22}'(j) \to \mh_{20}(j) \to \mh_{10}(j-2) \to \mh_{22}''(j-1),\\
 12\suchthat & \mh_{22}'(j) \to \mh_{10}(j-2) \to \mh_{22}''(j-2) \to \mh_{20}(j-1),
\end{align*}}%
with corresponding achievable layer orders $\l_2$
\begin{align*}
 7\suchthat & S_1 \to V_1 \to S_2 \to S_3 \to U \to V_2,\\
 8\suchthat & V_1 \to U \to S_1 \to V_2 \to S_2 \to S_3,\\
 9\suchthat & V_1 \to S_1 \to U \to S_2 \to S_3 \to V_2,\\
 10\suchthat & S_1 \to S_2 \to V_1 \to S_3 \to U \to V_2,\\
 11\suchthat & V_1 \to U \to S_1 \to S_2 \to V_2 \to S_3,\\
 12\suchthat & V_1 \to S_1 \to S_2 \to U \to S_3 \to V_2.
\end{align*}
Let $p'$ be the pmf $p'(s_1)p'(s_2)p'(s_3)p'(t_1)p'(t_2)p'(u)p'(v_1)p'(v_2)$ along with $s(s_1,s_2,s_3),t(t_1,t_2)$, and $v(v_1,v_2)$. Let $\Rr_1(p',\l_1)$ be the rate region corresponding to the layer order $\l_1 = 1,\ldots,6$ at receiver~1. For example, $\Rr_1(p',1)$ is the set of rate quadruples $(R_{10},R_{11},R_{20},R_{22})$ in~\eqref{eqn:hk-region1}. Similarly let $\Rr_2(p',\l_2)$ be the rate region corresponding to the layer order $\l_2 = 7,\ldots,12$ at receiver~2. 
This SWSC scheme achieves $\Rr_1(p',\l_1) \cap \Rr_2(p',\l_2)$ for any $\l_1 = 1,\ldots,6$ and $\l_2 = 7,\ldots,12$, which is sufficient to achieve the 4-dimensional auxiliary region in~\eqref{eqn:hk}; see Appendix~\ref{appx:hk} for the proof. 

%
%
%

\begin{theorem}
 \label{thm:hk}
Let $p$ denote the pmf $p(s)p(t)p(u)p(v)$ along with functions $x(s,t)$ and $w(u,v)$. Then
 \[
  \Rr_{1,\mathrm{MAC}}(p) \cap \Rr_{2,\mathrm{MAC}}(p) = \bigcup_{p'\simeq p} \bigcup_{\l_1=1}^6 \bigcup_{\l_2=7}^{12} [\Rr_1(p',\l_1) \cap \Rr_2(p',\l_2)].
  \]
\end{theorem}

Consequently, taking the union over all pmfs $p(s)p(t)p(u)p(v)$ and functions $x(s,t), w(u,v)$, the coding scheme in Table~\ref{tab:hk-31} achieves the Han--Kobayashi inner bound~\eqref{eqn:hk-inner} for the two-user interference channel $p(y_1,y_2|x,w)$.

\section{Concluding Remarks}
\label{sec:rmk}
In this paper, we proposed the sliding-window superposition coding scheme (SWSC) as an implementable alternative to  the rate-optimal simultaneous decoding. Combined with the conventional rate-splitting technique, the coding scheme can be generalized to achieve the Han--Kobayashi inner bound on the capacity region of the two-user interference channel. Since the publication of the initial work~\cite{Wang--Sasoglu--Kim2014} on SWSC, extensive simulations of the SWSC scheme have been performed in more practical communication scenarios, such as the Ped-B fading interference channel model~\cite{Kim--Ahn--Kim--Park--Wang--Chen--Park2015, Kim--Ahn--Kim--Park--Chen--Kim2016}. With several improvements in transceiver design, such as soft decoding, input bit-mapping and layer optimization, and power control, the performance figures presented here can be improved by another 10--20\%~\cite{Kim--Ahn--Kim--Park--Wang--Chen--Park2015}. System-level performance as well as requirements on the network operation for implementing SWSC in 5G cellular networks are discussed in~\cite{Kim--Ahn--Kim--Park--Chen--Kim2016}. These results indicate that SWSC is a promising candidate for interference management in future cellular networks.

On the theory side, the SWSC scheme can be further extended to support more senders and receivers; a more complete theory on this topic will be reported elsewhere. Here we present a new ``dimension'' of the problem to illustrate the richness of potential extensions. Consider the SWSC scheme with a 2-2 split, as defined in Section~\ref{sec:layering-order}. This scheme has 4 possible layer orders:
\begin{equation}
 \begin{split}
 \label{eqn:22so}
 1 &\suchthat X_1\to X_2 \to W_1 \to W_2,\\
 2 &\suchthat X_1\to W_1 \to X_2 \to W_2,\\
 3 &\suchthat W_1\to X_1 \to W_2 \to X_2,\\
 4 &\suchthat W_1\to W_2 \to X_1 \to X_2.
 \end{split}
\end{equation}
As in the 2-1 split case, this scheme is not sufficient to achieve the MLD region $\Rr^*$ in general. There are two additional \emph{nonalternating} layer orders:
\begin{align*}
 5 &\suchthat X_1\to W_1 \to W_2 \to X_2,\\
 6 &\suchthat W_1\to X_2 \to X_2 \to W_2,
\end{align*}
which also preserve the relative orders $X_1 \to X_2$ and $W_1 \to W_2$ but these layer orders do not admit corresponding decoding orders. 

It turns out all six layer orders can be achieved if the messages are scheduled in two dimensions. Instead of communicating the messages over $b$ consecutive blocks (in a single dimension), one can communicate $b_1(b_2-1)$ messages $M_1(jk),j\in[b_1],k\in[b_2-1]$, and $(b_1-1)b_2$ messages $M_2(jk),j\in[b_1-1],k\in[b_2]$, over $b_1b_2$ blocks (in two dimensions). Fig.~\ref{tab:2dswsc} illustrates the message scheduling for $b_1 = b_2 = 4$.
\begin{figure}[htbp]
\centering
\small
\begin{tabular}{c|c|c|c|c|}
$X_1,\;W_1$  &  & & &   \\
$X_2,\;W_2$  & $i_2 = 1$ & $2$ & $3$  & $i_2 = 4$\\\hline
 &  & & &   \\[-.8em]
\multirow{2}{*}{$i_1 = 1$} & $1, \qquad 1$ & $\quad 1, \quad m_2(11)$ & $\quad 1, \quad 
m_2(12)$ & $\quad 1,\quad m_2(13)$ \\[.5em]
 & $m_1(11), m_2(11)$ & $m_1(12), m_2(12)$ & $m_1(13), m_2(13)$ & $m_1(14),\quad 1$ 
\\[.5em]\hline
  &  & & &  \\[-.8em]
\multirow{2}{*}{$2$} & $m_1(11),\quad 1$  & $m_1(12), m_2(21)$ & $m_1(13), m_2(22)$ 
  & $m_1(1,4), m_2(23)$ \\[.5em]
 & $m_1(21), m_2(21)$ & $m_1(22), m_2(22)$ & $m_1(23), m_2(23)$ & $m_2(1,b_2),\quad 
1$\\[.5em] \hline
 &  & & & \\[-.8em]
\multirow{2}{*}{$3$} & $m_1(21), \quad 1$  & $m_1(22), m_2(31)$ & $m_1(23), m_2(32)$
& $m_1(24),m_2(33)$ \\[.5em]
 & $m_1(31), m_2(31)$ & $m_1(32), m_2(32)$ & $m_1(33), m_2(33)$ & 
$m_1(34), m_2(34)$\\[.5em] \hline
   &  & & & \\[-.8em]
\multirow{2}{*}{$i_1 = 4$} & $m_1(31), \quad 1$  & $m_1(32), m_2(41)$ & $m_1(33), 
m_2(42)$ & $m_1(34),m_2(43)$ \\[.5em]
 & $\quad 1,\quad m_2(41)$ & $\quad 1,\quad m_2(42)$ & $\quad 1,\quad m_2(43)$ & 
$1,\qquad 1$\\[.5em] \hline
 \multicolumn{5}{c}{}\\
\end{tabular}
\caption{Message scheduling for the two-dimensional SWSC.}
\label{tab:2dswsc}
\end{figure}
This two-dimensional SWSC scheme with the 2-2 split has symmetric encoding structure for both users. With properly chosen successive cancellation decoding, any layer order is feasible, which is sufficient to achieve the MLD region $\Rr^*$~\cite[Section 4.4.1]{Wang2015}. With further augmentation, this alternative SWSC scheme can also achieve the Han--Kobayashi inner bound~\cite[Section 4.4.2]{Wang2015}.

\appendices

\section{Proof of Lemma~\ref{lem:layer-splitting}}
\label{appx:layer-splitting}

First, for any rate pair $(R_1,R_2)$ in~\eqref{eqn:rs-mac}, we have
{\allowdisplaybreaks\begin{align}
 R_1 &\le I(X_1;Y_1) + I(X;Y_1|W,X_1)\nonumber\\*
 &\stackrel{(a)}{\le} I(X_1;Y_1|W) + I(X;Y_1|W,X_1)\nonumber\\*
 &\stackrel{(b)}{=} I(X;Y_1|W),\nonumber\\
 R_2 &\le I(W;Y_1|X_1)\nonumber\\*
 &\stackrel{(c)}{\le} I(W;Y_1|X_1,X_2)\nonumber\\*
 &\stackrel{(d)}= I(W;Y_1|X),\nonumber\\
 R_1 + R_2 &\le I(X_1;Y_1) + I(X;Y_1|W,X_1) + I(W;Y_1|X_1)\nonumber\\*
 &= I(X,W;Y_1), \label{eqn:equal-sum}
\end{align}}%
where $(a)$ and $(c)$ follow since $W$ is independent of $(X_1,X_2)$, and $(b)$ and $(d)$ follow since $X_1 \to X \to (W,Y_1)$ form a Markov chain. Thus, any rate point in $\Rr_{\mathrm{RS}}(p')$ with $p'\simeq p$ is also in $\Rr_{1,\mathrm{SD}}(p)$. 

Now it suffices to show that for any rate point $(I_1,I_2)$ on the \emph{dominant face}, i.e., $I_1+I_2 = I(X,W;Y_1)$, there exists a $p' \simeq p$ such that
 \begin{align*}
  I_1 &= I(X_1;Y_1) + I(X;Y_1|X_1,W),\\
  I_2 &= I(W;Y_1|X_1).
 \end{align*}
 To this end, note that when $X_1 = X$ and $X_2 = \emptyset$, expression~\eqref{eqn:rs-mac} attains one corner point $(I(X;Y_1), I(W;Y_1|X))$; when $X_1 = \emptyset$ and $X_2 = X$, expression~\eqref{eqn:rs-mac} attains the other corner point $(I(X;Y_1|W), I(W;Y_1))$. Moreover, the rate pair in~\eqref{eqn:rs-mac} and $(I_1,I_2)$ share the same sum-rate as in~\eqref{eqn:equal-sum}.
 Hence, it suffices to show that for every $\a \in [0,1]$, there exists a choice of $p(x_1)p(x_2)$ and function $x(x_1,x_2)$ such that
 \[
  I(W;Y_1|X_1) = I_2 = \a I(W;Y_1) + (1-\a) I(W;Y_1|X).
 \]
 
Let $p_{X_1}(x) = (1-\a)p_X(x)$ for $x \in \Xc$ and $p_{X_1}(\error) = \a$. Let $X_2$ be independent of $X_1$ and $p_{X_2}(x) = p_X(x)$ for $x \in \Xc$. Let
\[
 x(x_1,x_2) = \begin{cases}
 x_1, &\text{ if } x_1 \neq \error,\\
 x_2, &\text{ otherwise}.
\end{cases}
\]
This choice of $p(x_1)p(x_2)$ and $x(x_1,x_2)$ induces a conditional pmf
\begin{equation}
 \label{eqn:u-split}
  p_{X_1|X}(x_1|x) = \begin{cases}
                1-\a, & \text{ if } x_1 = x,\\
                \a, & \text{ if } x_1 = \error,\\
                0, & \text{ otherwise,}
               \end{cases}
 \end{equation}
which is an erasure channel with input $X$, output $X_1$, and erasure probability $\a$. Define $E = \mathbbm{1}_{\{X_1 = \error\}}$. It can be checked that $E \sim \Bern(\a)$ is independent of $X$ and $X_2$. Thus, we have
 \begin{align*}
  I(W;Y_1|X_1) & = I(W;Y_1|X_1,E) \\
  &= \a I(W;Y_1|X_1,E=1) + (1-\a) I(W;Y_1|X_1,E=0)\\
  &\stackrel{(a)}{=} \a I(W;Y_1|X_1 = \error, E=1) + (1-\a) I(W;Y_1|X,X_1,E=0)\\
  &\stackrel{(b)}{=} \a I(W;Y_1|X_1 = \error, E=1) + (1-\a) I(W;Y_1|X)\\
  &\stackrel{(c)}{=} \a I(W;Y_1|E=1) + (1-\a) I(W;Y_1|X)\\
  &\stackrel{(d)}{=} \a I(W;Y_1) + (1-\a) I(W;Y_1|X), 
 \end{align*}
 where $(a)$ follows since when $E = 0$, $X_1 = X$, $(b)$ follows since given $X$, $(W,Y_1)$ are conditionally independent of $(X_1,E)$, $(c)$ follows since $E =1$ is equivalent as $X_1 = \error$, and $(d)$ follows since $E$ is independent of $(X,W,Y_1)$. Therefore, as $\a$ increases from $0$ to $1$, the rate pair in~\eqref{eqn:rs-mac} moves continuously and linearly from one corner point to the other along the line $R_1 + R_2 = I(X,W;Y_1)$.

 
 \section{Proof of Proposition~\ref{pro:31swsc}}
 \label{appx:31swsc}
 By Lemma~\ref{lem:layer-splitting}, $\Rr_{1,\text{SD}} \cap \Rr_{2,\text{SD}}$ is equivalent to the set of $(R_1,R_2)$ such that
 \begin{equation}\label{eqn:fact3}
 \begin{split}
  R_1 &\le \min\{I(X_1';Y_1)+I(X;Y_1|X_1',W),\, I(X_1'';Y_2)+I(X;Y_2|X_1'',W)\}\\
  R_2 &\le \min\{I(W;Y_1|X_1'), I(W;Y_2|X_1'')\}
 \end{split}
 \end{equation}
for erasure channels $p(x_1'|x)$ and $p(x_1''|x)$ with erasure probabilities $\a'$ and $\a''$ respectively. 
Suppose that $\a' > \a''$. Then the channel $p(x_1'|x)$ is \emph{degraded} with respect to the channel $p(x_1''|x)$. Since the rate expressions in~\eqref{eqn:fact3} only depend on the marginal conditional pmfs of $p(x_1',x_1''|x)$, we assume without loss of generality that $X \to X_1'' \to X_1'$ form a Markov chain. By the functional representation lemma (twice), for $p(x_1'|x_1'')$, there exists an $X_2$ independent of $X_1'$ such that $X_1'' = f(X_1',X_2)$; for $p(x_1''|x)$, there exists an $X_3$ independent of $(X_2,X_1')$ such that $X = g(X_1'',X_3) = g(f(X_1',X_2),X_3) \triangleq x(X_1',X_2,X_3)$. Renaming $X_1 \triangleq X_1'$ and plugging $X_1'' = f(X_1,X_2)$ into~\eqref{eqn:fact3}, we obtain the rate region $\Rr_{\mathrm{SWSC}}(p',3,1,d_1,d_2)$ with $(d_1,d_2)$ given in~\eqref{eqn:31split-d1}. In the case when $\a' \le \a''$, we can assume that $X \to X_1' \to X_1''$ form a Markov chain. Then, using the functional representation similarly as above, the rate region in~\eqref{eqn:fact3} can be reduced to the rate region $\Rr_{\mathrm{SWSC}}(p',3,1,d_1,d_2)$ with $(d_1,d_2)$ given in~\eqref{eqn:31split-d2}.

\section{Proof of Theorem~\ref{thm:insufficiency}}
\label{appx:insufficiency}
We provide an example of strict inclusion between the two regions in the symmetric Gaussian interference channel (cf.~\eqref{eq:channel}) with $g_{11} = g_{22} = 1, g_{12} = g_{21} = g, S_1 = S_2 = S = P$ and $I_1 = I_2 = I = g^2P$. Assume that the interference channel has \emph{strong, but not very strong,} interference, i.e., $S < I < S(S+1)$. The capacity region of this channel is characterized by the set of rate pairs $(R_1, R_2)$ such that
\begin{align*}
R_1 &\le \C(S),\\
R_2 &\le \C(S),\\
R_1 + R_2 &\le \C(I+S),
\end{align*}%
which is achieved by simultaneous decoding with a single input distribution $X_1,X_2 \sim N(0,P)$~\cite{Sato1978b,Han--Kobayashi1981}. 

Given $\Rr(p,s,t,d_1,d_2)$, let $\Rr^*(s,t,d_1,d_2)$ be the closure of the union of $\Rr(p,s,t,d_1,d_2)$ over all $p$. Define 
\[
R_1^*(s,t,d_1,d_2) = \max\{R_1\suchthat (R_1,\C(S)) \in \Rr^*(s,t,d_1,d_2)\}
\]
as the maximal achievable rate $R_1$ such that $R_2$ is at individual capacity. In order to show the corner point of the capacity region is not
achievable using any $(p,s,t,d_1,d_2)$ rate-splitting scheme, it suffices to establish the following.

\smallskip
\begin{proposition}
\label{prop:gaussian}
For the symmetric Gaussian interference channel with $S < I < S(S+1)$,
\[
R_1^*(s,t,d_1,d_2) < {\textsf{\upshape C}}\Big(\frac{I}{1+S}\Big)
\]
for any finite $s,t$ and decoding orders $d_1,d_2$.
\end{proposition}

The remainder of this appendix is dedicated to the proof of Proposition~\ref{prop:gaussian}. First, we find the optimal decoding order at receiver~2 of the $(p,s,t,d_1,d_2)$ rate-splitting scheme that achieves $R_1^*(s,t,d_1,d_2)$. We note that in homogeneous superposition coding, message parts are encoded into independent codewords, and thus can be recovered in an arbitrary order in general (which is in sharp contrast to heterogeneous superposition coding, where $\mh_{ij}$ has to be recovered before $\mh_{ik}$ for $j < k, i = 1,2$). Henceforth, by renaming the message parts, we assume without loss of generality that at receiver~2,  the decoding order among message parts $\{\mh_{11},\ldots,\mh_{1s}\}$ is $\mh_{11} \to \mh_{12} \to \cdots \to \mh_{1s}$ and the decoding order among messages parts $\{\mh_{21},\ldots,\mh_{2t}\}$ is $\mh_{21} \to \mh_{22} \to \cdots \to \mh_{2t}$. Note that between message parts of $m_1$ and $m_2$, there are still flexibility for all possible permutations as long as the subsets $\{\mh_{11},\ldots,\mh_{1s}\}$ and $\{\mh_{21},\ldots,\mh_{2t}\}$ are in order. The next lemma states the optimal order among them.

\begin{lemma}
 \label{lem:d2}
 For any $(p,s,t,d_1,d_2)$ rate-splitting scheme that achieves $R_1^*(s,t,d_1,d_2)$, the decoding order at receiver~2 is
\[
 d_2^* \suchthat \mh_{11} \to \mh_{12} \to \cdots \to \mh_{1s} \to \mh_{21} \to \mh_{22} \to \cdots \to \mh_{2t}.
\]
\end{lemma}
\begin{IEEEproof} 
Fix any $(p,s,t,d_1,d_2)$ rate-splitting scheme that guarantees $R_2 =\textsf{C}(S)$. Suppose that $\mh_{2j}$ is recovered earlier than $\mh_{1k}$ at receiver~2, that is,
\[
d_2 \suchthat d_{21} \to \mh_{2j} \to \mh_{1k} \to d_{22}.
\] 
Now flip the decoding order of $\mh_{2j}$ and $\mh_{1k}$ in $\dt_2$ as
\[
\dt_2 \suchthat d_{21} \to \mh_{1k} \to \mh_{2j} \to d_{22}
\]
and construct $(p,s,t,d_1,\dt_2)$ rate-splitting scheme, where the message splitting, the underlying distribution, and decoding order $d_1$ remain the
same. Let $\Rt_{ij}$ be the rate of the message part $m_{ij}$ in the $(p,s,t,d_1,\dt_2)$ rate-splitting scheme. Then we have that all the rates remain the same except
\begin{align*}
R_{2j} &= I(W_{j};Y_2|W^{j-1},X^{k-1}),\\
\Rt_{2j} &= I(W_{j};Y_2|W^{j-1},X^{k}),\\
R_{1k} &= I(X_{k};Y_2|W^{j},X^{k-1}),\\
\Rt_{1k} &= I(X_{k};Y_2|W^{j-1},X^{k-1}).
\end{align*}
Note that $R_{2j} \le \Rt_{2j}$ since $X_k$ is independent of $(W^j,X^{k-1})$. On the other hand, since $R_{2j}$ already results in full rate at $R_2$, we must
have $\Rt_{2j} = R_{2j}$. It follows that
\[
I(X_{k};W_{j}|Y_2,W^{j-1},X^{k-1}) = 0
\]
and therefore $R_{1j} = \Rt_{1j}$. 
\end{IEEEproof}

\smallskip
Next, we discuss the structure of the decoding order at receiver~1.
\begin{lemma}
 \label{lem:d1}
 In order to show the insufficiency in achieving the corner point $(\textsf{\upshape C}(I/(1+S)),\textsf{\upshape C}(S))$ for any $(p,s,t,d_1,d_2)$ rate-splitting scheme, it suffices to show the insufficiency of any $(p,s,s,d_1^*,d_2^*)$ rate-splitting scheme with decoding orders
 \begin{equation}
 \label{eq:d1d2}
 \begin{split}
 d_1^* &\suchthat \mh_{1,\pi(1)} \to \mh_{2,\sigma(1)} \to \mh_{1,\pi(2)} \to \mh_{2,\sigma(2)} \to \cdots \to \mh_{1,\pi(s-1)} \to \mh_{2,\sigma(s-1)} \to \mh_{1,\pi(s)},\\
 d_2^* &\suchthat \mh_{11} \to \mh_{12} \to \cdots \to \mh_{1s} \to \mh_{21} \to \mh_{22} \to \cdots \to \mh_{2s},
\end{split}
\end{equation}
where $\pi \suchthat [s] \to [s], \sigma\suchthat [s] \to [s]$ are two permutations on the index set $[s]$.
\end{lemma}

\begin{IEEEproof}
First, there is no loss of generality in assuming $s = t$, because any $(p,s,t,d_1,d_2)$ scheme can be viewed as a special case of some $(p',\max\{s,t\},\max\{s,t\},d_1',d_2')$ scheme by nulling out the corresponding inactive variables and preserving the distribution and decoding orders of the active ones. For the alternating decoding order at receiver~1, we note that a $(p,s,s,d_1,d_2^*)$ scheme with \emph{arbitrary} decoding order $d_1$ can viewed as a special case of some $(\tilde{p},2s,2s,\tilde{d}_1^*,\tilde{d}_2^*)$ scheme with \emph{alternating} decoding order $\tilde{d}_1^*$. For example, for $s = 2$, any decoding order must be one of the following six forms
{\allowdisplaybreaks\begin{align*}
 &\mh_{1,\pi(1)} \to \mh_{1,\pi(2)} \to \mh_{2,\sigma(1)} \to \mh_{2,\sigma(2)},\\
 &\mh_{1,\pi(1)} \to \mh_{2,\sigma(1)} \to \mh_{1,\pi(2)} \to \mh_{2,\sigma(2)},\\
 &\mh_{2,\sigma(1)} \to \mh_{1,\pi(1)} \to \mh_{2,\sigma(2)} \to \mh_{1,\pi(1)},\\
 &\mh_{2,\sigma(1)} \to \mh_{2,\sigma(2)} \to \mh_{1,\pi(1)} \to \mh_{1,\pi(2)},\\
 &\mh_{1,\pi(1)} \to \mh_{2,\sigma(1)} \to \mh_{2,\sigma(2)} \to \mh_{1,\pi(2)},\\
 &\mh_{2,\sigma(1)} \to \mh_{1,\pi(1)} \to \mh_{1,\pi(2)} \to \mh_{2,\sigma(2)},
\end{align*}}%
which are all special cases of the alternating order
\[
 \mh_{1,\tilde{\pi}(1)} \to \mh_{2,\tilde{\sigma}(1)} \to \mh_{1,\tilde{\pi}(2)} \to \mh_{2,\tilde{\sigma}(2)} \to \mh_{1,\tilde{\pi}(3)} \to \mh_{2,\tilde{\sigma}(3)} \to \mh_{1,\tilde{\pi}(4)}.
\]
Moreover, because of the special structure of $\tilde{d}_2^*$, it remains the optimal decoding order (in the sense of Lemma~\ref{lem:d2}) even after nulling out the inactive message parts.
\end{IEEEproof}

\smallskip
Now, we provide a necessary condition for a rate-splitting scheme to achieve the corner point of the capacity region.

\begin{lemma}
\label{lem:gaussian}
If a $(p,s,t,d_1,d_2^*)$ rate-splitting scheme attains the corner point $(\textsf{\upshape C}(I/(1+S)),\textsf{\upshape C}(S))$, then $p$ must satisfy 
\[
X \sim \N(0,P)\quad\text{ and }\quad W \sim \N(0,P). 
\]
\end{lemma}

\begin{IEEEproof}
No matter what $d_1$ is, because of the optimal order $d_2^*$, the rate constraints for $R_2$ must satisfy
\begin{align}
 R_2 &\le \sum_{j = 1}^t I(W_j;Y_2|X,W^{j-1}) \nonumber\\
 & = I(W;Y_2|X) \nonumber\\
 &\le \C(S). \label{eq:R2}
\end{align}
Given $X$, the channel from $W$ to $Y_2$ is a Gaussian channel with SNR $S$. Therefore the condition $W\sim \N(0,P)$ is necessary for~\eqref{eq:R2} to
hold with equality. Similarly, $R_1$ must satisfy
\begin{align}
 R_1 &\le \sum_{j=1}^s I(X_j;Y_2|X^{j-1})\nonumber\\
 &= I(X;Y_2)\nonumber\\
 &\le \C \left(\frac{I}{1+S}\right). \label{eq:R1}
\end{align}
Given $W\sim \N(0,P)$, the channel from $X$ to $Y_2$
is a Gaussian channel with SNR $I/(1+S)$. Therefore, the condition $X\sim
\N(0,P)$ is necessary for~\eqref{eq:R1} to hold with equality.
\end{IEEEproof}

\smallskip
We also need the following technical lemma.

\begin{lemma}
\label{lem:independence}
Let $F(u,x)$ be any (continuous) distribution such that $X\sim \N(0,P)$ and $I(U;Y) = 0$, where $Y=X+N$ with $N \sim \N(0,1)$ independent of $X$. Then, $I(U;X) = 0$.
\end{lemma}

\begin{IEEEproof}
For every $u \in \Uc$, we have
\begin{align*}
I(X;Y|U=u) &= h(Y|U=u) - h(Y|X,U=u)\\
&\stackrel{(a)}{=} h(Y) - h(Y|X) \\
&= \textsf{C}(P),
\end{align*}
where $(a)$ follows since $Y$ is independent of $U$ and $U \to X \to Y$ form a Markov chain. Suppose for some $u$, $\E(X^2|U=u)<P$, i.e., the effective
channel SNR is strictly less than $P$. Then $I(X;Y|U=u) < P$. As a result, we must have $\E(X^2|U=u)\ge P$ for all $u \in \Uc$. On the other hand, 
\begin{align*}
P &\le \int \E(X^2|U=u) dF(u)\\
&= \E(X^2)\\
&= P,
\end{align*}
which forces $\E(X^2|U=u)=P$ for almost all $u$. Since the Gaussian input $\N(0,P)$ is the unique distribution that attains the rate $\textsf{C}(P)$ in the Gaussian channel with SNR $P$, the distribution $F(x|u)$ must be $\N(0,P)$ for almost all $u$. Therefore $I(U;X) = 0$. 
\end{IEEEproof}

\smallskip

We are ready to establish the suboptimality of rate-splitting schemes.

\smallskip

By Lemma~\ref{lem:d1}, it suffices to show the insufficiency of any $(p,s,s,d_1^*,d_2^*)$ rate-splitting scheme with decoding orders given in~\eqref{eq:d1d2}. The achievable rate region of this scheme is characterized by
\begin{align*}
 R_1 &\le \sum_{i=1}^s \min\{I(X_{\pi(i)};Y_1|X_{\pi(1)},\ldots,X_{\pi(i-1)},W_{\sigma(1)},\ldots,W_{\sigma(i-1)}),\, I(X_{\pi(i)};Y_2|X^{\pi(i)-1})\} :=I_1\\
 R_2 &\le \sum_{i=1}^{s-1} \min\{I(W_{\sigma(i)};Y_1|X_{\pi(1)},\ldots,X_{\pi(i)},W_{\sigma(1)},\ldots,W_{\sigma(i-1)}),\, I(W_{\sigma(i)};Y_2|X,W^{\sigma(i)-1})\} \\
 &\hspace{2em} + I(W_{\sigma(s)};Y_2|X,W_{\sigma(1)},\ldots,W_{\sigma(s-1)}) :=I_2
\end{align*}
Assume that the corner point of the capacity region is achieved by this scheme, i.e.,
\begin{align}
I_1 &= \textsf{C}(I/(1+S)),  \label{eq:R1corner}\\
I_2 &= \textsf{C}(S).\label{eq:R2corner}
\end{align}%
Then, by Lemma~\ref{lem:gaussian}, we must have $X \sim \N(0,P)$ and $W \sim \N(0,P)$. Consider
\begin{align}
 I_1 &\le I(X_{\pi(1)};Y_1) + \sum_{i \in [s]\setminus \pi(1)} I(X_i;Y_2|X^{i-1})\nonumber\\
 &\le I(X_{\pi(1)};Y_1) + I(X^{\pi(1)-1};Y_2) + I(X_{\pi(1)+1}^s;Y_2|X^{\pi(1)})\nonumber\\
 &\stackrel{(a)}{\le} I(X_{\pi(1)};Y_1) + I(X^{\pi(1)-1};Y_2|X_{\pi(1)}) + I(X_{\pi(1)+1}^s;Y_2|X^{\pi(1)})\nonumber\\
 &= h(Y_1) - h(Y_1|X_{\pi(1)}) + h(Y_2|X_{\pi(1)}) -h(Y_2|X^{\pi(1)}) + h(Y_2|X^{\pi(1)}) - h(Y_2|X) \nonumber\\
 &= h(Y_1) - h(Y_1|X_{\pi(1)}) + h(Y_2'|X_{\pi(1)}) - h(Y_2'|X) \label{eq:case5entropyI1}
\end{align}
where $Y_2' = Y_2/g = X + (W+N_2)/g$ and $(a)$ follows since $X_{\pi(1)}$ is independent of $X^{\pi(1)-1}$. Since 
\begin{align*}
\frac{1}{2}\log(2\pi e(S+1)/g^2) &= h(Y_2'|X)\\
&\le h(Y_2'|X_{\pi(1)})\\
&\le h(Y_2') 
= \frac{1}{2} \log(2\pi e(I + S + 1)/g^2),
\end{align*} 
there exists an $\a \in [0,1]$ such that $h(Y_2'|X_{\pi(1)}) = (1/2) \log(2\pi e$ $(\a I + S + 1)/g^2)$. Moreover, since $W \sim \N(0,P)$ and $I < S(1+S)$, the channel $X \to Y_1$ is a degraded version of the channel $X \to Y_2'$, i.e., $Y_1 =Y_2' + N'$, where $N' \sim \N(0,I+1-(S+1)/g^2)$ is independent of $X$ and $W$. By the entropy power inequality, 
\begin{align*}
2^{2h(Y_1|X_{\pi(1)})} &\ge 2^{2h(Y_2'|X_{\pi(1)})} + 2^{2h(N'|X_{\pi(1)})}\\
&= 2\pi e (\a S + I + 1).
\end{align*}
Therefore, it follows from~\eqref{eq:case5entropyI1} that 
\begin{align*}
I_1 &\le h(Y_1) - h(Y_1|X_{\pi(1)}) + h(Y_2'|X_{\pi(1)}) - h(Y_2'|X)\\
&\le \frac{1}{2}\log\left(\frac{(I+S+1)(\a I + S + 1)}{(\a S + I +
1)(1+S)}\right)\\
&\le \textsf{C}(I/(1+S)),
\end{align*}%
where the last step follows since $S < I$. To match the standing assumption in~\eqref{eq:R1corner}, we must have equality in $(a)$, which forces $\a = 1$ and $h(Y_2'|X_{\pi(1)}) = (1/2) \log(2\pi e( I + S + 1)/g^2) = h(Y_2')$,  i.e., $I(X_{\pi(1)};Y_2') = 0$. Note that $X, W \sim \N(0,P)$ and the channel from $X$ to $Y_2'$ is a Gaussian channel. Applying Lemma~\ref{lem:independence} yields 
\begin{equation}
\label{eq:pi1}
I(X_{\pi(1)};X) = 0.
\end{equation}

Now, $I_2$ can be simplified to 
\begin{align}
 I_2 &\le I(W_{\sigma(1)};Y_1|X_{\pi(1)}) + \sum_{i\in[s]\setminus \sigma(1)} I(W_i;Y_2|X,W^{i-1}) \nonumber\\
 &\stackrel{(b)}{=} I(W_{\sigma(1)};Y_1) + I(W^{\sigma(1)-1};Y_2|X) + I(W_{\sigma(1)+1}^s;Y_2|X,W^{\sigma(1)}) \nonumber\\
 &\stackrel{(c)}{\le} I(W_{\sigma(1)};Y_1) + I(W^{\sigma(1)-1};Y_2|X,W_{\sigma(1)}) + I(W_{\sigma(1)+1}^s;Y_2|X,W^{\sigma(1)}) \nonumber\\
 &= h(Y_1) -h(Y_1|W_{\sigma(1)}) + h(Y_2|X,W_{\sigma(1)}) -h(Y_2|X,W^{\sigma(1)}) + h(Y_2|X,W^{\sigma(1)}) -h(Y_2|X,W) \nonumber\\
 &= h(\Yt_1) -h(\Yt_1|W_{\sigma(1)}) + h(\Yt_2|W_{\sigma(1)}) -h(\Yt_2|W) \label{eq:case5entropyI2}
\end{align}
where $\Yt_1 = Y_1/g = W+(X+N_1)/g$ and $\Yt_2 = W+N_2$. Here $(b)$ follows since $I(X_{\pi(1)};Y_1|W_{\sigma(1)}) \le I(X_{\pi(1)};Y_1|W) = I(X_{\pi(1)};X+N_1) \le I(X_{\pi(1)};X) = 0$ and $I(X_{\pi(1)};Y_1) \le I(X_{\pi(1)};X) = 0$, which implies 
\begin{align*}
I(W_{\sigma(1)};Y_1|X_{\pi(1)}) &= I(W_{\sigma(1)};Y_1|X_{\pi(1)}) + I(X_{\pi(1)};Y_1)\\
&= I(W_{\sigma(1)};Y_1) + I(X_{\pi(1)};Y_1|W_{\sigma(1)})\\
&= I(W_{\sigma(1)};Y_1), 
\end{align*}
and $(c)$ follows since $W_{\sigma(1)}$ and $(W^{\sigma(1)-1},X)$ are independent. Since
\begin{align*}
\frac{1}{2}\log(2\pi e) &= h(\Yt_2|W)\\
 &\le h(\Yt_2|W_{\sigma(1)}) \\
 &\le h(\Yt_2)\\
 & = \frac{1}{2}\log(2\pi e(1+S)),
\end{align*}
there exists a $\b \in [0,1]$ such that $h(\Yt_2|W_{\sigma(1)}) = (1/2)\log(2\pi e$ $(1+\b S))$. Moreover, since $X \sim \N(0,P)$ and $I < S(1+S)$, $\Yt_1$ is a degraded version of $\Yt_2$, i.e., $\Yt_1 = \Yt_2 + \Nt$, where $\Nt \sim \N(0, (1+S)/g^2-1)$ is independent of $X$ and $W$. Applying the entropy power
inequality, we have
\begin{align*}
2^{2h(\Yt_1|W_{\sigma(1)})} &\ge 2^{2h(\Yt_2|W_{\sigma(1)})} + 2^{2h(\Nt|W_{\sigma(1)})}\\
&= 2\pi e (\b S + (1+S)/g^2).
\end{align*}
Therefore, it follows from~\eqref{eq:case5entropyI2} that
\begin{align*}
I_2 &\le h(\Yt_1)-h(\Yt_1|W_{\sigma(1)}) + h(\Yt_2|W_{\sigma(1)}) - h(Y_2|X,W)\\
&\le \frac{1}{2} \log \left(\frac{(I+S+1)(1+\b S)}{g^2(\b S +
(1+S)/g^2)}\right)
\le \textsf{C}(S),
\end{align*}
where the last step follows from the channel condition $I <(1+S)S$. To match the standing assumption in~\eqref{eq:R2corner}, we must have equality above, which forces $\b = 1$ and $h(\Yt_2|W_{\sigma(1)}) = (1/2)\log(2\pi e(1+ S)) = h(\Yt_2)$, i.e., $I(W_{\sigma(1)};\Yt_2) = 0$. Note that $W \sim \N(0,P)$ and the channel from $W$ to $\Yt_2$ is a Gaussian channel. Applying Lemma~\ref{lem:independence} yields 
\begin{equation}
\label{eq:sigma1}
I(W_{\sigma(1)};W) = 0.
\end{equation}

To continue analyzing the dependency between $(X_{\pi(1)},X_{\pi(2)})$ and $X$, we note that condition~\eqref{eq:sigma1} implies that
\begin{align}
 I(X_{\pi(2)};Y_1|X_{\pi(1)},W_{\sigma(1)}) &= I(X_{\pi(2)};Y_1|X_{\pi(1)}) +I(W_{\sigma(1)};Y_1|X_{\pi(1)},X_{\pi(2)}) - I(W_{\sigma(1)};Y_1|X_{\pi(1)}) \nonumber\\
 &\stackrel{(d)}{=} I(X_{\pi(2)};Y_1|X_{\pi(1)}), \label{eq:I3-1}
\end{align}
where $(d)$ follows since $I(W_{\sigma(1)};Y_1|X_{\pi(1)}) \le I(W_{\sigma(1)};Y_1|X_{\pi(1)},X_{\pi(2)}) \le I(W_{\sigma(1)};Y_1|X) \le I(W_{\sigma(1)};W) = 0$. Moreover, condition~\eqref{eq:pi1} implies that
\begin{align}
 I(X_{\pi(1)};Y_2|X_{\pi(2)}) &\le I(X_{\pi(1)};Y_2,X_{\pi(2)}) \nonumber\\
 &\le I(X_{\pi(1)};Y_2,X) \nonumber\\
 &= I(X_{\pi(1)};X,gX+W+N_2) 
=0 \nonumber
\end{align}
and thus
\begin{equation}
\label{eq:I3-2}
 h(Y_2|X_{\pi(2)}) = h(Y_2|X_{\pi(2)},X_{\pi(1)}).
\end{equation}
With~\eqref{eq:I3-1} and~\eqref{eq:I3-2}, we can bound $I_1$ alternatively as
\begin{align}
 I_1 &\le I(X_{\pi(2)};Y_1|X_{\pi(1)},W_{\sigma(1)}) + \sum_{i \in [s]\setminus \pi(2)} I(X_i;Y_2|X^{i-1})\nonumber\\
 &= I(X_{\pi(2)};Y_1|X_{\pi(1)}) + I(X^{\pi(2)-1};Y_2) + I(X_{\pi(2)+1}^s;Y_2|X^{\pi(2)})\nonumber\\
 &\le I(X_{\pi(2)};Y_1|X_{\pi(1)}) + I(X^{\pi(2)-1};Y_2|X_{\pi(2)}) + I(X_{\pi(2)+1}^s;Y_2|X^{\pi(2)})\nonumber\\
 &= h(Y_1|X_{\pi(1)}) - h(Y_1|X_{\pi(2)},X_{\pi(1)}) + h(Y_2|X_{\pi(2)}) - h(Y_2|X) \nonumber\\
 &= h(Y_1) - h(Y_1|X_{\pi(2)},X_{\pi(1)}) + h(Y_2|X_{\pi(2)},X_{\pi(1)}) - h(Y_2|X) \nonumber\\
 &= h(Y_1) - h(Y_1|X_{\pi(2)},X_{\pi(1)}) + h(Y_2'|X_{\pi(2)},X_{\pi(1)}) - h(Y_2'|X).\label{eq:case5entropyI3}
\end{align}
Note that the expression in~\eqref{eq:case5entropyI3} is in the same form of~\eqref{eq:case5entropyI1}, except that $X_{\pi(1)}$ is replaced by the pair $(X_{\pi(1)},X_{\pi(2)})$. From this point on, following the identical argument as before with variable substitution $X_{\pi(1)} \leftrightarrow (X_{\pi(1)},X_{\pi(2)})$, we conclude that
\[
I(X_{\pi(1)},X_{\pi(2)};X) = 0. 
\]
Now, repeating this procedure, we can similarly show that 
\begin{equation}
\label{eq:s-indep}
\begin{split}
 I(X_{\pi(1)},\ldots,X_{\pi(s-1)};X) = 0,\\
 I(W_{\sigma(1)},\ldots,W_{\sigma(s-1)};W) = 0.
\end{split}
\end{equation}

However, condition~\eqref{eq:s-indep} implies that for $i \in [s-1]$
\begin{align*}
 I(X_{\pi(i)};Y_1|X_{\pi(1)},\ldots,X_{\pi(i-1)},W_{\sigma(1)},\ldots,W_{\sigma(i-1)})
 &\le I(X_{\pi(i)};X,W,Y_1) 
 = I(X_{\pi(i)};X)
 = 0
\end{align*}
and that
\begin{align*}
 &I(X_{\pi(s)};Y_1|X_{\pi(1)},\ldots,X_{\pi(s-1)},W_{\sigma(1)},\ldots,W_{\sigma(s-1)}) - I(X;Y_1)\\
 &= I(W_{\sigma(1)},\ldots,W_{\sigma(s-1)};Y_1|X) - I(X_{\pi(1)},\ldots,X_{\pi(s-1)};Y_1) - I(W_{\sigma(1)},\ldots,W_{\sigma(s-1)};Y_1|X_{\pi(1)},\ldots,X_{\pi(s-1)})\\
 &\stackrel{(e)}{=}0,
\end{align*}
where $(e)$ follows since $I(W_{\sigma(1)},\ldots,W_{\sigma(s-1)};Y_1|X_{\pi(1)},\ldots,X_{\pi(s-1)}) \le I(W_{\sigma(1)},\ldots,W_{\sigma(s-1)};Y_1|X)$\\ $\le I(W_{\sigma(1)},\ldots,W_{\sigma(s-1)};W) = 0$ and $I(X_{\pi(1)},\ldots,X_{\pi(s-1)};Y_1) \le I(X_{\pi(1)},\ldots,X_{\pi(s-1)};X) = 0$. Therefore,
\begin{align*}
 I_1 &= \sum_{i=1}^s \min\{I(X_{\pi(i)};Y_1|X_{\pi(1)},\ldots,X_{\pi(i-1)},W_{\sigma(1)},\ldots,W_{\sigma(i-1)}),\, I(X_{\pi(i)};Y_2|X^{\pi(i)-1})\}\\
 &= \min\{I(X_{\pi(s)};Y_1|X_{\pi(1)},\ldots,X_{\pi(s-1)},W_{\sigma(1)},\ldots,W_{\sigma(s-1)}),\, I(X_{\pi(s)};Y_2|X^{\pi(s)-1})\}\\
 &\le I(X_{\pi(s)};Y_1|X_{\pi(1)},\ldots,X_{\pi(s-1)},W_{\sigma(1)},\ldots,W_{\sigma(s-1)})\\
 &= I(X;Y_1)\\
 &= \C(S/(1+I))\\
 &< \C(I/(S+I)),
\end{align*}
which is a contradiction to the standing assumption in~\eqref{eq:R1corner}. This completes the proof of Proposition~\ref{prop:gaussian}.

\section{The SWSC scheme in Table~\ref{tab:sw-rs}}
\label{appx:2-1swsc}

{\it Codebook generation.} Fix a pmf $p'(x_1)p'(x_2)p'(w)$ and a function $x(x_1,x_2)$. Randomly and independently generate a codebook for each block. For notational convention, we assume $m_{1}(0) = m_{1}(b) = 1$. For $\j \in [b]$, randomly and 
independently generate $2^{nR_1}$ sequences $x_1^n(m_{1}(\j -1)), m_{1}(\j -1) \in [2^{nR_1}]$, each according to $\prod_{i=1}^n p'_{X_1}(x_{1i})$. For $\j \in [b]$, randomly and independently generate $2^{nR_1}$ sequences $x_2^n(m_1(\j)), m_1(\j) \in [2^{nR_1}]$, each according to $\prod_{i=1}^n p'_{X_2}(x_{2i})$. For $\j \in [b]$, randomly and independently generate $2^{nR_2}$ sequences $w^n(m_{2}(\j)), m_{2}(\j) \in 
[2^{nR_2}]$, each according to $\prod_{i=1}^n p'_{W}(w_{i})$. This defines the codebook
\begin{align*}
 \Cc_\j = \bigl\{& x_1^n(m_{1}(\j -1)), x_2^n(m_{1}(\j)), w^n(m_{2}(\j))\suchthat  m_{1}(\j -1), 
m_{1}(\j) \in [2^{nR_1}], m_{2}(\j) \in [2^{nR_2}]\bigr\}, \quad \j \in [b].
\end{align*}

\smallskip

{\it Encoding.} In block $j \in [b]$, sender~1 transmits $x_{i}(x_{1i}(m_1(\j -1)),x_{2i}(m_1(\j)))$ at time $i \in [n]$ and sender~2 transmits
$w^n(m_{2}(\j))$. Table~\ref{tab:sw-rs} reveals the scheduling of the messages.

\smallskip

{\it Decoding.}  
Let the received sequences in block $\j$ be $y_1^n(\j)$ and $y_2^n(\j)$, $\j \in [b]$. For receiver~1, at the end of block~1, it finds the unique message $\mh_{2}(1)$ such that 
$$(w^n(\mh_{2}(1)),y_1^n(1),x_1^n(m_1(0))) \in \aep.$$ 
At the end of block $\j$, $2 \le \j\le b$, it finds the unique message $\mh_{1}(\j-1)$ 
such that 
$$(x_1^n(\mh_{1}(\j-2)), x_2^n(\mh_{1}(\j-1)), w^n(\mh_{2}(\j-1)), y_1^n(\j-1)) \in \aep$$
and 
$$(x_1^n(\mh_{1}(\j-1)), y_1^n(\j)) \in \aep$$
simultaneously. Then it finds the unique $\mh_{2}(\j)$
 such that
$$(w^n(\mh_{2}(\j)), y_1^n(\j), x_1^n(\mh_{1}(\j-1))) \in
\aep.$$ 
If any of the typicality checks fails, it declares an error.

We analyze the probability of decoding error averaged over codebooks. Assume  without loss of generality that $M_{1}(\j) = M_{2}(\j) = 1$. We divide the error events as follows
\begin{align*}
 &\Ec_{11}(\j-2) = \{\Mh_{1}(\j-2) \neq 1\},\\
 &\Ec_{12}(\j-1) = \{\Mh_{2}(\j-1) \neq 1\},\\
 &\Ec_{13}(\j-1) = \{(X_1^n(\Mh_{1}(\j-2)), X^n(1),W^n(\Mh_{2}(\j-1)), Y_1^n(\j-1))\not\in \aep \text{ or } (X_1^n(1), Y_1^n(\j)) \not\in \aep\},\\
 &\Ec_{14}(\j-1) = \{(X_1^n(\Mh_{1}(\j-2)), X^n(m_{1}(\j-1)), W^n(\Mh_{2}(\j-1)),Y_1^n(\j-1))\in \aep  \\ 
 &\hspace{4em} \text{ and } (X_1^n(m_{1}(\j-1)), Y_1^n(\j))\in \aep \text{ for some } m_{1}(\j-1) \neq 1\},\\
 &\Ec_{15}(\j) = \{(W^n(1), Y_1^n(\j), X_1^n(\Mh_{1}(\j-1))) \not\in \aep\},\\
 &\Ec_{16}(\j) = \{(W^n(m_{2}(\j)), Y_1^n(\j), X_1^n(\Mh_{1}(\j-1))) \in \aep \text{ for some } m_{2}(\j) \neq 1\}.
\end{align*}
We analyze by induction. By assumption $\Ec_{11}(0) =\emptyset$. Thus in block~1, the probability of error is upper bounded as 
\begin{align*}
 \P\{\Mh_{2}(1)\neq 1\} & = \P(\Ec_{12}(1)) \le \P(\Ec_{15}(1)) + \P(\Ec_{16}(1)).
\end{align*}
Now by the law of large numbers, $\P(\Ec_{15}(1)) \to 0$ as $n \to \infty$. By the packing lemma, $\P(\Ec_{16}(1)) \to 0$ as $n \to \infty$ if $R_2 < I(W;Y_1|X_1) - \d(\e)$. Now assume that the probability of error $\P(\Ec_{11}(\j-2)\cup \Ec_{12}(\j-1))$ in block~$\j-1$ tends to zero as $n \to
\infty$. In block $\j$, the probability of error is upper bounded as
\begin{align*}
 &\P\{(\Mh_{1}(\j-1),\Mh_{2}(\j)) \neq (1,1)\}\\
 &\le \P(\Ec_{11}(\j-2) \cup \Ec_{12}(\j-1) \cup \Ec_{11}(\j-1) \cup \Ec_{12}(\j))\\
 &\le \P(\Ec_{11}(\j-2) \cup \Ec_{12}(\j-1)) + \P(\Ec_{11}(\j-1)\cap \Ec_{11}^c(\j-2)
\cap \Ec_{12}^c(\j-1)) + \P(\Ec_{12}(\j) \cap \Ec_{11}^c(\j-1))\\
 &\le \P(\Ec_{11}(\j-2) \cup \Ec_{12}(\j-1)) + \P(\Ec_{13}(\j-1)\cap \Ec_{11}^c(\j-2)
\cap \Ec_{12}^c(\j-1)) + \P(\Ec_{14}(\j-1)\cap \Ec_{11}^c(\j-2) \cap \Ec_{12}^c(\j-1))\\
 &\hspace{2em} + \P(\Ec_{15}(\j) \cap \Ec_{11}^c(\j-1)) + \P(\Ec_{16}(\j) \cap
\Ec_{11}^c(\j-1)).
\end{align*}
By the induction assumption, the first term tends to zero as $n \to \infty$. By the independence of the codebooks, the law of large numbers, and the packing lemma, the second, fourth, and fifth terms tend to zero as $n \to \infty$ if $R_2 < I(W;Y_1|X_1) - \d(\e)$. The third term $\P(\Ec_{14}(\j-1)\cap \Ec_{11}^c(\j-2)\cap \Ec_{12}^c(\j-1))$ requires a special care. We have
\begin{align*}
 &\P(\Ec_{14}(\j-1)\cap \Ec_{11}^c(\j-2) \cap \Ec_{12}^c(\j-1))\\
 & = \P\{(X_1^n(1), X^n(m_{1}(\j-1)), W^n(1), Y_1^n(\j-1))\in \aep \text{ and }\\
 &\hspace{2em}(X_1^n(m_{1}(\j-1)), Y_1^n(\j)) \in \aep \text{ for some } m_{1}(\j-1) \neq 1\}\\
 & = \sum_{m_{1}(\j-1) \neq 1} \P\{(X_1^n(1), X^n(m_{1}(\j-1)), W^n(1), Y_1^n(\j-1)\in \aep \text{ and } (X_1^n(m_{1}(\j-1)), Y_1^n(\j)) \in \aep\}\\
 & \stackrel{(a)}{=} \sum_{m_{1}(\j-1) \neq 1} \P\{(X_1^n(1), X^n(m_{1}(\j-1)), W^n(1), Y_1^n(\j-1)\in \aep\}  \cdot \P\{(X_1^n(m_{1}(\j-1)), Y_1^n(\j)) \in \aep\}\\
 & \stackrel{(b)}{\le} 2^{nR_1} 2^{-n(I(X;Y_1|W,X_1)-\d(\e))} 2^{-n(I(X_1;Y_1)-\d(\e))},
\end{align*}
which tends to zero if $R_1 < I(X_1;Y_1) + I(X;Y_1|W,X_1) -2\d(\e)$. Here $(a)$ follows since, by the independence of the codebooks, the events
\[
 \{(X_1^n(1), X^n(m_{1}(\j-1)), W^n(1), Y_1^n(\j-1))\in \aep\}
\]
and 
\[
 \{(X_1^n(m_{1}(\j-1)), Y_1^n(\j)) \in \aep\}
\]
are independent for each $m_{1}(\j-1) \neq 1$, and $(b)$ follows by the independence of the codebooks and the joint typicality lemma.


For receiver~2, at the end of block $j$, $2 \le \j \le b$, it finds the unique $\mh_{1}(\j-1)$ such that 
$$(x_1^n(\mh_{1}(\j-2)), x_2^n(\mh_{1}(\j-1)),y_2^n(\j-1)) \in \aep$$
and 
$$(x_1^n(\mh_{1}(\j-1)), y_2^n(\j)) \in \aep$$
simultaneously.  Then it finds the unique $\mh_{2}(\j-1)$ such that 
$$(w^n(\mh_{2}(\j-1)), y_2^n(\j-1), x_1^n(\mh_{1}(\j-2)), x_2^n(\mh_{1}(\j-1))) \in \aep.$$ 
In the end, receiver~2 finds the unique $\mh_{2}(b)$ such that 
$$(w^n(\mh_{2}(b)), y_2^n(b), x_1^n(\mh_{1}(b-1)), x_2^n(m_1(b))) \in \aep.$$  
If any of the typicality checks fails, it declares an error. 
With a similar analysis as above, the decoding is successful if
\begin{align*}
 R_1 &< I(X_1;Y_2) + I(X_2;Y_2|X_1)-2\d(\e) = I(X;Y_2) - 2\d(\e),\\
 R_2 &< I(W;Y_2|X) -\d(\e).
\end{align*}

\section{Proof of Theorem~\ref{thm:hk}}
\label{appx:hk}

We first extend the layer-splitting lemma (Lemma~\ref{lem:layer-splitting}) to the three-user case and show that by splitting two inputs into two layers each and keeping one input unsplit, any rate triple on the dominant face of $\Rr_{\mathrm{MAC}}(A,B,C;Y)$ is achievable by successive cancellation decoding. 

The dominant face of $\Rr_{\mathrm{MAC}}(A,B,C;Y)$ is illustrated in Figure~\ref{fig:3mac}. We label the six corner points by $\Iv_{ABC},\Iv_{BAC},\Iv_{BCA},$ $\Iv_{CBA},\Iv_{CAB},\Iv_{ACB}$, corresponding to the following six rate vectors
\begin{align*}
 \Iv_{ABC} &= (I(A;Y),I(B;Y|A),I(C;Y|A,B)),\\
 \Iv_{BAC} &= (I(A;Y|B),I(B;Y),I(C;Y|B,A)),\\
 \Iv_{BCA} &= (I(A;Y|B,C),I(B;Y),I(C;Y|B)),\\
 \Iv_{CBA} &= (I(A;Y|C,B),I(B;Y|C),I(C;Y)),\\
 \Iv_{CAB} &= (I(A;Y|C),I(B;Y|C,A),I(C;Y)),\\
 \Iv_{ACB} &= (I(A;Y),I(B;Y|A,C),I(C;Y|A)).
\end{align*}
We partition this hexagon region into three subregions: two triangles $\triangle (\Iv_{ACB},\Iv_{ABC},\Iv_{BAC})$ and $\triangle (\Iv_{BCA}, \Iv_{CBA}, \Iv_{CAB})$, and a trapezoid $\Box (\Iv_{ACB}, \Iv_{BAC}, \Iv_{BCA}, \Iv_{CAB})$. In order to achieve each region by successive cancellation decoding, we split $A$ and $B$ into $(A_1,A_2)$ and $(B_1,B_2)$ respectively. In other words, we consider $p'$ of the form $p'(a_1)p'(a_2)p'(b_1)p'(b_2)$ $p'(c)$ and functions $a(a_1,a_2)$ and $b(b_1,b_2)$ such that $p'\simeq p(a)p(b)p(c)$. Let $\Rr(p',\l)$, $\l = 1,2,3$, be the set of achievable rate triples $(R_1,R_2,R_3)$ associated with the following layer orders (defined in a similar manner as in Section~\ref{sec:layering-order})
\begin{subequations}
\label{eqn:3mac-orders}
\begin{align}
 1&\suchthat A_1 \to B_1 \to A_2 \to C \to B_2,\\
 2&\suchthat B_1 \to C \to A_1 \to B_2 \to A_2,\\
 3&\suchthat B_1 \to A_1 \to C \to A_2 \to B_2.
\end{align}
\end{subequations}
For example, $\Rr(p',1)$ is the set of rate triples $(r_1,r_2,r_3)$ such that
\begin{equation}
\label{eqn:3mac-rs}
\begin{split}
 r_1 &\le I(A_1;Y) + I(A_2;Y|A_1,B_1) \\
 r_2 &\le I(B_1;Y|A_1) + I(B_2;Y|A_1,B_1,A_2,C)\\
 r_3 &\le I(C;Y|A_1,B_1,A_2).
\end{split}
\end{equation}

We need to show for every point in $\Rr_{\mathrm{MAC}}(A,B,C;Y)$, there exists some choice of $p'$ that achieves it. Similar to Lemma~\ref{lem:layer-splitting}, we choose the conditional pmfs $p(a_1|a)$ and $p(b_1|b)$ as erasure channels with erasure probabilities $\a$ and $\b$ respectively.
Then the rate expressions~\eqref{eqn:3mac-rs} can be further simplified as
\begin{align*}
 r_1 &= (1-\a)(1-\b) I(A;Y) + \a(1-\b) I(A;Y|B) + \b I(A;Y),\\
 r_2 &= (1-\a)(1-\b) I(B;Y|A) + \a(1-\b) I(B;Y) + \b I(B;Y|A,C),\\
 r_3 &= (1-\a)(1-\b) I(C;Y|A,B) + \a(1-\b) I(C;Y|B,A) + \b I(C;Y|A).
\end{align*}
In other words, letting $\rv :=(r_1,r_2,r_3)$, the achievable rate region $\Rr(\a,\b,\l)$ for $\l = 1$ is the set of rate vectors $\rv$ such that
\[
\rv \le (1-\a)(1-\b)\Iv_{ABC} + \a(1-\b)\Iv_{BAC} + \b \Iv_{ACB}.
\]
This region covers every point in the triangle $\triangle (\Iv_{ACB},\Iv_{ABC},\Iv_{BAC})$ by varying $\a,\b \in [0,1]$. Similarly, the achievable rate region $\Rr(\a,\b,\l)$ for $\l = 2$ is the set of rate vectors $\rv$ such that
\[
\rv \le (1-\a)\b \Iv_{CAB} + \a\b \Iv_{CBA} + (1-\b) \Iv_{BCA}.
\]
This region covers every point in the triangle $\triangle (\Iv_{BCA}, \Iv_{CBA}, \Iv_{CAB})$ by varying $\a,\b \in [0,1]$.
For layer order $\l = 3$, the achievable rate region $\Rr(\a,\b,\l)$ is given by
\[
 \rv \le (1-\a)(1-\b) \Iv_{BAC} + (1-\a)\b \Iv_{ACB} + \a(1-\b) \Iv_{BCA} + \a\b \Iv_{CAB}.
\]
Note that for each fixed $\b$, the trajectory of the achievable rate points when varying $\a$ from 0 to 1 is a line segment that is parallel to the two sides $(\Iv_{ACB},\Iv_{CAB})$ and $(\Iv_{BAC},\Iv_{BCA})$. By further varying $\b$ from 0 to 1, this layer order achieves every point in the trapezoid $\Box (\Iv_{ACB}, \Iv_{BAC}, \Iv_{BCA}, \Iv_{CAB})$.

\begin{figure}[t]
\centering
\footnotesize
 \psfrag{r1}[bc]{$R_1$}
 \psfrag{r2}[cl]{$R_2$}
 \psfrag{r3}[cl]{$R_3$}
 \psfrag{abc}{$ABC$}
 \psfrag{acb}{$ACB$}
 \psfrag{bac}{$BAC$}
 \psfrag{bca}{$BCA$}
 \psfrag{cab}{$CAB$}
 \psfrag{cba}{$CBA$}
 \includegraphics[scale = 0.7]{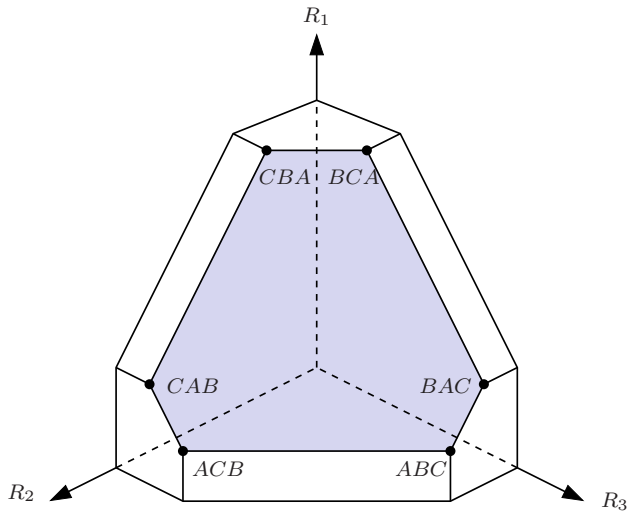}
 \caption{Achievable rate region of the three-user MAC $p(y|a,b,c)$.}
 \label{fig:3mac}
\end{figure}

\begin{lemma}[Layer splitting for a 3-user MAC~\cite{Grant--Rimoldi--Urbanke--Whiting2001}]
 \label{lem:layer-splitting3}
 For a 3-user MAC $p(y|a,b,c)$, the achievable rate region for input pmf $p=p(a)p(b)p(c)$ can be equivalently expressed as
 \[
  \Rr_{\mathrm{MAC}}(A,B,C;Y) = \bigcup_{p'\simeq p} \bigcup_{\l = 1}^3 \Rr(p',\l).
 \]
where the layer orders $\l = 1,2,3$ are given in~\eqref{eqn:3mac-orders}.
Moreover, let $p(a_1|a)$ and $p(b_1|b)$ be two erasure channels with erasure probabilities $\a$ and $\b$ respectively. Then,
\begin{equation}
\label{eqn:erasure-split}
 \Rr_{\mathrm{MAC}}(A,B,C;Y) = \bigcup_{\a \in [0,1],\b \in [0,1]} \bigcup_{\l=1}^3 \Rr(\a,\b,\l).
\end{equation}
\end{lemma}

In order to express the 4-dimensional auxiliary region~\eqref{eqn:hk}, we split $S$ into three layers $(S_1,S_2,S_3)$ and $T,V$ into two layers each $(T_1,T_2)$ and $(V_1,V_2)$. At receiver~1, we consider layer orders $\l_1$ given by
\begin{align*}
 1\suchthat & S_1 \to T_1 \to S_2 \to S_3 \to U \to T_2,\\
 2\suchthat & T_1 \to U \to S_1 \to T_2 \to S_2 \to S_3,\\
 3\suchthat & T_1 \to S_1 \to U \to S_2 \to S_3 \to T_2,\\
 4\suchthat & S_1 \to S_2 \to T_1 \to S_3 \to U \to T_2,\\
 5\suchthat & T_1 \to U \to S_1 \to S_2 \to T_2 \to S_3,\\
 6\suchthat & T_1 \to S_1 \to S_2 \to U \to S_3 \to T_2.
\end{align*}
Let $p'$ be the pmf $p(s_1)p(s_2)p(s_3)p(t_1)p(t_2)p(u)p(v_1)p(v_2)$ along with $s(s_1,s_2,s_3),t(t_1,t_2),v(v_1,v_2)$. Let $\Rr_1(p',\l_1)$ be the achievable rate region at receiver~1 for the layer order $\l_1 \in [6]$. For example, $\Rr_1(p',1)$ is the set of rate quadruples $(R_{10},R_{11},R_{20},R_{22})$ such that 
\begin{align*}
 R_{10} &\le I(S_1;Y_1) + I(S_2;Y_1|S_1,T_1) + I(S_3;Y_1|S_1,T_1,S_2),\\
 R_{20} &\le I(U;Y_1|S_1,T_1,S_2,S_3),\\
 R_{11} &\le I(T_1;Y_1|S_1) + I(T_2;Y_1|S_1,T_1,S_2,S_3,U).
\end{align*}
At receiver~2, we consider layer orders $\l_2= 7,8,\ldots,12$, which are obtained from $\l_1 = 1,2,\ldots,6$, respectively, by replacing $T_1$ by $V_1$ and $T_2$ by $V_2$. For example, the layer order $\l_2 = 7$ is obtained from the layer order $\l_1 = 1$ as
\[
 7 \suchthat S_1 \to V_1 \to S_2 \to S_3 \to U \to V_2.
\]
Let $\Rr_2(p',\l_2)$ be the achievable rate region at receiver~2 for the layer order $\l_2 = 7,8,\ldots, 12$. For example, $\Rr_2(p',7)$ is the set of rate quadruples $(R_{10},R_{11},R_{20},R_{22})$ such that
\begin{align*}
 R_{10} &\le I(S_1;Y_2) + I(S_2;Y_2|S_1,V_1)+I(S_3;Y_2|S_1,V_1,S_2),\\
 R_{20} &\le I(U;Y_2|S_1,V_1,S_2,S_3),\\
 R_{22} &\le I(V_1;Y_2|S_1) + I(V_2;Y_2|S_1,V_1,S_2,S_3,U).
\end{align*}

\begin{lemma}
 \label{lem:hk-split}
 Let $p$ denote the pmf $p(s)p(t)p(u)p(v)$ along with functions $x(s,t)$ and $w(u,v)$. Then
 \begin{multline}
 \label{eqn:hk-swsc}
  \Rr_{1,\mathrm{MAC}}(p) \cap \Rr_{2,\mathrm{MAC}}(p) \\
  = \cup_{p'\simeq p} \bigl[\bigl(\cup_{\l_1=1}^3 \cup_{\l_2=10}^{12} [\Rr_1(p',\l_1) \cap \Rr_2(p',\l_2)]\bigr) \cup \bigl(\cup_{\l_1=4}^6 \cup_{\l_2=7}^{9} [\Rr_1(p',\l_1) \cap \Rr_2(p',\l_2)]\bigr)\bigr].
 \end{multline}
\end{lemma}

\begin{proof}
 By Lemma~\ref{lem:layer-splitting3}, the target rate region can be equivalently expressed as
 \begin{align*}
  \Rr_{1,\mathrm{MAC}}(p) \cap \Rr_{2,\mathrm{MAC}}(p)
  & =  \left(\cup_{\a',\b \in [0,1]} \cup_{\tilde{\l}_1=1}^3 \Rr_1(\a',\b, \tilde{\l}_1)\right) \cap \left( \cup_{\a'',\g \in [0,1]}\cup_{\tilde{\l}_2=4}^6\Rr_2(\a'',\g,\tilde{\l}_2)\right)\\
  & = \cup_{\a',\a'',\b,\g \in [0,1]}\cup_{\tilde{\l}_1=1}^3 \cup_{\tilde{\l}_2=4}^6 [\Rr_1(\a',\b,\tilde{\l}_1) \cap \Rr_2(\a'',\g,\tilde{\l}_2)]
 \end{align*}
for some erasure channels $p(s_1'|s),p(s_1''|s),p(t_1|t),p(v_1|v)$ with erasure probabilities $\a',\a'',\b,\g$, respectively, and the layer orders are 
{\allowdisplaybreaks\begin{align*}
  1\suchthat & S_1' \to T_1 \to S_2' \to U \to T_2,\\
  2\suchthat & T_1 \to U \to S_1' \to T_2 \to S_2',\\
  3\suchthat & T_1 \to S_1' \to U \to S_2' \to T_2,\\
  4\suchthat & S_1'' \to V_1 \to S_2'' \to U \to V_2,\\
  5\suchthat & V_1 \to U \to S_1'' \to V_2 \to S_2'',\\
  6\suchthat & V_1 \to S_1'' \to U \to S_2'' \to V_2.
\end{align*}}%
Now following similar steps to the proof of Proposition~\ref{pro:31swsc}, we can merge $(S_1',S_2')$ and $(S_1'',S_2'')$ into $(S_1,S_2,S_3)$ as follows. When $\a' > \a''$, the channel $p(s_1'|s)$ is degraded with respect to $p(s_1''|s)$. We assume without loss of generality that $S \to S_1'' \to S_1'$ form a Markov chain. By the functional representation lemma (twice), for $p(s_1'|s_1'')$, there exists an $S_2$ independent of $S_1'$ such that $S_1'' = f(S_1',S_2)$; for $p(s_1''|s)$, there exists an $S_3$ independent of $(S_2,S_1')$ such that $S = g(S_1'',S_3) = g(f(S_1',S_2),S_3) \triangleq s(S_1',S_2,S_3)$. Renaming $S_1 \triangleq S_1'$ and plugging $S_1'' = f(S_1,S_2)$, the rate region $\Rr_1(\a',\b,\tilde{\l}_1)$, $\tilde{\l}_1 = 1,2,3$, becomes $\Rr_1(p',\l_1)$, $\l_1 = 1,2,3$, respectively. The rate region $\Rr_2(\a'',\g,\tilde{\l}_2)$, $\tilde{\l}_2 = 4,5,6$, becomes $\Rr_2(p',\l_2)$, $\l_2 = 10,11,12$, respectively. Thus, we have
\[
 \cup_{\tilde{\l}_1=1}^3 \cup_{\tilde{\l}_2=4}^6 [\Rr_1(\a',\b,\tilde{\l}_1) \cap \Rr_2(\a'',\g,\tilde{\l}_2)]  = \cup_{\l_1=1}^3 \cup_{\l_2=10}^{12} [\Rr_1(p',\l_1) \cap \Rr_2(p',\l_2)].
\]
When $\a' \le \a''$, we can assume that $S \to S_1' \to S_1''$ form a Markov chain. Following similar steps, the rate region $\Rr_1(\a',\b,\tilde{\l}_1)$, $\tilde{\l}_1 = 1,2,3$, becomes $\Rr_1(p',\l_1)$, $\l_1 = 4,5,6$, respectively. The rate region $\Rr_2(\a'',\g,\tilde{\l}_2)$, $\tilde{\l}_2 = 4,5,6$, becomes $\Rr_2(p',\l_2)$, $\l_2 = 7,8,9$, respectively. Thus, we have
\[
 \cup_{\tilde{\l}_1=1}^3 \cup_{\tilde{\l}_2=4}^6 [\Rr_1(\a',\b,\tilde{\l}_1) \cap \Rr_2(\a'',\g,\tilde{\l}_2)] = \cup_{\l_1=4}^6 \cup_{\l_2=7}^{9} [\Rr_1(p',\l_1) \cap \Rr_2(p',\l_2)].
\]
\end{proof}

\bibliographystyle{IEEEtran}
\bibliography{nit}

\newcommand{\noopsort}[1]{}
\begin{thebibliography}{10}
\providecommand{\url}[1]{#1}
\csname url@samestyle\endcsname
\providecommand{\newblock}{\relax}
\providecommand{\bibinfo}[2]{#2}
\providecommand{\BIBentrySTDinterwordspacing}{\spaceskip=0pt\relax}
\providecommand{\BIBentryALTinterwordstretchfactor}{4}
\providecommand{\BIBentryALTinterwordspacing}{\spaceskip=\fontdimen2\font plus
\BIBentryALTinterwordstretchfactor\fontdimen3\font minus
  \fontdimen4\font\relax}
\providecommand{\BIBforeignlanguage}[2]{{%
\expandafter\ifx\csname l@#1\endcsname\relax
\typeout{** WARNING: IEEEtran.bst: No hyphenation pattern has been}%
\typeout{** loaded for the language `#1'. Using the pattern for}%
\typeout{** the default language instead.}%
\else
\language=\csname l@#1\endcsname
\fi
#2}}
\providecommand{\BIBdecl}{\relax}
\BIBdecl

\bibitem{Boudreau--Panicker--Guo--Chang--Wang--Vrzic2009}
G.~Boudreau, J.~Panicker, N.~Guo, R.~Chang, N.~Wang, and S.~Vrzic,
  ``Interference coordination and cancellation for 4g networks,'' \emph{IEEE
  Communications Magazine}, vol.~47, no.~4, pp. 74--81, April 2009.

\bibitem{Seol--Cheun2009}
C.~Seol and K.~Cheun, ``A statistical inter-cell interference model for
  downlink cellular ofdma networks under log-normal shadowing and multipath
  rayleigh fading,'' \emph{IEEE Transactions on Communications}, vol.~57,
  no.~10, pp. 3069--3077, October 2009.

\bibitem{Cadambe--Jafar2008}
V.~Cadambe and S.~A. Jafar, ``Interference alignment and degrees of freedom of
  the {$K$}-user interference channel,'' \emph{{IEEE} Trans. Inf. Theory},
  vol.~54, no.~8, pp. 3425--3441, Aug. 2008.

\bibitem{El-Gamal--Kim2011}
A.~El~Gamal and Y.-H. Kim, \emph{Network Information Theory}.\hskip 1em plus
  0.5em minus 0.4em\relax Cambridge: Cambridge University Press, 2011.

\bibitem{Bidokhti--Prabhakaran2014}
S.~S. Bidokhti and V.~M. Prabhakaran, ``Is non-unique decoding necessary?''
  \emph{{IEEE} Trans. Inf. Theory}, vol.~60, no.~5, pp. 2594--2610, May 2014.

\bibitem{Costa--El-Gamal1987}
M.~H.~M. Costa and A.~El~Gamal, ``The capacity region of the discrete
  memoryless interference channel with strong interference,'' \emph{{IEEE}
  Trans. Inf. Theory}, vol.~33, no.~5, pp. 710--711, 1987.

\bibitem{Sato1978b}
H.~Sato, ``\noopsort{b}{O}n the capacity region of a discrete two-user channel
  for strong interference,'' \emph{{IEEE} Trans. Inf. Theory}, vol.~24, no.~3,
  pp. 377--379, May 1978.

\bibitem{Shang--Kramer--Chen2009}
X.~Shang, G.~Kramer, and B.~Chen, ``A new outer bound and the
  noisy-interference sum-rate capacity for {G}aussian interference channels,''
  \emph{{IEEE} Trans. Inf. Theory}, vol.~55, no.~2, pp. 689--699, Feb. 2009.

\bibitem{Annapureddy--Veeravalli2009}
V.~S. Annapureddy and V.~V. Veeravalli, ``{G}aussian interference networks:
  {S}um capacity in the low interference regime and new outer bounds on the
  capacity region,'' \emph{{IEEE} Trans. Inf. Theory}, vol.~55, no.~7, pp.
  3032--3050, Jul. 2009.

\bibitem{Motahari--Khandani2011}
A.~S. Motahari and A.~K. Khandani, ``To decode the interference or to consider
  it as noise,'' \emph{{IEEE} Trans. Inf. Theory}, vol.~57, no.~3, pp.
  1274--1283, March 2011.

\bibitem{Liu--Nair--Xia2014}
S.~Liu, C.~Nair, and L.~Xia, ``Interference channels with very weak
  interference,'' in \emph{Proc. {IEEE} Int. Symp. Inf. Theory}, June 2014, pp.
  1031--1035.

\bibitem{Baccelli--El-Gamal--Tse2011}
F.~Baccelli, A.~El~Gamal, and D.~N.~C. Tse, ``Interference networks with
  point-to-point codes,'' \emph{{IEEE} Trans. Inf. Theory}, vol.~57, no.~5, pp.
  2582--2596, May 2011.

\bibitem{Bandemer--El-Gamal--Kim2015}
B.~Bandemer, A.~E. Gamal, and Y.-H. Kim, ``Optimal achievable rates for
  interference networks with random codes,'' \emph{{IEEE} Trans. Inf. Theory},
  vol.~61, no.~12, pp. 6536--6549, Dec 2015.

\bibitem{Han--Kobayashi1981}
T.~S. Han and K.~Kobayashi, ``A new achievable rate region for the interference
  channel,'' \emph{{IEEE} Trans. Inf. Theory}, vol.~27, no.~1, pp. 49--60,
  1981.

\bibitem{Yedla--Nguyen--Pfister--Narayanan2011}
A.~Yedla, P.~S. Nguyen, H.~D. Pfister, and K.~R. Narayanan, ``Universal codes
  for the {G}aussian {MAC} via spatial coupling,'' in \emph{Proc. 49th Ann.
  Allerton Conf. Comm. Control Comput.}, Sept 2011, pp. 1801--1808.

\bibitem{Wang--Sasoglu2014}
\BIBentryALTinterwordspacing
L.~Wang and E.~{\c Sa\c so\u glu}, ``Polar coding for interference networks,''
  2014. [Online]. Available: \url{http://arxiv.org/abs/1401.7293}
\BIBentrySTDinterwordspacing

\bibitem{Rimoldi--Urbanke1996}
B.~Rimoldi and R.~Urbanke, ``A rate-splitting approach to the {Gaussian}
  multiple-access channel,'' \emph{{IEEE} Trans. Inf. Theory}, vol.~42, no.~2,
  pp. 364--375, Mar 1996.

\bibitem{Grant--Rimoldi--Urbanke--Whiting2001}
A.~J. Grant, B.~Rimoldi, R.~Urbanke, and P.~A. Whiting, ``Rate-splitting
  multiple access for discrete memoryless channels,'' \emph{{IEEE} Trans. Inf.
  Theory}, vol.~47, no.~3, pp. 873--890, 2001.

\bibitem{Zhao--Tan--Avestimehr--Diggavi--Pottie2012}
Y.~Zhao, C.~W. Tan, A.~Avestimehr, S.~Diggavi, and G.~Pottie, ``On the maximum
  achievable sum-rate with successive decoding in interference channels,''
  \emph{{IEEE} Trans. Inf. Theory}, vol.~58, no.~6, pp. 3798--3820, Jun. 2012.

\bibitem{Wang--Sasoglu--Kim2014}
L.~Wang, E.~{\c Sa\c so\u glu}, and Y.~H. Kim, ``Sliding-window superposition
  coding for interference networks,'' in \emph{Proc. {IEEE} Int. Symp. Inf.
  Theory}, June 2014, pp. 2749--2753.

\bibitem{Imai--Hirakawa1977}
H.~Imai and S.~Hirakawa, ``A new multilevel coding method using
  error-correcting codes,'' \emph{{IEEE} Trans. Inf. Theory}, vol.~23, no.~3,
  pp. 371--377, May 1977.

\bibitem{Wachsmann--Fischer--Huber1999}
U.~Wachsmann, R.~F.~H. Fischer, and J.~B. Huber, ``Multilevel codes:
  theoretical concepts and practical design rules,'' \emph{IEEE Transactions on
  Information Theory}, vol.~45, no.~5, pp. 1361--1391, Jul 1999.

\bibitem{Zehavi1992}
E.~Zehavi, ``8-psk trellis codes for a rayleigh channel,'' \emph{IEEE
  Transactions on Communications}, vol.~40, no.~5, pp. 873--884, May 1992.

\bibitem{Caire--Taricco--Biglieri1998}
G.~Caire, G.~Taricco, and E.~Biglieri, ``Bit-interleaved coded modulation,''
  \emph{{IEEE} Trans. Inf. Theory}, vol.~44, no.~3, pp. 927--946, May 1998.

\bibitem{Park--Kim--Wang2014}
H.~Park, Y.~H. Kim, and L.~Wang, ``Interference management via sliding-window
  superposition coding,'' in \emph{2014 IEEE Globecom Workshops (GC Wkshps)},
  Dec 2014, pp. 972--976.

\bibitem{Kim--Ahn--Kim--Park--Wang--Chen--Park2015}
K.~T. Kim, S.-K. Ahn, Y.-H. Kim, H.~Park, L.~Wang, C.-Y. Chen, and J.~Park,
  ``Adaptive sliding-window coded modulation in cellular networks,'' in
  \emph{2015 IEEE Global Communications Conference (GLOBECOM)}, Dec 2015, pp.
  1--7.

\bibitem{Kim--Ahn--Kim--Park--Chen--Kim2016}
K.~T. Kim, S.~K. Ahn, Y.~S. Kim, J.~Park, C.~Y. Chen, and Y.~H. Kim,
  ``Interference management via sliding-window coded modulation for {5G}
  cellular networks,'' \emph{IEEE Communications Magazine}, vol.~54, no.~11,
  pp. 82--89, November 2016.

\bibitem{5Gdiscussion1}
\BIBentryALTinterwordspacing
``{Vision and Schedule for 5G Radio Technologies (RWS-150039)},'' 3GPP TSG RAN
  Workshop on 5G, {Phoenix, AZ, USA}, September 2015, 19 pp. [Online].
  Available:
  \url{http://www.3gpp.org/ftp/workshop/2015-09-17\_18\_RAN\_5G/Docs/RWS-150039.zip}
\BIBentrySTDinterwordspacing

\bibitem{5Gdiscussion3}
\BIBentryALTinterwordspacing
``{Interference coordination for 5G new radio interface (R1-162185)},'' 3GPP
  TSG RAN WG1 \#84bis, April 2016, 4 pp. [Online]. Available:
  \url{http://www.3gpp.org/ftp/TSG\_RAN/WG1\_RL1/TSGR1\_84b/Docs/R1-162185.zip}
\BIBentrySTDinterwordspacing

\bibitem{5Gdiscussion4}
\BIBentryALTinterwordspacing
``{Discussion on interference management based on advanced transceivers for NR
  (R1-164023)},'' 3GPP TSG RAN WG1 \#85, May 2016, 3 pp. [Online]. Available:
  \url{http://www.3gpp.org/ftp/TSG\_RAN/WG1\_RL1/TSGR1\_85/Docs/R1-164023.zip}
\BIBentrySTDinterwordspacing

\bibitem{5Gdiscussion5}
\BIBentryALTinterwordspacing
``{Discussion on spatial multiplexing for NR (R1-167887)},'' 3GPP TSG RAN WG1
  \#86, August 2016, 2 pp. [Online]. Available:
  \url{http://www.3gpp.org/ftp/TSG\_RAN/WG1\_RL1/TSGR1\_86/Docs/R1-167887.zip}
\BIBentrySTDinterwordspacing

\bibitem{5Gdiscussion6}
\BIBentryALTinterwordspacing
``{Discussion on modulation for NR (R1-166776)},'' 3GPP TSG RAN WG1 \#86,
  August 2016, 4 pp. [Online]. Available:
  \url{http://www.3gpp.org/ftp/TSG\_RAN/WG1\_RL1/TSGR1\_86/Docs/R1-166776.zip}
\BIBentrySTDinterwordspacing

\bibitem{5Gdiscussion7}
\BIBentryALTinterwordspacing
``{Discussion on interference management based on advanced transceivers for NR
  (R1-166791)},'' 3GPP TSG RAN WG1 \#86, August 2016, 4 pp. [Online].
  Available:
  \url{http://www.3gpp.org/ftp/TSG\_RAN/WG1\_RL1/TSGR1\_86/Docs/R1-166791.zip}
\BIBentrySTDinterwordspacing

\bibitem{Orlitsky--Roche2001}
A.~Orlitsky and J.~R. Roche, ``Coding for computing,'' \emph{{IEEE} Trans. Inf.
  Theory}, vol.~47, no.~3, pp. 903--917, 2001.

\bibitem{Cover--Thomas2006}
T.~M. Cover and J.~A. Thomas, \emph{Elements of Information Theory},
  2nd~ed.\hskip 1em plus 0.5em minus 0.4em\relax New York: Wiley, 2006.

\bibitem{Wang--Sasoglu--Bandemer--Kim2013}
L.~Wang, E.~{\c Sa\c so\u glu}, B.~Bandemer, and Y.-H. Kim, ``A comparison of
  superposition coding schemes,'' in \emph{Proc. {IEEE} Int. Symp. Inf.
  Theory}, Istanbul, Turkey, Jul. 2013.

\bibitem{Fawzi--Savov2012}
\BIBentryALTinterwordspacing
O.~Fawzi and I.~Savov, ``Rate-splitting in the presence of multiple
  receivers,'' 2012. [Online]. Available: \url{http://arxiv.org/abs/1207.0543}
\BIBentrySTDinterwordspacing

\bibitem{Cover--El-Gamal1979}
T.~M. Cover and A.~El~Gamal, ``Capacity theorems for the relay channel,''
  \emph{{IEEE} Trans. Inf. Theory}, vol.~25, no.~5, pp. 572--584, Sep. 1979.

\bibitem{Cover--Leung1981}
T.~M. Cover and C.~S.~K. Leung, ``An achievable rate region for the
  multiple-access channel with feedback,'' \emph{{IEEE} Trans. Inf. Theory},
  vol.~27, no.~3, pp. 292--298, 1981.

\bibitem{Carleial1982}
A.~B. Carleial, ``Multiple-access channels with different generalized feedback
  signals,'' \emph{{IEEE} Trans. Inf. Theory}, vol.~28, no.~6, pp. 841--850,
  Nov. 1982.

\bibitem{Xie--Kumar2005}
L.-L. Xie and P.~R. Kumar, ``An achievable rate for the multiple-level relay
  channel,'' \emph{{IEEE} Trans. Inf. Theory}, vol.~51, no.~4, pp. 1348--1358,
  2005.

\bibitem{Kramer--Gastpar--Gupta2005}
G.~Kramer, M.~Gastpar, and P.~Gupta, ``Cooperative strategies and capacity
  theorems for relay networks,'' \emph{{IEEE} Trans. Inf. Theory}, vol.~51,
  no.~9, pp. 3037--3063, Sep. 2005.

\bibitem{Foschini1996}
G.~J. Foschini, ``Layered space-time architecture for wireless communication in
  a fading environment when using multi-element antennas,'' \emph{Bell Labs
  Tech. J.}, vol.~1, no.~2, pp. 41--59, 1996.

\bibitem{Li--Huang--Foschini--Valenzuela2000}
X.~Li, H.~Huang, G.~J. Foschini, and R.~A. Valenzuela, ``Effects of iterative
  detection and decoding on the performance of blast,'' in \emph{Global
  Telecommunications Conference, 2000. {GLOBECOM '00. IEEE}}, vol.~2, 2000, pp.
  1061--1066 vol.2.

\bibitem{Foschini--Chizhik--Gans--Papadias--Valenzuela2003}
G.~J. Foschini, D.~Chizhik, M.~J. Gans, C.~Papadias, and R.~A. Valenzuela,
  ``Analysis and performance of some basic space-time architectures,''
  \emph{{IEEE} Journal on Selected Areas in Communications}, vol.~21, no.~3,
  pp. 303--320, Apr 2003.

\bibitem{Wolniansky--Foschini--Golden--valenzuela1998}
P.~W. Wolniansky, G.~J. Foschini, G.~D. Golden, and R.~A. Valenzuela,
  ``V-blast: an architecture for realizing very high data rates over the
  rich-scattering wireless channel,'' in \emph{Signals, Systems, and
  Electronics, 1998. ISSSE 98. 1998 URSI International Symposium on}, Sep 1998,
  pp. 295--300.

\bibitem{LTE}
{3GPP TS 36.212}, ``Multiplexing and channel coding,'' Release 12, 2013.

\bibitem{Wang2015}
L.~Wang, ``Channel coding techniques for network communication,'' {Ph.D.}
  Thesis, University of California, San Diego, La Jolla, CA, 2015.

\end{thebibliography}

\end{document}